\documentclass[11pt]{article}
\usepackage[utf8]{inputenc}
\usepackage[english]{babel}
\usepackage[active]{srcltx}
\usepackage{hyperref}
\usepackage[all]{xy}
\usepackage{amsfonts}
\usepackage{amsthm}
\usepackage{enumerate}
\usepackage{mathrsfs}
\usepackage{multicol}
\usepackage[all]{xy}
\usepackage{amsfonts}
\usepackage{latexsym,amsmath,amssymb}
\usepackage[usenames,dvipsnames]{color}
\usepackage{pstricks}
\usepackage[absolute,overlay]{textpos}
\usepackage{graphics,wrapfig,times}
\usepackage{xfrac}

\usepackage{color}
\definecolor{DarkBlue}{rgb}{0.1,0.1,0.5}
\definecolor{Red}{rgb}{0.9,0.1,0.1}
\definecolor{Green}{rgb}{0.3,0.7,0.0}
\definecolor{green2}{rgb}{0.1,0.7,0.2}
\definecolor{blue2}{rgb}{0.0,0.6,0.7}
\definecolor{pink}{rgb}{1,0.0,1}
\definecolor{orange}{rgb}{0.9,0.0,0.1}

\newtheorem{theorem}{Theorem}
\newtheorem{corollary}{Corollary}

\newtheorem{lemma}[theorem]{Lemma}
\newtheorem{proposition}{Proposition}

\newtheorem{definition}{Definition}

\renewcommand{\d}{\mathrm{d}}

\renewcommand{\d}{\mathrm{d}}
\newcommand{\W}{\mathcal{W}}
\newcommand{\derpar}[2]{\displaystyle\frac{\partial{#1}}{\partial{#2}}}

\newcommand{\restric}[2]{\left.#1\right|_{#2}}
\newcommand{\Lag}{\mathcal{L}}
\newcommand{\Leg}{\mathcal{FL}}
\newcommand{\vf}{\mathfrak{X}}
\newcommand{\df}{\Omega}

\newcommand{\Tan}{\mathrm{T}}
\newcommand{\inn}{{\mathop{i}\nolimits}}
\newcommand{\Lie}{\mathop{\mathrm{L}}\nolimits}

\newcommand{\bal}{\begin{align*}}
\newcommand{\eal}{\end{align*}}
\def\beq{\begin{equation}}
\def\eeq{\end{equation}}
\def\bea{\begin{eqnarray}}
\def\eea{\end{eqnarray}}
\def\beann{\begin{eqnarray*}}
\def\eeann{\end{eqnarray*}}
\def\ben{\begin{enumerate}}
\def\een{\end{enumerate}}
\def\bit{\begin{itemize}}
\def\eit{\end{itemize}}

\newtheorem{definicio}{Definition}[section]
\newtheorem{resultat}{Result}[section]
\newtheorem{assum}{Assumption}[section]
\newtheorem{prop}{Proposition}[section]
\newtheorem{lema}{Lemma}[section]
\newtheorem{theo}{Theorem}[section]
\newcommand{\bd}{\begin{definicio} } 
\newcommand{\ed}{\end{definicio} } 
\newcommand{\bt}{\begin{theo} } 
\newcommand{\et}{\end{theo} } 
\newcommand{\bi}{\begin{itemize} } 
\newcommand{\ei}{\end{itemize} } 
\newcommand{\be}{\begin{enumerate} } 
\newcommand{\ee}{\end{enumerate} } 
\newcommand{\br}{\begin{resultat} } 
\newcommand{\er}{\end{resultat} } 
\newcommand{\ba}{\begin{assum} } 
\newcommand{\ea}{\end{assum} } 
\newcommand{\bl}{\begin{lema}}
\newcommand{\el}{\end{lema}}
\newcommand{\bp}{\begin{prop}}

\def\vf{\mathfrak X}
\def\df{{\mit\Omega}}
\def\Lag{{\cal L}}

\def\d{{\rm d}}

\def\Tan{{\rm T}}
\def\Lie{\mathop{\rm L}\nolimits}
\def\inn{\mathop{i}\nolimits}
\def\Cinfty{{\rm C}^\infty}

\def\tabaddress#1{{\small\it\begin{tabular}[t]{c}#1
\\[1.2ex]\end{tabular}}}

\title{\sc Multisymplectic unified formalism for Einstein-Hilbert Gravity}
\author{
{\sc  Jordi Gaset\thanks{{\bf e}-{\it mail}:
   jordi.gaset@upc.edu} },
   {\sc Narciso Rom\'an-Roy\thanks{{\bf e}-{\it mail}:
   narciso.roman@upc.edu  / ORCID: 0000-0003-3663-9861.}}  \\
   \tabaddress{Department of Mathematics.
   Ed. C-3, Campus Norte UPC\\
   C/ Jordi Girona 1. 08034 Barcelona. Spain.}}
   \date{\today \\
   }

\begin{document}

\maketitle

\pagestyle{myheadings}
\markright{\rm J. Gaset and N. Rom\'an-Roy,
   \sl Multisymplectic unified formalism for Einstein-Hilbert Gravity.}
\maketitle
\thispagestyle{empty}

\begin{abstract}
We present a covariant multisymplectic formulation for the Einstein-Hilbert model of General Relativity.
As it is described by a second-order singular Lagrangian,
this is a gauge field theory with constraints.
The use of the unified Lagrangian-Hamiltonian formalism
is particularly interesting when it is applied to these kinds of theories, 
since it simplifies the treatment of them; in particular, 
the implementation of the constraint algorithm, 
the retrieval of the Lagrangian description, 
and the construction of the covariant Hamiltonian formalism. 
In order to apply this algorithm to the covariant field equations, 
they must be written in a suitable geometrical way, 
which consists of using integrable distributions,
represented by multivector fields of a certain type.
We apply all these tools to the Einstein-Hilbert model without and with energy-matter sources.
We obtain and explain the geometrical and physical meaning 
of the Lagrangian constraints and we construct the multimomentum
(covariant) Hamiltonian formalisms in both cases.
As a consequence of the gauge freedom and the constraint algorithm,
we see how this model is equivalent to a first-order regular theory,
without gauge freedom. In the case of presence of energy-matter sources,
we show how some relevant geometrical and physical
characteristics of the theory depend on the type of source.
In all the cases, we obtain explicitly multivector fields which are
solutions to the gravitational field equations.
Finally, a brief study of symmetries and conservation laws is done in this context.
\end{abstract}

 \bigskip
\noindent {\bf Key words}:
 \textsl{$2$nd-order classical field theories, Higher-order jet bundles, Multisymplectic forms, Hilbert-Einstein action, Einstein equations.}

\vbox{\raggedleft AMS s.\,c.\,(2010): \null 
{\it Primary}: 53D42, 55R10, 70S05, 83C05; {\it Secondary}: 49S05, 53C15, 53C80, 53Z05.}\null

\newpage

\tableofcontents

\newpage

\section{Introduction}
\label{intro}

The geometrisation of the theories of gravitation (General Relativity) and, 
in particular the {\sl multisymplectic framework}, allows us to do a
covariant description of these theories, considering and
understanding several inherent characteristics of it,
and it has been studied by different authors.
For instance, relevant references devoted to develop geometrically 
general aspects of the theory are 
\cite{art:Capriotti,CVB-2006,CreTe-2016,CreTe-2016b,ESG-1995,GIMMSY-mm,IS-2016,rovelli},
the reduction of the order of the theory and the projectability 
of the Poincar\'e-Cartan form associated with  the Hilbert-Einstein action
is explained in \cite{first,rosado,rosado2},
meanwhile in \cite{Krupka,KrupkaStepanova} 
different aspects of the theory are studied using Lepage-Cartan forms,
and in \cite{vey1,vey2} a multisymplectic analysis of the vielbein formalism 
of General Relativity is done.
Finally, some general features of the gravitational theory 
following the polysymplectic version of the multisymplectic formalism
are described in  \cite{GMS-97,Sd-95},
including the problem of its precanonical quantization \cite{Ka-q1,Ka-q2,Ka-2016}.

This paper is a contribution to the study of  the most classical variational model 
for General Relativity that is, the {\sl Einstein-Hilbert theory} 
(with and without energy-matter sources),
using the {\sl multisymplectic framework} for giving a covariant description of it. 
As it is well-known, this model is described by a 
second-order singular Lagrangian, and thus this study presents
General Relativity as a higher-order premultisymplectic field theory  with constraints.
Our study is done from a different perspective since
we use the unified Lagrangian-Hamiltonian formalism developed for first and
second-order multisymplectic field theories 
\cite{art:Echeverria_Lopez_Marin_Munoz_Roman04,pere}
(which was stated first by R. Skinner and R. Rusk for 
autonomous mechanical systems \cite{art:Skinner_Rusk83}),
and is specially interesting for analyzing non-regular constraint theories.
Then we derive from it the Lagrangian and  multimomentum Hamiltonian formalism.

As a consequence of the singularity of the Lagrangian, 
the Einstein-Hilbert model exhibits gauge freedom and it can be reduced
to a first-order field theory
\cite{first,Krupka,KrupkaStepanova,rosado,rosado2}.
Then, related to this topic, we analyse also a first-order theory equivalent 
to Einstein-Hilbert (without matter-energy sources), 
which is described by a first-order regular Lagrangian,
showing, in this way, that General Relativity can be realised as
a regular multisymplectic field theory (without constraints).
This first-order model is different from the {\sl affine-metric} or 
{\sl Einstein-Palatini } approach
which is also a first-order but non-regular (gauge) theory.
The gauge freedom of the Hilbert-Einstein theory is also discussed,
in order to show clearly the relation with the first-order case.
In the case of the Einstein-Hilbert model with energy-matter sources,
we show how the behaviour of the theory (the constraints arising in the
constraint algorithm and the achievement of the multimomentum Hamiltonian
formalism) depends on the characteristics 
of the Lagrangian representing the sources.
This study is done in detail for the most standard types of 
energy-matter sources: those coupled to the metric.

In our analysis, the field equations are stated on their standard form (for sections),
but also in a more geometrical way, as equations for distributions 
(using multivector fields), which is the most suitable way in order to obtain
the set of points where these equations have consistent solutions.
Thus, this allows us to apply the extension of
the geometrical constraint algorithms for singular dynamical systems
to (higher-order) singular field theories 
(see \cite{art:deLeon_Marin_Marrero_Munoz_Roman05,pere}).
In all the cases, we obtain explicitly multivector fields 
(i.e., distributions) which are solutions to the gravitational field equations.
As we will see, this constraint algorithm plays a crucial role in obtaining
the main features of the theory.

Another interesting aspect of the theory is the study of its symmetries.
An analysis of this subject is also given,
stating the basic definitions and properties of symmetries and conservation laws
in the Lagrangian formalism (including the corresponding version of Noether's theorem)
\cite{EMR-96,GPR-2017},
and extending these concepts and properties to the unified formalism.
The application of these results to the Einstein-Hilbert model is briefly analyzed,
recovering some previous results \cite{Krupka,rosado}.

The organization of the paper is the following:
In Section \ref{Sec2} we present the Hilbert-Einstein Lagrangian
(without energy-matter sources)
in the ambient of higher-order jet bundles 
and their associated multimomentum bundles. We develop the
Lagrangian-Hamiltonian formalism for the theory and we apply the constraint algorithm,
obtaining the final constraint submanifold where the field equations have
consistent solutions. Then, we recover both the Lagrangian
and Hamiltonian formalisms, using in the last case, 
different charts of coordinates which
show how this second-order theory can be equivalent to a first-order one.
This equivalent first-order model is studied in Section \ref{compare}, using
again the unified formalism, and recovering then the 
Lagrangian and Hamiltonian settings from it.
Section \ref{Sec4} is devoted to analyse the Hilbert-Einstein Lagrangian
with energy-matter sources, following the same procedure as in Section \ref{Sec2},
and comparing the features of both cases.
In Section \ref{symmetries} we discuss the basic concepts
about symmetries in the Lagrangian and the unified formalisms and for the Einstein-Hilbert model in particular.
Finally, in the appendices, the calculation of multivector fields
which are solutions to the field equations for all these models 
are explicitly done; as well as a brief review on 
multivector fields and distributions.

All the manifolds are real, second countable and $\Cinfty$. The maps and the structures are $\Cinfty$.  Sum over repeated indices is understood.

\section{The Einstein-Hilbert model (without energy-matter sources)}
\label{Sec2}

\subsection{Previous statements. The Hilbert-Einstein Lagrangian}
\label{presta}

Fist we consider the Hilbert Lagrangian for the Einstein equations of gravity
without sources (no matter-energy is present).

The configuration bundle for this system is a fiber bundle
$\pi\colon E\rightarrow M$, 
where $M$ is a connected orientable 4-dimensional manifold representing space-time, 
whose volume form is denoted $\eta \in \df^4(M)$.
$E$ is the manifold of Lorentz metrics on $M$;
that is, for every $x\in M$, the fiber $\pi^{-1}(x)$
is the set of metrics with signature $(-+++)$ acting on $T_xM$.

The adapted fiber coordinates in $E$ are $(x^\mu,g_{\alpha\beta})$,
($\mu,\alpha,\beta=0,1,2,3$), such that 
$\eta=\d x^0\wedge\ldots\wedge\d x^3\equiv\d^4x$ and
where $g_{\alpha\beta}$ are the component functions of the metric. 
It is usefull to consider also the components $g_{\beta\alpha}$ with $\beta>\alpha$,
since we should remember they are not independent 
because the metric is symmetric,  $g_{\alpha\beta}=g_{\beta\alpha}$. 
Actually there are $10$ independent variables, 
resulting that the dimension of the fibers is $10$ and $\dim E=14$. 
When we sum over the indices on the fiber and not all the components, 
we order the indices as $0\leq\alpha\leq\beta\leq3$.

In order to state the formalism we need to consider the 
\textsl{$k$th-order jet bundles} of the projection $\pi$, $J^k\pi$, ($k=1,2,3$);
which are the manifolds of the $k$-jets of local sections 
$\phi \in \Gamma(\pi)$;
that is, equivalence classes of local sections of $\pi$
(see \cite{book:Saunders89} for details).
Points in $J^k\pi$ are denoted by $j^k_x\phi$, 
with $x \in M$ and $\phi \in \Gamma(\pi)$ being a representative
of the equivalence class.
If $\phi \in \Gamma(\pi)$, we denote the \textsl{$k$th prolongation} of $\phi$ to $J^k\pi$ by
$j^k\phi \in \Gamma(\bar{\pi}^k)$.
We have the following natural projections: if $r \leqslant k$,
$$
\begin{array}{rcl}
\pi^k_r \colon J^k\pi & \longrightarrow & J^r\pi \\
j^k_x\phi & \longmapsto & j^r_x\phi
\end{array}
\quad ; \quad
\begin{array}{rcl}
\pi^k \colon J^k\pi & \longrightarrow & E \\
j^k_x\phi & \longmapsto & \phi(x)
\end{array}
\quad ; \quad
\begin{array}{rcl}
\bar{\pi}^k \colon J^k\pi & \longrightarrow & M \\
j^k_x\phi & \longmapsto & x
\end{array} \ .
$$
Observe that $\pi^s_r\circ\pi^k_s=\pi^k_r$, $\pi^k_0 = \pi^k$, 
$\pi^k_k=\textnormal{Id}_{J^k\pi}$, and $\bar{\pi}^k = \pi \circ \pi^k$.
The induced coordinates in $J^3\pi$ are 
$(x^\mu,\,g_{\alpha\beta},\,g_{\alpha\beta,\mu},\,g_{\alpha\beta,\mu\nu},
\,g_{\alpha\beta,\mu\nu\lambda})$, ($0\leq\mu\leq\nu\leq\lambda\leq3$). Again, we will use all the permutations, although only the ordered ones are proper coordinates.

A special kind of vector fields are the {\sl coordinate total derivatives} 
\cite{pere,book:Saunders89},
which are locally given as
\beq
\displaystyle
D_\tau=\derpar{}{x^\tau}
+\sum_{\substack{\alpha\leq\beta\\\mu\leq\nu\leq\lambda}} 
\left(g_{\alpha\beta,\tau}\derpar{}{g_{\alpha\beta}}+ g_{\alpha\beta,\mu\tau}\derpar{}{g_{\alpha\beta,\mu}}+
 g_{\alpha\beta,\mu\nu\tau}\derpar{}{g_{\alpha\beta,\mu\nu}}
+ g_{\alpha\beta,\mu\nu\lambda\tau}\derpar{}{g_{\alpha\beta,\mu\nu\lambda}}\right).
\label{Di}
\eeq
Observe that, if $f\in\Cinfty(J^k\pi)$, then $D_\tau f\in\Cinfty(J^{k+1}\pi)$.

Next, consider the bundle $J^1\pi$ and let ${\cal M}\pi\equiv\Lambda_2^4(J^{1}\pi)$ 
be the bundle of $4$-forms over $J^{1}\pi$ vanishing 
by the action of two $\bar{\pi}^{1}$-vertical vector fields; 
with the canonical projections
$$
\pi_{J^{1}\pi} \colon \Lambda_2^4(J^{1}\pi) \to J^{1}\pi \quad ; \quad
\bar{\pi}_M= \bar{\pi}^{1} \circ \pi_{J^{1}\pi} \colon \Lambda_2^4(J^{1}\pi) \to M \, .
$$
Induced local coordinates in $\Lambda_2^4(J^{1}\pi)$ are
$(x^\mu,g_{\alpha\beta},g_{\alpha\beta,\mu},p,p^{\alpha\beta,\mu},\overline{p}^{\alpha\beta,\mu\nu})$, with $0\leq\alpha\leq\beta\leq3$ and $\mu,\nu=0,\dots,3$.
This bundle is endowed with the
\textsl{tautological (or Liouville) $4$-form} 
$\Theta_{1}\in\df^4(\Lambda_2^4(J^{1}\pi))$
and the \textsl{canonical (or Liouville) $5$-form} 
$\Omega_{1} = -d\Theta_{1} \in \Omega^5(\Lambda_2^4(J^{1}\pi))$
(it is a multisymplectic form; that is, it is closed and $1$-nondegenerate),
whose local expressions are
\beann
\Theta_1&=& 
p\,\d^4x+\sum_{\alpha\leq\beta}\left(p^{\alpha\beta,\mu}\d g_{\alpha\beta}\wedge \d^{3}x_\mu+
\overline{p}^{\alpha\beta,\mu\nu}\d g_{\alpha\beta,\mu}\wedge \d^{3}x_{\nu}\right)
 \ , \\
\Omega_1&=&
-\d p\wedge \d^4x-\sum_{\alpha\leq\beta}\left(\d p^{\alpha\beta,\mu}\wedge \d g_{\alpha\beta}\wedge \d^{3}x_\mu+
\d \overline{p}^{\alpha\beta,\mu\nu}\wedge \d g_{\alpha\beta,\mu}\wedge \d^{3}x_{\nu}\right) \ ;
\eeann
where $\displaystyle \d^3x_\nu=\inn\left(\derpar{}{x^\nu}\right)\d^4x$.
Now, consider the $\pi_{J^{1}\pi}$-transverse submanifold
$\jmath_s \colon J^2\pi^\dagger \hookrightarrow \Lambda_2^4(J^1\pi)$ 
defined locally by the constraints 
$\overline{p}^{\alpha\beta,\mu\nu}=\overline{p}^{\alpha\beta,\nu\mu}$,
which is called the \textsl{extended $2$-symmetric multimomentum bundle} 
(although it is defined using coordinates,
this construction is canonical \cite{art:Saunders_Crampin90}).
Let 
$$
\pi_{J^1\pi}^\dagger \colon J^2\pi^\dagger \to J^1\pi
\quad ; \quad
\bar{\pi}_M^\dagger = \bar{\pi}^{1} \circ \pi_{J^1\pi}^\dagger \colon J^2\pi^\dagger \to M
$$
be the canonical projections. Natural coordinates in $J^2\pi^\dagger$ are 
$(x^\mu,g_{\alpha\beta},g_{\alpha\beta,\mu},p,p^{\alpha\beta,\mu},p^{\alpha\beta,\mu\nu})$
($0\leq\alpha\leq\beta\leq 3$; $0\leq\mu\leq\nu\leq3$), with 
$\displaystyle\jmath_s^*\overline{p}^{\alpha\beta,\mu\nu}=\frac{1}{n(\mu\nu)}p^{\alpha\beta,\mu\nu}$ (where $n(\mu\nu)$ is a combinatorial factor which 
$n(\mu\nu)=1$ for $\mu=\nu$, and $n(\mu\nu)=2$ for $\mu\neq\nu$).
Denote $\Theta_1^s =\jmath_s^*\Theta_1 \in \Omega^4(J^2\pi^\dagger)$
and the multisymplectic form 
$\Omega_1^s= \jmath_s^*\Omega_1 = -d\Theta_1^s \in \Omega^5(J^2\pi^\dagger)$, 
which are called \textsl{symmetrized Liouville $m$ and $(m+1)$-forms},
and their coordinate expressions are
\beann
\Theta_1^s&=& p\,\d^4x+\sum_{\alpha\leq\beta}p^{\alpha\beta,\mu}\d g_{\alpha\beta}\wedge \d^{3}x_\mu+\sum_{\alpha\leq\beta}\frac{1}{n(\mu\nu)}p^{\alpha\beta,\mu\nu}\d g_{\alpha\beta,\mu}\wedge \d^{3}x_{\nu} \ , \\
\Omega_1^s&=&-\d p\wedge \d^4x-\sum_{\alpha\leq\beta}\d p^{\alpha\beta,\mu}\wedge \d g_{\alpha\beta}\wedge \d^{3}x_\mu-\sum_{\alpha\leq\beta}\frac{1}{n(\mu\nu)}\d p^{\alpha\beta,\mu\nu}\wedge \d g_{\alpha\beta,\mu}\wedge \d^{3}x_{\nu} \ .
\eeann

Finally, consider the quotient bundle $J^2\pi^\ddagger = J^2\pi^\dagger / \Lambda^4_1(J^1\pi)$,
which is called the \textsl{restricted $2$-symmetric multimomentum bundle}, 
and it is endowed with the natural projections
$$
\mu \colon J^2\pi^\dagger \to J^2\pi^\ddagger\quad ; \quad
\pi_{J^1\pi}^\ddagger \colon J^2\pi^\ddagger \to J^1\pi\quad ; \quad
\bar{\pi}_M^\ddagger \colon J^2\pi^\ddagger \to M .
$$
Observe that $J^2\pi^\ddagger$ is also the submanifold of
$\Lambda^4_2(J^1\pi) / \Lambda^4_1(J^1\pi)$ 
defined by the local constraints $\overline{p}^{\alpha\beta,\mu\nu} 
- \overline{p}^{\alpha\beta,\nu\mu} = 0$. 
Hence, natural coordinates in $J^2\pi^\ddagger$ are 
$(x^\mu,g_{\alpha\beta},g_{\alpha\beta,\mu},p^{\alpha\beta,\mu},p^{\alpha\beta,\mu\nu})$,
($0\leq\alpha\leq\beta\leq 3$; $0\leq\mu\leq\nu\leq3$).
Obviously, $\dim J^{2}\pi^\ddagger = \dim J^{2}\pi^\dagger - 1$.

The {\sl Hilbert-Einstein Lagrangian density} (in vacuum) is a 
$\overline{\pi}^2$-semibasic form $\mathcal{L}_{\mathfrak V}\in\Omega^4(J^2\pi)$, then
$\mathcal{L}_{\mathfrak V}=L_{\mathfrak V}\,(\overline{\pi}^2)^*\eta$, 
where $L_{\mathfrak V}\in\Cinfty(J^2\pi)$ is the {\sl Hilbert-Einstein Lagrangian function} (in vacuum)
given by
$$
L_{\mathfrak V}=\varrho R=\varrho g^{\alpha\beta}R_{\alpha\beta}\ ;
$$
here $\varrho=\sqrt{|det(g_{\alpha\beta})|}$, $R$ is the {\sl scalar curvature},
$R_{\alpha\beta}=D_\gamma\Gamma^{\gamma}_{\alpha\beta}-D_\alpha\Gamma^{\gamma}_{\gamma\beta}+
\Gamma^{\gamma}_{\alpha\beta}\Gamma^{\delta}_{\delta\gamma}-
\Gamma^{\gamma}_{\delta\beta}\Gamma^{\delta}_{\alpha\gamma}$
are the components of the {\sl Ricci tensor},
$\displaystyle\Gamma^{\lambda}_{\mu\nu}=
\frac{1}{2}\,g^{\lambda\rho}\left( g_{\nu\rho,\mu}+ 
 g_{\rho\mu,\nu}- g_{\mu\nu,\rho}\right)$ are the {\sl Christoffel symbols of
the Levi-Civita connection} of $g$, and
$g^{\alpha\beta}$ denotes the inverse matrix of $g$, 
namely: $g^{\alpha\beta}g_{\beta\gamma}=\delta^\alpha_\gamma$.
As the Christoffel symbols depend on first-order derivatives of $g_{\mu\nu}$ and
taking into account the expression \eqref{Di} we have that
the Lagrangian contains second-order derivatives 
of the components of the metric and thus this is a second-order field theory. 



\subsection{Lagrangian-Hamiltonian unified formalism}
\label{lhuni}

\subsubsection{The higher-order jet multimomentum bundles}

For the Lagrangian-Hamiltonian unified formalism, we have to consider the
\textsl{symmetric higher-order jet multimomentum bundles} 
$\mathcal{W}=J^3\pi\times_{J^1\pi}J^2\pi^\dagger$  and
$\W_r = J^{3}\pi \times_{J^{1}\pi} \, J^{2}\pi^\ddagger$
(see \cite{pere, pere2} for details), which have as natural local coordinates
$(x^\mu,\,g_{\alpha\beta},\,g_{\alpha\beta,\mu},\,g_{\alpha\beta,\mu\nu},
\,g_{\alpha\beta,\mu\nu\lambda},p,p^{\alpha\beta,\mu},p^{\alpha\beta,\mu\nu})$
and $(x^\mu,\,g_{\alpha\beta},\,g_{\alpha\beta,\mu},\,g_{\alpha\beta,\mu\nu},
\,g_{\alpha\beta,\mu\nu\lambda},p^{\alpha\beta,\mu},p^{\alpha\beta,\mu\nu})$,
($0\leq\alpha\leq\beta\leq 3$; $0\leq\mu\leq\nu\leq3$).
These bundles are endowed with the canonical projections
\beann
\rho_1 \colon \W \to J^{3}\pi \ ,\
&\rho_2 \colon \W \to J^{2}\pi^\dagger \ ,\ &
\rho_M \colon \W \to M
\\
\rho^r_1 \colon \W_r \to J^{3}\pi \ ,\
&\rho^r_2 \colon \W_r \to J^{2}\pi^\ddagger  \ ,\ &
\rho_M^r \colon \W_r \to M \ .
\eeann
Furthermore, the quotient map
$\mu \colon J^{2}\pi^\dagger \to J^{2}\pi^\ddagger$
induces a natural submersion $\mu_\W \colon \W \to \W_r$.

We can define the canonical pairing 
$$
\begin{array}{rcl}
\mathcal{C} \colon J^{2}\pi \times_{J^{1}\pi} \Lambda_2^4(J^{1}\pi) & \longrightarrow & \Lambda_1^4(J^{1}\pi) \\
(j^{2}_x\phi,\omega) & \longmapsto & (j^{1}\phi)^*_{j^{1}_x\phi}\omega
\end{array} \ ,
$$
and, 
from here, we have a new pairing $\mathcal{C}^s \colon J^2\pi \times_{J^1\pi} J^2\pi^\dagger \to \Lambda_1^4(J^1\pi)$ defined as
$$
\mathcal{C}^s(j^{2}_x\phi,\omega) = \mathcal{C}(j^{2}_x\phi,j_s(\omega)) = (j^1\phi)_{j^1_x\phi}^* \ j_s(\omega) \, .
$$
Therefore, the \textsl{second-order coupling $4$-form} in $\mathcal{W}$ is the $\rho_M$-semibasic $4$-form
$\hat{\mathcal{C}} \in \Omega^4(\mathcal{W})$ defined by
$$
\hat{\mathcal{C}}(j^3_x\phi,\omega) = \mathcal{C}^s(\pi^3_2(j^3_x\phi),\omega) 
\quad ,\quad  (j^3_x\phi,\omega) \in \mathcal{W} \ .
$$
As $\hat{\mathcal{C}}$ is a $\rho_M$-semibasic $4$-form, 
there exists a function $\hat{C} \in C^\infty(\mathcal{W})$ such that 
$\hat{\mathcal{C}} = \hat{C}\rho_M^*\eta$, 
and we have the coordinate expression
$$
\hat{\mathcal{C}}= 
\left(p+\sum_{\alpha\leq\beta}p^{\alpha\beta,\mu}g_{\alpha\beta,\mu}
+\sum_{\substack{\alpha\leq\beta\\\mu\leq\nu}}p^{\alpha\beta,\mu\nu}g_{\alpha\beta,\mu\nu}\right)\d^4x \ .
$$

Denoting by 
$\hat{\mathcal{L}}=(\pi^3_2 \circ \rho_1)^*\mathcal{L}_{\mathfrak V}\in \Omega^4(\mathcal{W})$,
we can write $\hat{\mathcal{L}} = \hat{L} \,\rho_M^*\eta$,
where $\hat{L} = (\pi^3_2 \circ \rho_1)^*L_{\mathfrak V}\in C^\infty(\mathcal{W})$.
Then, we introduce the {\sl Hamiltonian submanifold}
$$
\mathcal{W}_o = \left\{ w \in \mathcal{W} \colon \hat{\mathcal{L}}(w) = \hat{\mathcal{C}}(w) \right\} \stackrel{\jmath_o}{\hookrightarrow} \mathcal{W} \ ,
$$
which is defined by the constraint
$$
\hat{\mathcal{C}}-\hat{L}\equiv p+\sum_{\alpha\leq\beta}p^{\alpha\beta,\mu}g_{\alpha\beta,\mu}
+\sum_{\substack{\alpha\leq\beta\\\mu\leq\nu}}p^{\alpha\beta,\mu\nu}g_{\alpha\beta,\mu\nu}-\hat{L} = 0 \ .
$$
and it is $\mu_\mathcal{W}$-transverse and diffeomorphic to $\mathcal{W}_r$, $\Phi_o\colon\W_o\to\W_r$.
Then, the submanifold $\mathcal{W}_o$ induces a \textsl{Hamiltonian section}
$\hat{h} \in \Gamma(\mu_\mathcal{W})$ defined as
$\hat{h} = \jmath_o \circ \Phi_o^{-1} \colon \mathcal{W}_r \to \mathcal{W}$, 
which is specified giving the local \textsl{Hamiltonian function}
$$
\hat H=\sum_{\alpha\leq\beta}p^{\alpha\beta,\mu}g_{\alpha\beta,\mu}+\sum_{\substack{\alpha\leq\beta\\\mu\leq\nu}}p^{\alpha\beta,\mu\nu}g_{\alpha\beta,\mu\nu}-\hat L \ .
$$
that is, 
\beann
\hat{h}
(x^\mu,\,g_{\alpha\beta},\,g_{\alpha\beta,\mu},\,g_{\alpha\beta,\mu\nu},
\,g_{\alpha\beta,\mu\nu\lambda},p^{\alpha\beta,\mu},p^{\alpha\beta,\mu\nu}) 
=\\
(x^\mu,\,g_{\alpha\beta},\,g_{\alpha\beta,\mu},\,g_{\alpha\beta,\mu\nu},
\,g_{\alpha\beta,\mu\nu\lambda},-\hat{H},p^{\alpha\beta,\mu},p^{\alpha\beta,\mu\nu}).
\eeann
Hence, we have the diagram:
$$
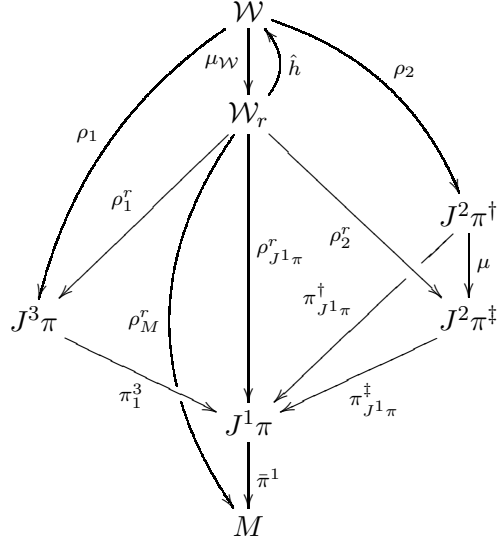
\xymatrix{
\ & \ & \mathcal{W} \ar@/_1.3pc/[llddd]_{\rho_1} \ar[d]_-{\mu_\mathcal{W}} \ar@/^1.3pc/[rrdd]^{\rho_2} & \ & \ \\
\ & \ & \mathcal{W}_r \ar@/_1pc/[u]_{\hat{h}} \ar[lldd]_{\rho_1^r} \ar[rrdd]_{\rho_2^r} \ar[ddd]^<(0.4){\rho_{J^{1}\pi}^r} \ar@/_2.5pc/[dddd]_-{\rho_M^r}|(.675){\hole} & \ & \ \\
\ & \ & \ & \ & J^{2}\pi^\dagger \ar[d]^-{\mu} \ar[lldd]_{\pi_{J^{1}\pi}^\dagger}|(.25){\hole} \\
J^{3}\pi \ar[rrd]_{\pi^{3}_{1}} & \ & \ & \ & J^{2}\pi^\ddagger \ar[dll]^{\pi_{J^{1}\pi}^\ddagger} \\
\ & \ & J^{1}\pi \ar[d]^{\bar{\pi}^{1}} & \ & \ \\
\ & \ & M & \ & \
}
$$

Now we define the {\sl Liouville forms } in $\W_r$,
$\Theta_r = (\rho_2 \circ \hat{h})^*\Theta_1^s \in \df^4(\W_r)$ and
$\Omega_r=-\d\Theta_r=(\rho_2 \circ \hat{h})^*\Omega_1^s\in\df^5(\W_r)$,
with local expressions
$$
\begin{array}{l}
\displaystyle
\Theta_r=
-\hat{H}\d^4x+\sum_{\alpha\leq\beta}p^{\alpha\beta,\mu}\d g_{\alpha\beta}\wedge \d^{3}x_\mu
+\sum_{\alpha\leq\beta}\frac{1}{n(\mu\nu)}p^{\alpha\beta,\mu\nu}\d g_{\alpha\beta,\mu}\wedge \d^{3}x_{\nu}
 \ , \\[10pt]
\displaystyle
\Omega_r =\d \hat{H} \wedge \d^4x-\sum_{\alpha\leq\beta}\d p^{\alpha\beta,\mu}\wedge \d g_{\alpha\beta}\wedge \d^{3}x_\mu-\sum_{\alpha\leq\beta}\frac{1}{n(\mu\nu)}\d p^{\alpha\beta,\mu\nu}\wedge \d g_{\alpha\beta,\mu}\wedge \d^{3}x_{\nu} \ .
\end{array}
$$
In the following, we commit an abuse of notation denoting also
$\hat{L} = (\pi^3_2 \circ \rho_1^r)^*L_{\mathfrak V}\in C^\infty(\mathcal{W}_r)$.
Then,  it is useful to consider the following decomposition \cite{first,rosado}:
\beq
\hat L=\sum_{\alpha\leq\beta}\hat L^{\alpha\beta,\mu\nu}g_{\alpha\beta,\mu\nu}+\hat L_0 \ ,
\label{L0}
\eeq
where
\bea
\hat L^{\alpha\beta,\mu\nu}&=&
\frac{1}{n(\mu\nu)}\frac{\partial \hat L}{\partial g_{\alpha\beta,\mu\nu}}=
\frac{n(\alpha\beta)}{2}\varrho(g^{\alpha\mu}g^{\beta\nu}+
g^{\alpha\nu}g^{\beta\mu}-2g^{\alpha\beta}g^{\mu\nu})\ ,
\label{L1}
 \\
\hat L_0&=&\varrho g^{\alpha\beta}\{g^{\gamma\delta}(g_{\delta\mu,\beta}\Gamma^{\mu}_{\alpha\gamma}-g_{\delta\mu,\gamma}\Gamma^{\mu}_{\alpha\beta}) +\Gamma^{\delta}_{\alpha\beta}\Gamma^{\gamma}_{\gamma\delta}-\Gamma^{\delta}_{\alpha\gamma}\Gamma^{\gamma}_{\beta\delta}\} \ .
\label{L2}
\eea
The point on this decomposition is to isolate the acceleration term, because 
$\hat L^{\alpha\beta,\mu\nu}$ and $\hat L_0$
project onto functions $L^{\alpha\beta,\mu\nu}\in C^\infty(E)$ and 
$L_0\in C^\infty(J^1\pi)$, respectively.
Another useful function is 
\beq
\hat L^{\alpha\beta,\mu}=
\frac{\partial \hat L}{\partial g_{\alpha\beta,\mu}} - 
\sum_{\nu=0}^3\frac{1}{n(\mu\nu)}D_\nu\left( \frac{\partial \hat L}{\partial g_{\alpha\beta,\mu\nu}}\right)=
\frac{\partial \hat L_0}{\partial g_{\alpha\beta,\mu}} - D_\nu \hat L^{\alpha\beta,\mu\nu}\ .
\label{L3}
\eeq

These forms are degenerate; namely,
\beq
\ker\,\Theta_r=\ker\,\Omega_r=
{\left<\frac{\partial}{\partial g_{\alpha\beta,\mu\nu}},\frac{\partial}{\partial g_{\alpha\beta,\mu\nu\lambda}}\right>}_{0\leq\alpha\leq\beta\leq3;\, 0\leq\mu\leq\nu\leq\lambda\leq3}\ .
\label{gaugevf}
\eeq
For a premultisymplectic form $\Omega$,
we call {\sl (geometric) gauge vector fields} to those vector fields belonging to $\ker\,\Omega$.
In this way, the coordinate vector fields in \eqref{gaugevf} are local gauge vector fields.
Furthermore, $\Theta_r$ is $(\pi^3_1\circ\rho^r_1)$-projectable.

\subsubsection{The Lagrangian-Hamiltonian problem}
\label{laghamunif}

Consider the system $(\W_r,\Omega_r)$.

\begin{definition}
A section $\psi \in \Gamma(\bar{\pi}^{k})$ is {\rm holonomic} if
$j^k(\pi^{k} \circ \psi) = \psi$; that is, $\psi$ is the $k$th prolongation of a section
$\phi = \pi^{k} \circ \psi \in \Gamma(\pi)$,
and an integrable and $\bar{\pi}_M$-transverse multivector field 
${\bf X} \in \mathfrak{X}^4(J^k\pi)$ is
{\rm holonomic} if its integral sections are holonomic
(see the appendix \ref{mvf} for details on multivector fields).

A section $\psi \in \Gamma(\bar{\pi}_M^\ddagger)$ is {\rm holonomic in 
$J^2\pi^\ddagger$} if
$\bar{\pi}_{J^1\pi}^\ddagger\circ\psi\in\Gamma(\bar{\pi}^{1})$ is holonomic in $J^1\pi$,
and an integrable and $\bar{\pi}_M^\ddagger$-transverse multivector field 
${\bf X} \in \mathfrak{X}^4(J^2\pi^\ddagger)$ is
{\rm holonomic} if its integral sections are holonomic.

Finally, a section $\psi \in \Gamma(\rho_M^r)$ is {\rm holonomic in $\W_r$} if
$\rho_1^r \circ \psi \in \Gamma(\bar{\pi}^{3})$ is holonomic in $J^{3}\pi$,
and an integrable and $\rho_M^r$-transverse multivector field ${\bf X} \in \mathfrak{X}^4(\mathcal{W}_r)$ is
{\rm holonomic} if its integral sections are holonomic.
\end{definition}

The local expression  of a holonomic multivector field 
${\bf X} \in \mathfrak{X}^4(\mathcal{W}_r)$ is
\bea
\mathbf{X} &=& 
\bigwedge_{\lambda=1}^4\sum_{\substack{\alpha\leq\beta\\\mu\leq\nu\leq\tau}}
\left(\derpar{}{x^\lambda}+g_{\alpha\beta,\lambda}\derpar{}{g_{\alpha\beta}}+
g_{\alpha\beta,\mu\lambda}\derpar{}{g_{\alpha\beta,\mu}}+
g_{\alpha\beta,\mu\nu\lambda}\derpar{}{g_{\alpha\beta,\mu\nu}}+
\right.\nonumber \\
& &
\qquad\qquad\qquad\left. 
F_{\alpha\beta,\mu\nu\tau\lambda}\derpar{}{g_{\alpha\beta,\mu\nu\tau}}+
G^{\alpha\beta,\mu}_\lambda\derpar{}{p^{\alpha\beta,\mu}}+ 
G^{\alpha\beta,\mu\nu}_\lambda\derpar{}{p^{\alpha\beta,\mu\nu}} \right) \ ,
\label{locholmv}
\eea
and, if 
$\psi(x^\lambda) = (x^\lambda,\,\psi_{\alpha\beta}(x^\lambda),
\,\psi_{\alpha\beta,\mu}(x^\lambda),\,\psi_{\alpha\beta,\mu\nu}(x^\lambda),
\,\psi_{\alpha\beta,\mu\nu\tau}(x^\lambda),
\psi^{\alpha\beta,\mu}(x^\lambda),\psi^{\alpha\beta,\mu\nu}(x^\lambda))$
is an integral section of ${\bf X}$, its component functions
satisfy the following system of partial differential equations
\bea
\derpar{\psi_{\alpha\beta}}{x^\lambda}=g_{\alpha\beta,\lambda}\circ\psi \ , \
\derpar{\psi_{\alpha\beta,\mu}}{x^\lambda}=g_{\alpha\beta,\mu\lambda}\circ\psi \ , \
\derpar{\psi_{\alpha\beta,\mu\nu}}{x^\lambda}=g_{\alpha\beta,\mu\nu\lambda}\circ\psi \ , \
\nonumber\\
\derpar{\psi_{\alpha\beta,\mu\nu\tau}}{x^\lambda}=F_{\alpha\beta,\mu\nu\tau\lambda}\circ\psi \ , \
\derpar{\psi^{\alpha\beta,\mu}}{x^\lambda}=G^{\alpha\beta,\mu}_\lambda\circ\psi  \ , \
\derpar{\psi^{\alpha\beta,\mu\nu}}{x^\lambda}=G^{\alpha\beta,\mu\nu}_\lambda\circ\psi  \ .
\label{pdesect}
\eea
It is important to point out that the fact that a multivector field in $\W_r$ 
has the local expression \eqref{locholmv}
(and then being locally decomposable and $\rho^r_M$-transverse)
is just a necessary condition to be holonomic, since it may not be integrable.
However, if such a multivector field admits integral sections, 
then its integral sections are holonomic.
In general, a locally decomposable and $\rho^r_M$-transverse multivector field
which has \eqref{locholmv} as coordinate expression, is said to be {\sl semiholonomic} in $\W_r$.

The \textsl{Lagrangian-Hamiltonian problem} associated with the system $(\W_r,\Omega_r)$
consists in finding holonomic sections $\psi \in \Gamma(\rho_M^r)$ satisfying
any of the following equivalent conditions:
\begin{enumerate}
\item $\psi$ is a solution to the equation
\begin{equation}\label{eqn:UnifFieldEqSect}
\psi^*\inn(X)\Omega_r = 0 \, , \quad \mbox{for every } X \in \mathfrak{X}(\mathcal{W}_r) \ .
\end{equation}
\item $\psi$ is an integral section of a multivector field contained in a class of holonomic multivector fields $\{ {\bf X} \} \subset \mathfrak{X}^4(\mathcal{W}_r)$
satisfying the equation
\begin{equation}\label{eqn:UnifDynEqMultiVF}
\inn({\bf X})\Omega_r = 0 \ .
\end{equation}
\end{enumerate}

As the form $\Omega_r$ is $1$-degenerate we have that
$(\mathcal{W}_r,\Omega_r)$ is a premultisymplectic system,
and solutions to (\ref{eqn:UnifFieldEqSect}) or (\ref{eqn:UnifDynEqMultiVF})
do not exist everywhere in $\mathcal{W}_r$. Then \cite{pere}:

\begin{proposition}
\label{prop:GraphLegMapSect}
A section $\psi \in \Gamma(\rho_M^r)$ solution to the equation \eqref{eqn:UnifFieldEqSect} 
takes values in a $140$-codimensional submanifold 
$\jmath_{\Lag_{\mathfrak V}}\colon\mathcal{W}_{\mathcal{L}_{\mathfrak V}} \hookrightarrow \mathcal{W}_r$ which is identified with the graph of
a bundle map $\mathcal{FL}_{\mathfrak V}\colon J^3\pi \to J^{2}\pi^\ddagger$, over $J^1\pi$, defined locally by
\beq
\mathcal{FL}_{\mathfrak V}^{\ \ *}\,p^{\alpha\beta,\mu}=
\frac{\partial \hat L}{\partial g_{\alpha\beta,\mu}} - 
\sum_{\nu=0}^{3}\frac{1}{n(\mu\nu)}
D_\nu\left( \frac{\partial \hat L}{\partial g_{\alpha\beta,\mu\nu}}\right)=
\hat L^{\alpha\beta,\mu} \ , \
\mathcal{FL}_{\mathfrak V}^{\ \ *}\,p^{\alpha\beta,\mu\nu}=
\frac{\partial \hat L}{\partial g_{\alpha\beta,\mu\nu}} \, .
\label{Legmap}
\eeq
What is equivalent, the submanifold 
$\mathcal{W}_{\mathcal{L}_{\mathfrak V}}$ is the graph of a bundle morphism
$\widetilde{\mathcal{FL}}_{\mathfrak V} \colon J^3\pi \to J^2\pi^\dagger$ over $J^1\pi$ defined locally by
\beann
\displaystyle\widetilde{\mathcal{FL}}_{\mathfrak V}^{\ *}\,p^{\alpha\beta,\mu}&=&
\frac{\partial \hat L}{\partial g_{\alpha\beta,\mu}} - 
\sum_{\nu=0}^{3}\frac{1}{n(\mu\nu)}
D_\nu\left( \frac{\partial \hat L}{\partial g_{\alpha\beta,\mu\nu}}\right)=
\hat L^{\alpha\beta,\mu} \ ,
\\
\widetilde{\mathcal{FL}}_{\mathfrak V}^{\ *}\,p^{\alpha\beta,\mu\nu} &=&
\frac{\partial \hat L}{\partial g_{\alpha\beta,\mu\nu}} \, , \\
\displaystyle
\widetilde{\mathcal{FL}}_{\mathfrak V}^{\ *}\ p&=& \hat{L} - 
\sum_{\alpha\leq\beta}g_{\alpha\beta,\mu}\left(\frac{\partial \hat L}{\partial g_{\alpha\beta,\mu}} - 
\sum_{\nu=0}^{3}\frac{1}{n(\mu\nu)}
D_\nu\left( \frac{\partial \hat L}{\partial g_{\alpha\beta,\mu\nu}}\right)\right)-
\sum_{\substack{\alpha\leq\beta\\\mu\leq\nu}}g_{\alpha\beta,\mu\nu}\frac{\partial \hat L}{\partial g_{\alpha\beta,\mu\nu}}\ .
\eeann
\end{proposition}

The maps $\Leg_{\mathfrak V}$ and $\widetilde{\Leg}_{\mathfrak V}$ are the \textsl{restricted} and
the {\sl extended Legendre maps} (associated with the Lagrangian density
$\Lag_{\mathfrak V}$), and they satisfy that 
$\mathcal{FL}_{\mathfrak V} = \mu \circ \widetilde{\mathcal{FL}}_{\mathfrak V}$
(they can be also be defined intrinsically, as we will see in Section \ref{LagForm}).
For every $j^3_x\phi \in J^3\pi$, we have that
$\operatorname{rank}(\widetilde{\mathcal{FL}}_{\mathfrak V}(j^3_x\phi)) = \operatorname{rank}(\mathcal{FL}_{\mathfrak V}(j^3_x\phi))$.

Remember that, according to \cite{art:Saunders_Crampin90}, 
a second-order Lagrangian density $\mathcal{L} \in \Omega^4(J^2\pi)$ is \textsl{regular} if
$$
\operatorname{rank}(\widetilde{\mathcal{FL}}(j^3_x\phi)) 
= \operatorname{rank}(\mathcal{FL}(j^3_x\phi)) = \dim J^{2}\pi + \dim J^{1}\pi - \dim E = \dim J^2\pi^\ddagger \, ,
$$
otherwise, the Lagrangian density is \textsl{singular}.
Regularity is equivalent to demand that $\mathcal{FL} \colon J^3\pi \to J^{2}\pi^\ddagger$ 
is a submersion onto $J^{2}\pi^\ddagger$ and
this implies that there exist local sections of 
$\mathcal{FL}$. If $\mathcal{FL}$ admits a global section
$\Upsilon \colon J^2\pi^\ddagger \to J^3\pi$, 
then the Lagrangian density is said to be {\sl hyperregular}.
Recall that the regularity of $\mathcal{L}$ determines if the section
$\psi \in \Gamma(\rho_M^r)$ solution to the equation \eqref{eqn:UnifFieldEqSect} lies in
$\mathcal{W}_\mathcal{L}$ or in a submanifold
$\mathcal{W}_f \hookrightarrow \mathcal{W}_\mathcal{L}$ where the section $\psi$ takes values.
In order to obtain this {\sl final constraint submanifold}, the best way is to work with the equation \eqref{eqn:UnifDynEqMultiVF} instead of (\ref{eqn:UnifFieldEqSect}).

Observe that the map 
$\rho_1^{\mathfrak V}= 
\rho_1^r \circ\jmath_{\mathcal{L}_{\mathfrak V}} \colon \mathcal{W}_{\mathcal{L}_{\mathfrak V}} \to J^3\pi$ 
is a diffeomorphism.

\subsubsection{Field equations for multivector fields}

First, the premultisymplectic constraint algorithm
\cite{art:deLeon_Marin_Marrero_Munoz_Roman05}
states that:

\begin{proposition}\label{prop:CompSubmanifoldMultiVF}
A solution ${\bf X} \in \mathfrak{X}^4(\mathcal{W}_r)$ to equation \eqref{eqn:UnifDynEqMultiVF}
exists only on the points of the  \textsl{compatibility
submanifold} $\mathcal{W}_c \hookrightarrow \mathcal{W}_r$ defined by
\begin{align*}
\mathcal{W}_c &= \left\{ w \in \mathcal{W}_r \colon (i(Z)\d\hat{H})(w) = 0 \, , \mbox{ for every }
Z \in \ker(\Omega_r) \right\} \\
&= \left\{ w \in \mathcal{W}_r \colon (i(Y)\Omega_r)(w) = 0 \, , \mbox { for every }
Y \in \mathfrak{X}^{V(\rho_2^r)}(\mathcal{W}_r)\right\} \, .
\end{align*}
\end{proposition}

Bearing in mind \eqref{gaugevf} and that 
$\displaystyle\inn\left(\derpar{}{g_{\alpha\beta,\mu\nu\tau}}\right)\d\hat{H}= 0$,
the functions locally defining this submanifold have the following coordinate expressions
\beq
\inn\left( \derpar{}{g_{\alpha\beta,\mu\nu\tau}}\right)\d\hat{H} = 
p^{\alpha\beta,\mu\nu}-\derpar{\hat{L}}{g_{\alpha\beta,\mu\nu}} \ .
\label{firstcons}
\eeq
Then, the tangency condition for the multivector fields ${\bf X}$
which are solutions to \eqref{eqn:UnifDynEqMultiVF} on $\mathcal{W}_{c}$
gives rise to $24$ new constraints 
$$
p^{\alpha\beta,\mu}-\frac{\partial \hat L}{\partial g_{\alpha\beta,\mu}}+
\sum_{\nu=0}^{3}\frac{1}{n(\mu\nu)}
D_\nu\left( \frac{\partial \hat L}{\partial g_{\alpha\beta,\mu\nu}}\right)= 0 \, .
$$
which define a submanifold of $\mathcal{W}_{c}$ 
that coincides with the submanifold $\mathcal{W}_\mathcal{L}$.
Now the study of the tangency of ${\bf X}$ along $\mathcal{W}_\mathcal{L}$
could introduce new constraints depending on the regularity of $\Lag$,
and the algorithm continues until we reach the submanifold $\mathcal{W}_f$.
The final result is given in the next theorem:

\begin{theorem}
\label{theo:submanifold} 
Let ${\cal W}_f\hookrightarrow\W_r$ be the submanifold defined locally by the constraints
$$
p^{\alpha\beta,\mu\nu}-\frac{\partial \hat L}{\partial g_{\alpha\beta,\mu\nu}}=0\quad , \quad 
p^{\alpha\beta,\mu}-\hat{L}^{\alpha\beta,\mu}=0\quad , \quad 
 \hat{L}^{\alpha\beta}=0 \quad , \quad 
D_\tau\hat{L}^{\alpha\beta}=0 \ ;
$$
for $0\leq\alpha\leq\beta\leq3$, $0\leq\mu\leq\nu\leq3$ and $0\leq\tau\leq3$. 
Then, there exist classes of  holonomic multivector fields 
$\{{\bf X}\}\subset\mathfrak{X}^4({\cal W}_r)$
which are tangent to $\W_f$ and such that
\beq{}\label{eqn:UniVec5}
\inn{({\bf X})}\Omega_r|_{\W_f}=0 \quad ,\quad\forall {\bf X}\in\{{\bf X}\}\subset\mathfrak{X}^4({\cal W}_r) \ .
\eeq
\end{theorem}
\begin{proof}
In order to find the final submanifold $\W_f$ we use a local
coordinate procedure which is equivalent to the
constraint algorithm for premultisymplectic field theories.
Bearing in mind \eqref{locholmv}, 
the local expression of a representative of a class of a 
semiholonomic multivector fields, not necessarily integrable, is, in this case,
\beann
{\bf X}=\bigwedge_{\tau=0}^3X_\tau&=&
\bigwedge_{\tau=0}^3\sum_{\substack{\alpha\leq\beta\\\mu\leq\nu\leq\lambda}}
\left(\derpar{}{x^\tau}
+ g_{\alpha\beta,\tau}\derpar{}{g_{\alpha\beta}}+ g_{\alpha\beta,\mu\tau}\derpar{}{g_{\alpha\beta,\mu}}+ g_{\alpha\beta,\mu\nu\tau}\derpar{}{g_{\alpha\beta,\mu\nu}}\right.
\\ & &
\left. F_{\alpha\beta;\mu\nu\lambda,\tau}\frac{\partial}{\partial g_{\alpha\beta,\mu\nu\lambda}}+G^{\alpha\beta,\mu }_\tau\frac{\partial}{\partial p^{\alpha\beta,\mu}}+G^{\alpha\beta,\mu\nu}_\tau\frac{\partial}{\partial p^{\alpha\beta,\mu\nu}}\right) \ ,
\eeann
then, equation \eqref{eqn:UnifDynEqMultiVF} leads to
\bea\label{eqn:UniVec1}
G^{\alpha\beta,\nu}_{\nu}-\frac{\partial\hat{L}}{\partial g_{\alpha\beta}}=0 \ ,
\\ \label{eqn:UniVec2}
\sum_{\nu=0}^3\frac{1}{n(\mu\nu)}G^{\alpha\beta,\mu\nu}_\nu-\frac{\partial\hat{L}}{\partial g_{\alpha\beta,\mu}}+p^{\alpha\beta,\mu}=0 \ ,
\\ \label{eqn:UniVec3}
p^{\alpha\beta,\mu\nu}-\hat L^{\alpha\beta,\mu\nu}=0 \ .
\eea
Equations \eqref{eqn:UniVec3} are what we obtain in Proposition \ref{prop:CompSubmanifoldMultiVF} (see \eqref{firstcons}),
and they are the constraints defining the compatibility submanifold $\mathcal{W}_c\hookrightarrow\W_r$. 
The tangency conditions on them, 
$$
\Lie(X_\tau)(p^{\alpha\beta,\mu\nu}-\frac{\partial \hat L}{\partial g_{\alpha\beta,\mu\nu}})\vert_{\W_c}=0 \ ,
$$
allows us to determine some coefficients
\beq
G^{\alpha\beta,\mu\nu}_\tau=
D_\tau\frac{\partial \hat L}{\partial g_{\alpha\beta,\mu\nu}} 
\quad ; \quad \mbox{\rm (on $\W_c$)} \ .
\label{G1}
\eeq
These new identities are not compatible with \eqref{eqn:UnifDynEqMultiVF}.
Indeed, combining them with \eqref{eqn:UniVec2} we have:
\beq
0=\sum_{\nu=0}^3\frac{1}{n(\mu\nu)}D_\nu\frac{\partial \hat L}{\partial g_{\alpha\beta,\mu\nu}}-
\frac{\partial\hat{L}}{\partial g_{\alpha\beta,\mu}}+
p^{\alpha\beta,\mu}=p^{\alpha\beta,\mu}-\hat{L}^{\alpha\beta,\mu}
\quad ; \quad \mbox{\rm (on $\W_c$)} \ .
\label{Wlag2}
\eeq
These restrictions define the submanifold $\mathcal{W}_{\mathcal{L}_{\mathfrak V}}\hookrightarrow\mathcal{W}_c$.
The tangency conditions on these new constraints, 
$$
\Lie(X_\tau)(p^{\alpha\beta,\mu}-\hat L^{\alpha\beta,\mu})\vert_{\W_{\Lag_{\mathfrak V}}}=0\ ,
$$
lead to
\beq
G^{\alpha\beta,\mu}_\tau=D_\tau\frac{\partial \hat L}{\partial g_{\alpha\beta,\mu}}-D_\tau D_\sigma\hat L^{\alpha\beta,\mu \sigma}
\quad ; \quad \mbox{\rm (on $\W_{\Lag_{\mathfrak V}}$)} \ .
\label{G2}
\eeq
Contracting the indices $\mu$ and $\tau$ in these restrictions 
and combining them with \eqref{eqn:UniVec1}, we obtain the new functions
$$
\hat L^{\alpha\beta}:=
\frac{\partial \hat{L}}{\partial g_{\alpha\beta}}-D_\nu \hat{L}^{\alpha\beta,\nu}=\frac{\partial \hat L}{\partial g_{\alpha\beta}}-D_\nu\frac{\partial \hat L}{\partial g_{\alpha\beta,\nu}}+
\sum_{\nu\leq\mu}D_\nu D_\mu\frac{\partial \hat L}{\partial g_{\alpha\beta,\nu\mu}}=0
\quad ; \quad \mbox{\rm (on $\W_{\Lag_{\mathfrak V}}$)}  \ ,
$$
which are explicitly
\beq
\label{eqn:Res3}
\hat L^{\alpha\beta}=
-\varrho\,n(\alpha\beta) \left(R^{\alpha\beta}-
\frac{1}{2}g^{\alpha\beta}R\right)=0
\quad ; \quad \mbox{\rm (on $\W_{\Lag_{\mathfrak V}}$)} \ .
\eeq
These are the Euler-Lagrange equations, and  when they are evaluated on sections 
in $\W_{\Lag_{\mathfrak V}}$ we recover  the Einstein equations 
$\displaystyle \left.( R_{\alpha\beta}-\sfrac{1}{2}g_{\alpha\beta}R)\right|_\psi=0$. 
From its definition we can see that $\hat L^{\alpha\beta}$ do not depend 
neither on the momenta, nor on higher order velocities than 
the accelerations of the components of the metric, therefore 
$\hat L^{\alpha\beta}$ project onto ${J^2\pi}$. 
The equations \eqref{eqn:Res3} are algebraic combinations of the coordinates of 
$\W_\Lag$ and a solution can only exists on the points where they vanish. 
Thus, $\hat L^{\alpha\beta}$ are new constraints
which define locally the submanifold 
$\mathcal{W}_1\hookrightarrow\mathcal{W}_{\mathcal{L}_{\mathfrak V}}\hookrightarrow\mathcal{W}_r$.
(Note that, as a consequence of the Bianchi identities, these constraints are not independent all of them).
Continuing with the constraint algorithm, we
consider the tangency conditions on these constraints, 
$$
\Lie(X_\tau)\hat{L}^{\alpha\beta}\vert_{\W_1}=0 \ ,
$$
which lead to
\beq
\label{eqn:Res4}
D_\tau\hat L^{\alpha\beta}=
D_\tau \left(-\varrho\,n(\alpha\beta)\left(R^{\alpha\beta}-\frac{1}{2}g^{\alpha\beta}R\right)\right)=0 
\quad ; \quad \mbox{\rm (on $\W_1$)}\ .
\eeq
These are new constraints again (observe that these functions 
$D_\tau\hat L^{\alpha\beta}$ project onto $J^3\pi$,
since they do not depend on the higher-order derivatives and the momenta). They define locally the submanifold 
$\mathcal{W}_f\hookrightarrow\mathcal{W}_1\hookrightarrow\mathcal{W}_{\mathcal{L}_{\mathfrak V}}\hookrightarrow\mathcal{W}_r$.
Finally, the new tangency conditions, 
$$
\Lie(X_\sigma)D_\tau\hat{L}^{\alpha\beta}\vert_{\W_f}=0 \ ,
$$
 lead to
\bea
\sum_{\substack{\gamma\leq\lambda\\\mu\leq\nu\leq\kappa}}  \left(
\derpar{}{x^\sigma}+g_{\gamma\lambda,\sigma}\derpar{}{g_{\gamma\lambda}}+ 
g_{\gamma\lambda,\mu\sigma}\derpar{}{g_{\gamma\lambda,\mu}}+ g_{\gamma\lambda,\mu\nu\sigma}\derpar{}{g_{\gamma\lambda,\mu\nu}}\right.
\nonumber \\
   \left. 
+F_{\gamma\lambda;\mu\nu\kappa,\sigma}
\frac{\partial}{\partial g_{\gamma\lambda,\mu\nu\kappa}}\right)D_\tau\hat L^{\alpha\beta}=0
 \quad ; \quad
 \mbox{\rm (on $\W_f$)}\ .
\label{Eq:LastConsAlgo}
\eea
and these equations allows us to determine some functions 
$F_{\gamma\lambda;\mu\nu\kappa,\sigma}$.
The manifold $\mathcal{W}_f$ is actually the final constraint submanifold 
because there exist integrable holonomic multivector fields solutions to
equations \eqref{eqn:UniVec5} on $\W_f$, tangent to $\W_f$,
which are (partially) determined by the conditions
\eqref{G1}, \eqref{G2}, and \eqref{Eq:LastConsAlgo}; that is,
\bea
{\bf{X}}&&=\bigwedge_{\tau=0}^3\sum_{\substack{\alpha\leq\beta\\\mu\leq\nu\leq\lambda}}\left(\frac{\partial}{\partial x^\tau}+ g_{\alpha\beta,\tau}\frac{\partial}{\partial g_{\alpha\beta}}+ g_{\alpha\beta,\mu\tau}\frac{\partial}{\partial g_{\alpha\beta,\mu}}+ g_{\alpha\beta,\mu\nu\tau}\frac{\partial}{\partial g_{\alpha\beta,\mu\nu}}+\right.
\nonumber \\
&&\left.D_\tau D_\gamma (g_{\lambda\sigma}(\Gamma_{\nu \alpha }^\lambda\Gamma_{\mu \beta}^\sigma+\Gamma_{\nu \beta}^\lambda\Gamma_{\mu \alpha }^\sigma))\frac{\partial}{\partial g_{\alpha\beta,\mu\nu\gamma}}+ D_\tau\hat{L}^{\alpha\beta,\mu}\frac{\partial}{\partial p^{\alpha\beta,\mu}}+D_\tau\frac{\partial \hat L}{\partial g_{\alpha\beta,\mu\nu}}\frac{\partial}{\partial p^{\alpha\beta,\mu\nu}}\right)\ .
\label{mvfuni}
\eea
One can prove (after a long computation) that this is actually 
an integrable solution (see section \ref{solutions} for more details). 
Finally, we have that the complete set of constraint functions defining 
the final constraint submanifold $\W_f\hookrightarrow\W_r$ 
are given by \eqref{eqn:UniVec3}, \eqref{Wlag2}, \eqref{eqn:Res3} and \eqref{eqn:Res4};
that is,
$$
p^{\alpha\beta,\mu\nu}-\frac{\partial \hat L}{\partial g_{\alpha\beta,\mu\nu}}=0\quad , \quad 
p^{\alpha\beta,\mu}-\hat{L}^{\alpha\beta,\mu}=0\quad , \quad 
 \hat{L}^{\alpha\beta}=0 \quad , \quad 
D_\tau\hat{L}^{\alpha\beta}=0 \ .
$$
\end{proof}

\subsubsection{Field equations for sections}\label{FieEqSec}

Once the holonomic multivector fields which are solutions to 
equation \eqref{eqn:UnifDynEqMultiVF} (on $\W_f$) have been obtained, 
in order to obtain the field equations for sections
we can use, either the equations \eqref{pdesect},
or the equivalent equations \eqref{eqn:UnifFieldEqSect}
which the integral sections of these multivector fields satisfy.
Thus, if these sections are locally given by 
$$
\psi(x^\lambda)=(x^\lambda,\,\psi_{\alpha\beta}(x^\lambda),\,\psi_{\alpha\beta,\mu}(x^\lambda),\,\psi_{\alpha\beta,\mu\nu}(x^\lambda),
\,\psi_{\alpha\beta,\mu\nu\tau}(x^\lambda),\,\psi^{\alpha\beta,\mu}(x^\lambda),\,\psi^{\alpha\beta,\mu\nu}(x^\lambda))\ ,
$$
the equation \eqref{eqn:UnifDynEqMultiVF} leads to
\bea
\frac{\partial \psi^{\alpha\beta,\mu}}{\partial x^\mu}-\frac{\partial \hat L}{\partial g_{\alpha\beta}}&=&0 \ ,
\label{eqsecunif1} \\
\frac{\partial \psi^{\alpha\beta,\mu\nu}}{\partial x^\nu}+\psi^{\alpha\beta,\mu}-\frac{\partial \hat L}{\partial g_{\alpha\beta,\mu}}&=&0 \ ,
\label{eqsecunif2} \\
\psi^{\alpha\beta,\mu\nu}-\hat L^{\alpha\beta,\mu\nu}&=&0 \ ,
\label{eqsecunif3} \\
\psi_{\alpha\beta,\mu}-\frac{\partial \psi_{\alpha\beta}}{\partial x^\mu}&=&0 \ ,
\label{eqsecunif4} \\
\psi_{\alpha\beta,\mu\nu}-\frac{1}{n(\mu\nu)}\left(\frac{\partial \psi_{\alpha\beta,\mu}}{\partial x^{\nu}}+\frac{\partial \psi_{\alpha\beta,\nu}}{\partial x^{\mu}}\right)&=&0 \ .
\label{eqsecunif5}
\eea
Equations \eqref{eqsecunif4} and \eqref{eqsecunif5}
are part of the holonomy conditions.
Equations \eqref{eqsecunif2} and \eqref{eqsecunif3},
as they do not involve the derivatives of the fields higher than $3$, are just
relations among the coordinates of the points in $\W_r$, which are equivalent to equations
\eqref{eqn:UniVec3} and \eqref{eqn:UniVec2}, respectively, and they
define the Legendre map introduced in \eqref{Legmap}.
They show that, as discussed above, the section $\psi$ take values in the submanifold 
$$
\mathcal{W}_{\mathcal{L}_{\mathfrak V}}=\left\{w\in\mathcal{W}_r \ |\ p^{\alpha\beta,\mu\nu}=\frac{\partial \hat L}{\partial g_{\alpha\beta,\mu\nu}}(w)\ ,\ 
p^{\alpha\beta,\mu}=\hat L^{\alpha\beta,\mu}(w)\right\}=
{\rm graph}\,\Leg_{\mathfrak V} \ .
$$
Finally, combining the equations (\ref{eqsecunif1}) 
with the local expression of the Legendre map given by
the equations (\ref{eqsecunif2}) and (\ref{eqsecunif3}) we obtain
\beq
\hat L^{\alpha\beta}|_\psi:=
\left.\left(\derpar{\hat L}{g_{\alpha\beta}} - 
 D_\mu\,\derpar{\hat L}{g_{\alpha\beta,\mu}} + 
\sum_{\substack{\mu\leq\nu}}D_\mu D_\nu \, \derpar{\hat L}{g_{\alpha\beta,\mu\nu}}\right)\right|_\psi=
\left.-\varrho\,n(\alpha\beta) \left(R^{\alpha\beta}-
\frac{1}{2}g^{\alpha\beta}R\right)\right|_\psi=0 \ .
\label{eqsEL}
\eeq
These are the Euler-Lagrange equations for a section $\psi\in\Gamma(\rho^r_M)$, 
which are equivalent to  the Einstein equations 
\beq
\displaystyle \left( R_{\alpha\beta}-\frac{1}{2}g_{\alpha\beta}R\right)\Big\vert_\psi=0 \ ;
\label{Einseq}
\eeq
and, as it is well known, they are of order two.

If $\psi$ is a holonomic section solution to \eqref{eqn:UnifFieldEqSect}, the tangency conditions on the Einstein's equations 
are automatically satisfied. Indeed, the last constraints \eqref{eqn:Res4} read 
$$
\left.\displaystyle\left(D_\tau\hat{L}^{\alpha\beta}\right)\right|_{\psi}=\frac{\partial(\hat{L}^{\alpha\beta}\circ\psi)}{\partial x^\tau}=0 \ ,
$$
which is automatically satisfied because $\psi$, in particular, 
is a solution to the Einstein equations \eqref{Einseq} and
then \eqref{eqsEL} holds. 
Using the same reasoning, we can check that  \eqref{Eq:LastConsAlgo} is also automatically satisfied. 
These last equations fix the gauge freedom, therefore the gauge symmetry does not show when considering the Einstein's equations for sections.


\subsection{Recovering the Lagrangian and Hamiltonian formalisms}

\subsubsection{Lagrangian formalism}
\label{LagForm}

(See \cite{pere, pere2} for the general details).
Let $\Theta_1^s \in \Omega^{4}(J^2\pi^\dagger)$ and 
$\Omega_1^s \in \Omega^{5}(J^2\pi^\dagger)$ be the
symmetrized Liouville forms in $J^2\pi^\dagger$. The {\sl Poincar\'{e}-Cartan forms} in
$J^3\pi$ are the forms defined as 
$$
\Theta_{\mathcal{L}_{\mathfrak V}}= \widetilde{\mathcal{FL}}_{\mathfrak V}^{\ *}\Theta_1^s \in \Omega^{4}(J^3\pi)
\quad , \quad
\Omega_{\mathcal{L}_{\mathfrak V}}= \widetilde{\mathcal{FL}}_{\mathfrak V}^{\ *}\Omega_1^s = -\d\Theta_{\mathcal{L}_{\mathfrak V}} \in \Omega^{5}(J^3\pi) \ .
$$
The Poincar\'{e}-Cartan $4$-form satisfies that
$\rho_2^*\Theta_1^s=\rho_1^*\Theta_{\mathcal{L}_{\mathfrak V}}$,
and the same result holds for the Poincar\'{e}-Cartan $5$-form $\Omega_{\mathcal{L}_{\mathfrak V}}$.

\noindent {\bf Remark}:
These forms coincide with the usual Poincar\'{e}-Cartan forms for second-order classical field
theories that can be found in the literature \cite{art:Aldaya_Azcarraga78_2,proc:Garcia_Munoz83,art:Kouranbaeva_Shkoller00,art:Munoz85},
which are constructed taking the Lagrangian density $\Lag_{\mathfrak V}$ and
using the canonical structures of the higher-order jet bundles.
Therefore, starting from the $4$-Poincar\'{e}-Cartan form
$\Theta_{\mathcal{L}_{\mathfrak V}}$, the Legendre maps $\widetilde{\mathcal{FL}}_{\mathfrak V}$ and
${\mathcal{FL}}_{\mathfrak V}$ can be defined in an intrinsic way as follows
\cite{art:Echeverria_Lopez_Marin_Munoz_Roman04}:
if $Z_1,\ldots,Z_m\in\Tan_{\pi^3(j^3_x\phi)}E$, and
  ${\bar Z}_1,\ldots,{\bar Z}_m\in\Tan_{j^3_x\phi}J^3\pi$ are such that
  $\Tan_{j^3_x\phi}\pi^3\bar Z_\alpha=Z_\alpha$; then, 
for every $j^3_x\phi \in J^3\pi$, the extended Legendre map,
  $\widetilde{{\cal F}\Lag_{\mathfrak V}}\colon J^3\pi\to J^2\pi^\dagger$, is given by
$[\widetilde{{\cal F}\Lag}_{\mathfrak V}(j^3_x\phi)](Z_1,\ldots,Z_m):=
  (\Theta_{\Lag_{\mathfrak V}})_{j^3_x\phi}({\bar Z}_1,\ldots,{\bar Z}_m)$,
and then the restricted Legendre map is 
${\cal F}\Lag_{\mathfrak V} :=\mu\circ\widetilde{{\cal F}\Lag}_{\mathfrak V}$.

Using natural coordinates in $J^3\pi$, we have the local expression
\bea
&\Theta_{\mathcal{L}_{\mathfrak V}}&=-\left(\sum_{\alpha\leq\beta}L^{\alpha\beta,\mu}g_{\alpha\beta,\mu}+\sum_{\alpha\leq\beta}L^{\alpha\beta,\mu\nu}g_{\alpha\beta,\mu\nu}-L\right)\d^4x
\nonumber \\
&&+\sum_{\alpha\leq\beta}L^{\alpha\beta,\mu}\d g_{\alpha\beta}\wedge \d^3x_\mu+\sum_{\alpha\leq\beta}L^{\alpha\beta,\mu\nu}\d g_{\alpha\beta,\mu}\wedge \d^3x_{\nu}\ .
\label{thetalag}
\eea
Notice that, if
\beq
H\equiv 
(\jmath_{\Lag_{\mathfrak V}}\circ(\rho_1^{\mathfrak V})^{-1})^*\hat H=
\sum_{\substack{\alpha\leq \beta}}L^{\alpha\beta,\mu\nu}g_{\alpha\beta,\mu\nu}+
\sum_{\substack{\alpha\leq \beta}}L^{\alpha\beta,\mu}g_{\alpha\beta,\mu}-L=
\varrho\, g_{\alpha\beta,\mu}g_{kl,\nu}H^{\alpha\beta k l \mu\nu} \ ,
\label{Hamfun}
\eeq
where
\beq
H^{\alpha\beta k l\mu\nu}=
\frac{1}{4}g^{\alpha\beta}g^{kl}g^{\mu\nu}-\frac14g^{\alpha k}g^{\beta l}g^{\mu\nu}+\frac12g^{\alpha k}g^{l\mu}g^{\beta\nu}-\frac12g^{\alpha\beta}g^{l\nu}g^{k\mu} \ ,
\label{Hindex}
\eeq
then
$$
 \Omega_{\mathcal{L}_{\mathfrak V}}=
 -\d\Theta_{\Lag_{\mathfrak V}}=
\d H\wedge\d^4x-\sum_{\alpha\leq\beta}\d L^{\alpha\beta,\mu}\d g_{\alpha\beta}\wedge \d^{m-1}x_\mu-\sum_{\alpha\leq\beta}\d L^{\alpha\beta,\mu\nu}\d g_{\alpha\beta,\mu}\wedge \d^{m-1}x_{\nu}
\in\df^5(J^3\pi) \ ;
$$
where we have denoted
$L=(\jmath_{\Lag_{\mathfrak V}}\circ(\rho_1^{\mathfrak V})^{-1})^*\hat L$,
$L^{\alpha\beta,\mu\nu}=(\jmath_{\Lag_{\mathfrak V}}\circ(\rho_1^{\mathfrak V})^{-1})^*\hat L^{\alpha\beta,\mu\nu}$,
$L^{\alpha\beta,\mu}=(\jmath_{\Lag_{\mathfrak V}}\circ(\rho_1^{\mathfrak V})^{-1})^*\hat L^{\alpha\beta,\mu}$,
and $L_0=(\jmath_{\Lag_{\mathfrak V}}\circ(\rho_1^{\mathfrak V})^{-1})^*\hat L_0$, which have the same coordinate expressions than
$\hat L$, $\hat L^{\alpha\beta,\mu\nu}$,
$\hat L_0$, and $\hat L^{\alpha\beta,\mu}$, 
given in\eqref{L0}, \eqref{L1}, \eqref{L2}, and \eqref{L3}, respectively.
Observe that this is a pre-multisymplectic form since, locally,
$$
\ker\,\Omega_{\Lag_{\mathfrak V}}=
{\left<\frac{\partial}{\partial g_{\alpha\beta,\mu\nu}},\frac{\partial}{\partial g_{\alpha\beta,\mu\nu\lambda}}\right>}_{0\leq\alpha\leq\beta\leq3;\, 0\leq\mu\leq\nu\leq\lambda\leq3}\ .
$$

Thus we have the Lagrangian system $(J^3\pi,\Omega_{\mathcal{L}_{\mathfrak V}})$,
and the \textsl{Lagrangian problem} associated with this system
consists in finding holonomic sections $\psi_\Lag=j^3\phi\in \Gamma(\bar\pi^3)$ (with $\phi\in\Gamma(\pi)$) satisfying
any of the following equivalent conditions:
\begin{enumerate}
\item $\psi_\Lag$ is a solution to the equation
\begin{equation}\label{eq:ls}
\psi_\Lag^*\inn(X)\Omega_{\Lag_{\mathfrak V}} = 0 \, , \quad \mbox{for every } X \in \mathfrak{X}(J^3\pi) \ .
\end{equation}
\item $\psi_\Lag$ is an integral section of a multivector field contained 
in a class of holonomic multivector fields $\{ {\bf X}_\Lag \}\subset \mathfrak{X}^4(J^3\pi)$
satisfying the equation
\begin{equation}
\label{eq:lf}
\inn({\bf X}_\Lag)\Omega_{\Lag_{\mathfrak V}}=0 \ .
\end{equation}
\end{enumerate}

In order to recover the Lagrangian field equations, the Poincar\'{e}-Cartan forms defined in $J^3\pi$
satisfy $(\rho_1^{\mathfrak V})^*\Theta_{\mathcal{L}_{\mathfrak V}} = \jmath_{\mathcal{L}_{\mathfrak V}}^*\Theta_r$
and $(\rho_1^{\mathfrak V})^*\Omega_{\mathcal{L}_{\mathfrak V}}= \jmath_{\mathcal{L}_{\mathfrak V}}^*\Omega_r$.
Then, the solution to the Lagrangian problem associated with 
the singular Lagrangian system $(J^3\pi,\Omega_{\mathcal{L}_{\mathfrak V}})$,
which is stated in the equations \eqref{eq:ls} and \eqref{eq:lf},
is given by the following Proposition \ref{prop:UnifToLagSect} and Theorem \ref{thm:UnifToLagMultiVF}:

\begin{proposition}
\label{prop:UnifToLagSect}
If $\psi\in \Gamma(\rho_M^r)$ be a holonomic section solution to the equation
\eqref{eqn:UnifFieldEqSect}, then the section $\psi_\mathcal{L} = \rho_1^r \circ \psi \in \Gamma(\bar{\pi}^3)$
is holonomic, and is a solution to the equation
\begin{equation}\label{eqn:LagDynEqSect}
\psi_\mathcal{L}^*i(X)\Omega_{\mathcal{L}_{\mathfrak V}}=0 \quad , \quad \mbox{for every } X \in \mathfrak{X}(J^3\pi) \ .
\end{equation}
Conversely, if $\psi_\mathcal{L} \in \Gamma(\bar{\pi}^3)$ 
is a holonomic section solution to the field equation
\eqref{eqn:LagDynEqSect}, then the section
$\psi = \jmath_{\mathcal{L}_{\mathfrak V}}\circ (\rho_1^{\mathfrak V})^{-1} \circ \psi_\mathcal{L} \in \Gamma(\rho_M^r)$
is holonomic and it is a solution to the equation \eqref{eqn:UnifFieldEqSect}.
\end{proposition}

In local coordinates in $J^3\pi$, the equation for the holonomic section 
$\psi_\mathcal{L}=j^3\phi$ are the Euler-Lagrange equations
\begin{equation}
 \label{EL-EH}
\left(\restric{\derpar{L}{g_{\mu\nu}}-D_\mu
\derpar{L}{g_{\alpha\beta,\mu}} + 
\sum_{\substack{\mu\leq\nu}}D_\mu D_\nu \, \derpar{L}{g_{\alpha\beta,\mu\nu}}\right)}{j^3\phi} = 0 \ .
\end{equation}

\begin{theorem}
\label{thm:UnifToLagMultiVF}
Let $\mathbf{X} \in \mathfrak{X}^4(\mathcal{W}_r)$ be a holonomic multivector field solution to the
equation \eqref{eqn:UnifDynEqMultiVF}, at least on the points of a submanifold
$\jmath_f\colon\mathcal{W}_f \subseteq\mathcal{W}_{\mathcal{L}_{\mathfrak V}}\hookrightarrow \mathcal{W}_r$,
and tangent to $\mathcal{W}_f$. Then there
exists a unique holonomic multivector field $\mathbf{X}_\mathcal{L} \in \mathfrak{X}^4(J^3\pi)$
solution to the following equation,
at least on the points of $S_f = \rho_1^{\mathfrak V}(\mathcal{W}_f)$, and tangent to $S_f$,
\begin{equation}\label{eqn:LagDynEqMultiVF}
\inn(\mathbf{X}_\mathcal{L})\Omega_{\mathcal{L}_{\mathfrak V}} = 0 \, ,
\end{equation}
Conversely, if $\mathbf{X}_\mathcal{L} \in \mathfrak{X}^4(J^3\pi)$ is a holonomic multivector field
solution to the equation \eqref{eqn:LagDynEqMultiVF}, at least on the points of a submanifold
$S_f \hookrightarrow J^3\pi$, and tangent to $S_f$; then there exists a unique
holonomic multivector field $\mathbf{X} \in \mathfrak{X}^4(\mathcal{W}_r)$ which is a solution to the equation
\eqref{eqn:UnifDynEqMultiVF}, at least on the points of $\mathcal{W}_f=(\rho_1^{\mathfrak V})^{-1}(S_f)
\hookrightarrow \mathcal{W}_{\Lag_{\mathfrak V}}\hookrightarrow \mathcal{W}_r$,
and tangent to $\mathcal{W}_f$.
(See diagram \eqref{maindiag}).

The relation between these multivector fields is \
$\mathbf{X}_\mathcal{L} \circ \rho_1^r \circ\jmath_f= \Lambda^4\Tan\rho_1^r \circ \mathbf{X} \circ\jmath_f$.
\end{theorem}

As we have pointed out before, the equalities \eqref{eqn:UniVec3} and \eqref{Wlag2}
define the submanifold $\W_{\Lag_{\mathfrak V}}$ which is diffeomorphic with $J^3\pi$,
and the constraint functions defining  the Lagrangian final constraint submanifold
$S_f\hookrightarrow J^3\pi$ are
\bea
L^{\alpha\beta}=
\derpar{L}{g_{\alpha\beta}}-D_\mu
\derpar{L}{g_{\alpha\beta,\mu}} + 
\sum_{\substack{\mu\leq\nu}}D_\mu D_\nu \, \derpar{L}{g_{\alpha\beta,\mu\nu}}=
-\varrho\,n(\alpha\beta) \left(R^{\alpha\beta}-\frac{1}{2}g^{\alpha\beta}R\right)&=&0 \ ,\quad
\label{Einseqagain}
\\
D_\tau L^{\alpha\beta}&=&0 \ . \quad
\label{Einseqagain2}
\eea

The local expression of a representative of a class of
holonomic multivector fields in $J^3\pi$ is
\beq
{\bf X}=
\bigwedge_{\tau=0}^3\sum_{\substack{\alpha\leq\beta\\\mu\leq\nu\leq\lambda}}
\left(\derpar{}{x^\tau}
+ g_{\alpha\beta,\tau}\derpar{}{g_{\alpha\beta}}+ g_{\alpha\beta,\mu\tau}\derpar{}{g_{\alpha\beta,\mu}}+ g_{\alpha\beta,\mu\nu\tau}\derpar{}{g_{\alpha\beta,\mu\nu}}+
F_{\alpha\beta;\mu\nu\lambda,\tau}\frac{\partial}{\partial g_{\alpha\beta,\mu\nu\lambda}}\right) \ ;
\label{imvlag}
\eeq
then,  there are holonomic multivector fields
which are solutions to the equation \eqref{eq:lf} on $S_f$, and tangent to $S_f$.
They are obtained from \eqref{mvfuni} using Theorem \ref{thm:UnifToLagMultiVF}: 
\beann
{\bf{X}}_\Lag=\displaystyle
\bigwedge_{\tau=0}^3 \sum_{\substack{\alpha\leq\beta\\\mu\leq\nu\leq\lambda}}&\left(\derpar{}{x^\tau}+ g_{\alpha\beta,\tau}\frac{\partial}{\partial g_{\alpha\beta}}+ g_{\alpha\beta,\mu\tau}\frac{\partial}{\partial g_{\alpha\beta,\mu}}+ g_{\alpha\beta,\mu\nu\tau}\frac{\partial}{\partial g_{\alpha\beta,\mu\nu}}+\right.
\\
&\left.  (g_{\lambda\sigma}(\Gamma_{\nu \alpha }^\lambda\Gamma_{\mu \beta}^\sigma+\Gamma_{\nu \beta}^\lambda\Gamma_{\mu \alpha }^\sigma))\derpar{}{g_{\alpha\beta,\mu\nu\lambda}}\right) \ .
\eeann

Finally, for the equations of the integral sections of these multivector fields
(equation \eqref{eq:ls}), from \eqref{Einseq}, 
we obtain that \eqref{Einseqagain},
evaluated on the points in the image of holonomic sections 
$\psi_\Lag=j^3\phi$ in $J^3\pi$ 
(see Prop \ref{prop:UnifToLagSect} and \eqref{EL-EH}), 
are equivalent to the Einstein equations
\bea
L^{\alpha\beta}|_{j^3\phi}&=&
\left(\restric{\derpar{L}{g_{\alpha\beta}} - 
 D_\mu\,\derpar{L}{g_{\alpha\beta,\mu}} + 
\sum_{\substack{\mu\leq\nu}}D_\mu D_\nu \, \derpar{L}{g_{\alpha\beta,\mu\nu}}\right)}{j^3\phi} \nonumber \\
&=&\left.-\varrho\,n(\alpha\beta) \left(R^{\alpha\beta}-\frac{1}{2}g^{\alpha\beta}R\right)\right|_{j^3\phi}=0 \ .
\label{EEq}
\eea

All these results can be also obtained applying the constraint algorithm straightforwardly
for the equation \eqref{eq:lf}, in the same way as we have done for the unified formalism;
then doing a purely Lagrangian analysis.
Thus, the Euler-Lagrange equations for an holonomic multivector field like
\eqref{imvlag} (which are obtained from \eqref{eq:lf}) read as
$$
\sum_{\substack{\rho\leq\sigma,\mu\leq\nu,\lambda\leq\tau}}
\left(\frac{\partial^2L}{\partial g_{\alpha\beta,\mu\nu}\partial g_{\rho\sigma,\lambda\tau}}\right)\left(F_{\rho\sigma;\lambda\tau\mu,\nu}-D_\nu g_{\rho\sigma;\lambda\tau\mu}\right)+
L^{\alpha\beta}=0 \ ,
$$
and, as for the Hilbert-Einstein Lagrangian the Hessian matrix 
$\displaystyle\left(\frac{\partial^2L}{\partial g_{\alpha\beta,\rho\sigma}\partial g_{\mu\nu,\lambda\tau}}\right)$ vanishes identically,
we obtain that $L^{\alpha\beta}=0$, which are
the compatibility conditions for the Euler-Lagrange equations;  that is,
the primary Lagrangian constraints \eqref{Einseqagain}.
From here, the constraint algorithm continues by requiring 
the tangency condition, as it is usual (see \cite{art:GR-2016}).


\subsubsection{Hamiltonian formalism}
\label{hamiltonian}

(See \cite{pere,art:Roman09} for the general details).
Consider the Legendre maps introduced in Proposition \ref{prop:GraphLegMapSect}.
Then
$$
\Tan_{{j^3_x}\phi}\mathcal{FL}_{\mathfrak V}=\left( \begin{array}{ccccc}
1 & 0 & 0 & 0 & 0 \\
0 & 1 & 0 & 0 & 0 \\
0 & 0 & 1 & 0 & 0  \\
0 & 
\displaystyle\frac{\partial \hat L^{\alpha\beta,\mu}}{\partial g_{\gamma\delta}} & \displaystyle\frac{\partial \hat L^{\alpha\beta,\mu}}{\partial g_{\gamma\delta,\tau}}&
 0 & 0  \\
0 &
\displaystyle\frac{\partial^2\hat L}{\partial g_{\gamma\delta}\partial g_{\alpha\beta,\mu\nu}} & 0 & 0 & 0 
\end{array} \right) \ ,
$$
and we have that ${\rm rank}(\Tan_{{j^3_x}\phi}\mathcal{FL}_{\mathfrak V})=54$. Furthermore, locally we have that
\beq
\ker\,(\Leg_{\mathcal{L}_{\mathfrak V}})_*=
\ker\,\Omega_{\Lag_{\mathfrak V}}={\left<\frac{\partial}{\partial g_{\alpha\beta,\mu\nu}},\frac{\partial}{\partial g_{\alpha\beta,\mu\nu\lambda}}\right>}_{0\leq\alpha\leq\beta\leq3;\, 0\leq\mu\leq\nu\leq\lambda\leq3}\ ,
\label{kerfl}
\eeq
and thus $\mathcal{FL}_{\mathfrak V}$ is highly degerated.

Denote $\widetilde{\mathcal{P}}=\widetilde{\mathcal{FL}}_{\mathfrak V}(J^3\pi) \stackrel{\tilde{\jmath}}{\hookrightarrow} J^2\pi^\dagger$
and $\mathcal{P}=\mathcal{FL}_{\mathfrak V}(J^3\pi)\stackrel{\jmath}{\hookrightarrow} J^2\pi^\ddagger$,
and let $\mathcal{FL}_{\mathfrak V}^o$ be the map defined by 
$\mathcal{FL}_{\mathfrak V}=\jmath\circ \mathcal{FL}_{\mathfrak V}^o$ and 
$\bar{\pi}_{\mathcal{P}}\colon\mathcal{P}\to M$ 
the natural projection.
In order to assure the existence of the Hamiltonian formalism 
we have to assure that
the Lagrangian density $\mathcal{L}_{\mathfrak V}\in\Omega^4(J^2\pi)$ is, at least, 
{\sl almost-regular}; that is, $\mathcal{P}$ is a closed submanifold 
of $J^2\pi^\ddagger$, $\mathcal{FL}_{\mathfrak V}$ is a submersion onto its image
and, for every $j^3_x\phi \in J^3\pi$, the fibers 
$\mathcal{FL}_{\mathfrak V}^{-1}(\mathcal{FL}_{\mathfrak V}(j^3_x\phi))$
are connected submanifolds of $J^3\pi$.
Then, the following result allows us to consider the Hamiltonian formalism:

\begin{proposition}
\label{prop:AlmReg}
$\mathcal{L}_{\mathfrak V}$ is an almost-regular Lagrangian
and
$\mathcal{P}$ is diffeomorphic to $J^1\pi$.
\label{1stdifeo}
\end{proposition}
\begin{proof}
$\mathcal{P}$ is a closed submanifold of $J^2\pi^\ddagger$ 
since it is defined by the constraints
$$
p^{\alpha\beta,\mu\nu}-\frac{\partial \hat L}{\partial g_{\alpha\beta,\mu\nu}}=0;\quad 
p^{\alpha\beta,\mu}-\hat L^{\alpha\beta,\mu}=0\ .
$$
The dimension of $\mathcal{P}$ is $4+10+40=54$ and,
as ${\rm rank}(\Tan\mathcal{FL}_{\mathfrak V})=54$ in every point, 
$\Tan\mathcal{FL}_{\mathfrak V}$ is surjective and $\mathcal{FL}_{\mathfrak V}$ is a submersion.
Finally, bearing in mind \eqref{kerfl}, 
we conclude that the fibers of the Legendre map,
$\mathcal{FL}_{\mathfrak V}^{-1}(\mathcal{FL}_{\mathfrak V}(j^3_x\phi))$
(for every $j^3_x\phi\in J^3\pi$),
are just the fibers of the projection $\bar\pi^3_1$, which are connected submanifolds of $J^3\pi$. Recall that $J^3\pi$ is connected because we are considering metrics with fixed signature.
Thus, $\mathcal{L}_{\mathfrak V}$ is an almost-regular Lagrangian.

Furthermore,
taking any local section $\phi$ of the projection $\pi^3_1$,
the map $\Phi=\Leg_{\mathfrak V}\circ\phi\colon J^1\pi\to{\cal P}$
is a local diffeomorphism (which does not depend on the section chosen).
Then, using these local sections, from a differentiable structure of $J^1\pi$ we can construct
a differentiable structure for ${\cal P}$; 
hence ${\cal P}$ and $J^1\pi$ are diffeomorphic.
$$
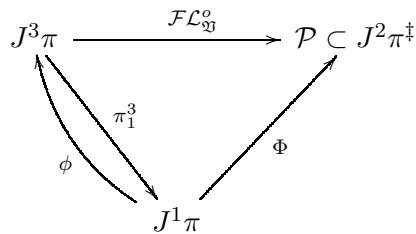
\xymatrix{
 \ J^3\pi \ar@/_0pc/[rr]^{\Leg_{\mathfrak V}^o} \ar@/_0pc/[ddr]^{\pi^3_1} 
\ & \ & \ 
\mathcal{P}\subset J^2\pi^\ddagger 
 \\  \ & \ & \ & \ &  \\
& \ J^1\pi \  \ar@/^1pc/[uul]^{\phi} \ar@/_0pc/[uur]_{\Phi}& 
}
$$
\end{proof}

Then, it can be proved that the $\mu$-transverse submanifolds
$\widetilde{\cal P}$ and ${\cal P}$ are diffeomorphic
and the diffeomorphism, denoted
 $\tilde\mu\colon\tilde{\cal P}\to{\cal P}$,
is just the restriction of the projection $\mu$ to $\tilde{\cal P}$.
Therefore we can define a
{\sl Hamiltonian $\mu$-section} as $h_{\mathfrak V} = \tilde{\jmath} \circ \widetilde{\mu}^{-1}$, 
which  is specified by a local Hamiltonian function
$H_{\mathcal{P}} \in C^\infty(\mathcal{P})$; that is,
$h_{\mathfrak V}(x^\mu,g_{\alpha\beta},g_{\alpha\beta,\mu},p^{\alpha\beta,\mu},p^{\alpha\beta,\mu\nu})= (x^\mu,g_{\alpha\beta},g_{\alpha\beta,\mu},-H_{\mathcal{P}},p^{\alpha\beta,\mu},p^{\alpha\beta,\mu\nu})$.
This function $H_{\mathfrak V}$ is the Hamiltonian function defined on 
$\mathcal{P}$ and is given by
$H=(\Leg_{\mathfrak V}^o)^*\,H_{\mathfrak V}$;
where $H$ is the function defined in \eqref{Hamfun},
which is $\Leg_{\mathfrak V}^o$-projectable.
$$
\xymatrix{
\widetilde{\mathcal{P}} \ar[rr]^-{\tilde{\jmath}} \ar[d]^-{\tilde{\mu}} & \ & J^{2}\pi^\dagger \ar[d]^-{\mu} & \ & \mathcal{W} \ar[d]_-{\mu_\mathcal{W}} \ar[ll]_-{\rho_2} \\
\mathcal{P} \ar[rr]^-{\jmath} \ar[urr]^-{h_{\mathfrak V}}& \ & J^{2}\pi^\ddagger & \ & \mathcal{W}_r \ar@/_0.7pc/[u]_{\hat{h}} \ar[ll]_-{\rho_2^r}
}
$$
Now, we can define the Hamiltonian forms
$$
\Theta_{h_{\mathfrak V}}:=h_{\mathfrak V}^*\Theta_1^s \in\df^4({\cal P}) \quad , \quad
\Omega_{h_{\mathfrak V}}:=-\d\Theta_{h_{\mathfrak V}}=h_{\mathfrak V}^*\Omega_1^s \in\df^{5}({\cal P}) \ ,
$$
and thus we have the Hamiltonian system $({\cal P},\Omega_{h_{\mathfrak V}})$.
Then, the \textsl{Hamiltonian problem} associated with this system
consists in finding holonomic sections $\psi_h\colon M\rightarrow\mathcal{P}$ 
satisfying any of the following equivalent conditions:
\begin{enumerate}
\item $\psi_h$ is a solution to the equation
\beq{}\label{eq:hsec}
\psi_h^*i(X)\Omega_{h_{\mathfrak V}}=0 \quad, \quad 
\mbox{\rm for every $X \in\mathfrak{X}(\mathcal{P})$}\ .
\eeq
\item $\psi_h$ is an integral section of a multivector field contained 
in a class of holonomic multivector fields 
$\{{\bf X}_h\}\subset\mathfrak{X}^4(\mathcal{P})$
satisfying the equation
\beq{}\label{eq:hvf}
\inn{({\bf X}_h)}\Omega_{h_{\mathfrak V}}=0 \quad ,\quad
\forall {\bf X}_h\in\{{\bf X}_h\}\subset\mathfrak{X}^4(\mathcal{P}) \ .
\eeq
\end{enumerate}
(Here, holonomic sections and multivector fields are defined as in $J^2\pi^\dagger$).
Then the Hamiltonian formalism is recovered as follows:

\begin{proposition}
Let $\psi \in \Gamma(\rho_M^r)$
be a solution to the equation \eqref{eqn:UnifFieldEqSect}. Then, the section
$\psi_h = \mathcal{FL}_{\mathfrak V}^o \circ \rho_1^r \circ \psi = \mathcal{FL}_{\mathfrak V}^o \circ \psi_\mathcal{L} \in \Gamma(\bar{\pi}_{\mathcal{P}})$ is a solution
to the equation
$$
\psi_h^*\inn(X)\Omega_{h_{\mathfrak V}} = 0 \, , \quad 
\mbox{for every } X \in \mathfrak{X}(\mathcal{P}) \, .
$$
$$
\xymatrix{
\ & \ & \mathcal{W}_r 
\ar[rr]^{\rho_2}
\ar[dll]_{\rho^r_1} \ar[drr]^{\rho^r_2}
 \ar[ddd]^<(0.4){\rho^r_M}
\ar[ddrr]_<(0.20){\rho^r_{\cal P}}
& & \ J^{2}\pi^\dagger \ar[d]^<(0.4){\mu}\\
J^3\pi \ar[ddrr]^{\bar{\pi}^3}
 \ar[rrrr]^<(0.20){\mathcal{FL}_{\mathfrak V}}|(.425){\hole}|(.49){\hole} \ar[rrrrd]_<(0.65){\mathcal{FL}_{\mathfrak V}^o}|(.42){\hole}|(.495){\hole} & \ & \ & \ & J^2\pi^\ddagger \\
\ & \ & \ & \ & \mathcal{P} 
\ar@{^{(}->}[u]_{\jmath} \ar[dll]_{\bar{\pi}_{\mathcal{P}}} \\
\ & \ & M \ar@/^1pc/[uull]^{\psi_\mathcal{L}} 
\ar@/_1pc/@{->}[urr]_<(0.65){\psi_h = \mathcal{FL}_{\mathfrak V}^o \circ \psi_\mathcal{L}} \ar@/^1pc/[uuu]^<(0.3){\psi} 
& \ & \
}
$$
\label{Hamprop}
\end{proposition}

\begin{theorem}
\label{Hamteor}
Let $\mathbf{X} \in \mathfrak{X}^4(\mathcal{W}_r)$ be a holonomic multivector field which is a solution to the equation \eqref{eqn:UnifDynEqMultiVF}, at least on the points of a submanifold
$\jmath_f\colon\mathcal{W}_f \subseteq\mathcal{W}_{\mathcal{L}_{\mathfrak V}}\hookrightarrow \mathcal{W}_r$,
and tangent to $\mathcal{W}_f$.
Then there exists a holonomic multivector field
$\mathbf{X}_h \in \mathfrak{X}^4(\mathcal{P})$ which is a solution to the following equation,
at least on the points of $P_f= \mathcal{FL}_{\mathfrak V}(S_f)$, and tangent to $P_f$,
\begin{equation}\label{eqn:HamDynEqMultiVFSing}
\restric{\inn(\mathbf{X}_h)\Omega_{h_{\mathfrak V}}}{P_f} = 0 \, .
\end{equation}
Conversely, if $\mathbf{X}_h \in \mathfrak{X}^4(\mathcal{P})$ is a holonomic multivector field
which is a solution to the equation \eqref{eqn:HamDynEqMultiVFSing},
at least on a submanifold $P_f\hookrightarrow{\cal P}$,
and tangent to $P_f$;
then there exist locally decomposable, $\rho_M^r$-transverse and integrable multivector fields
$\mathbf{X} \in \mathfrak{X}^4(\mathcal{W}_r)$ 
which are solutions to the equation \eqref{eqn:UnifDynEqMultiVF},
at least on the points of $\mathcal{W}_f=(\rho_2^{\mathfrak V})^{-1}(P_f)
\hookrightarrow \mathcal{W}_{\Lag_{\mathfrak V}}\hookrightarrow \mathcal{W}_r$,
and tangent to $\mathcal{W}_f$.

If $\mathbf{X}$ is $\rho_{\cal P}^r$-projectable
(or, what is equivalent, if the multivector field ${\bf X}_\Lag$
in Theorem \ref{thm:UnifToLagMultiVF} is $\Leg_{\mathfrak V}^o$-projectable),
then the relation between these multivector fields is \
$\mathbf{X}_h \circ\rho_{\cal P}^r\circ\jmath_f= \Lambda^4\Tan\rho_{\cal P}^r\circ \mathbf{X}\circ\jmath_f$.
\end{theorem}
\beq
\xymatrix{
\ & \ & \mathcal{W}_r 
\ar@/_1.3pc/[ddll]_{\rho_1^r}
\ar@/^4.9pc/[dddrr]^{\rho_{\cal P}^r} 
\ar@/^1.3pc/[ddrr]^{\rho_2^r} 
& \ & \ \\ \ & \ & \mathcal{W}_{\mathcal{L}_{\mathfrak V}}
\ar[dll]_{\rho_1^{\mathfrak V}}
\ar[ddrr]_<(0.53){\rho^{\mathfrak V}_{\cal P}} 
\ar[drr]^<(0.45){\rho_2^{\mathfrak V}} \ar@{^{(}->}[u]^{j_{\mathcal{L}_{\mathfrak V}}} & \ & \ \\
J^3\pi \ar[rrrr]^<(0.55){\mathcal{FL}_{\mathfrak V}}|(.495){\hole} \ar[drrrr]_<(0.35){\mathcal{FL}_{\mathfrak V}^o}|(.495){\hole} & \ & \ & \ & J^2\pi^\ddagger \\
\ & \ &  \mathcal{W}_f \ar@{^{(}->}[uu] \ar[dll] \ar[drr]  & \ & \mathcal{P} \ar@{^{(}->}[u]^-{\jmath} \\
S_f \ar@{^{(}->}[uu] & \ & \ & \ & P_f \ar@{^{(}->}[u] \\
}
\label{maindiag}
\eeq

{\bf Formulation using non multimomentum coordinates.}

From the unified formalism, the easiest way to describe locally
the Hamiltonian formalism consists in taking
$(x^\mu,g_{\alpha\beta},g_{\alpha\beta,\mu})$ 
as local coordinates adapted to $\mathcal{P}$.
As the function $H$ defined in \eqref{Hamfun} is $\Leg_{\mathfrak V}^o$-projectable,
the Hamiltonian function defined on $\mathcal{P}$ is just
\beq\label{hamiltonianexp}
H_{\mathfrak V}=
\sum_{\substack{\alpha\leq \beta}}L^{\alpha\beta,\mu\nu}g_{\alpha\beta,\mu\nu}+
\sum_{\substack{\alpha\leq \beta}}L^{\alpha\beta,\mu}g_{\alpha\beta,\mu}-L=
\varrho\, g_{\alpha\beta,\mu}g_{kl,\nu}H^{\alpha\beta k l \mu\nu}\ ,
\eeq
where 
$H^{\alpha\beta k l\mu\nu}$ is given by \eqref{Hindex}.
As $\mathcal{L}_{\mathfrak V}$ is almost regular, the Hamiltonian section 
$h_{\mathfrak V}\colon\mathcal{P}\rightarrow J^2\pi^\dagger$ exists
and its local expression is
$$
h_{\mathfrak V}(x^\mu,g_{\alpha\beta},g_{\alpha\beta,\mu})=(x^\mu,g_{\alpha\beta},g_{\alpha\beta,\mu},-H_{\mathfrak V},L^{\alpha\beta,\mu},L^{\alpha\beta,\mu\nu})\ .
$$
Now we define the Hamilton-Cartan forms
$\Theta_{h_{\mathfrak V}}=h_{\mathfrak V}^*\Theta_1^s\in\Omega^4(\mathcal{P})$ 
and $\Omega_{h_{\mathfrak V}}=-\d\Theta_{h_{\mathfrak V}}\in\Omega^5(\mathcal{P})$, whose coordinate expressions are
\bea\nonumber
\Theta_{h_{\mathfrak V}}&=&-H_{\mathfrak V}\,\d^4x+\sum_{\alpha\leq\beta}L^{\alpha\beta,\mu}dg_{\alpha\beta}\wedge \d^3x_\mu+\sum_{\alpha\leq\beta}L^{\alpha\beta,\mu\nu}dg_{\alpha\beta,\mu}\wedge \d^3x_{\nu} \ ,
\\ \label{HamiltonianForm}
 \Omega_{h_{\mathfrak V}}&=&-\d\Theta_{h_{\mathfrak V}} \nonumber \\ &=&
\d H_{\mathfrak V}\wedge\d^4x-\sum_{\alpha\leq\beta}\d L^{\alpha\beta,\mu}\wedge\d g_{\alpha\beta}\wedge \d^3x_\mu-\sum_{\alpha\leq\beta}\d L^{\alpha\beta,\mu\nu}\wedge\d g_{\alpha\beta,\mu}\wedge \d^3x_{\nu}\ .
\eea
(Observe that,
with this choice of coordinates,$\Theta_{h_{\mathfrak V}}$ and $\Omega_{h_{\mathfrak V}}$ looks  locally like
$\Theta_{\mathcal{L}_{\mathfrak V}}$ and $\Omega_{\mathcal{L}_{\mathfrak V}}$).
Thus, we have the Hamiltonian system $(\mathcal{P},\Omega_{h_{\mathfrak V}})$.
Then, Proposition \ref{Hamprop} and Theorem  \ref{Hamteor} establish 
the relation between the solutions to the Hamiltonian and the unified problem.

In this case, first observe that, locally,
$$
\ker\,(\pi^r_{\cal P})_*=
{\left<
\derpar{}{p^{\alpha\beta,\mu}},\derpar{}{p^{\alpha\beta,\mu\nu}};
\derpar{}{g_{\alpha\beta,\mu\nu}},\derpar{}{g_{\alpha\beta,\mu\nu\lambda}}\right>}_{0\leq\alpha\leq\beta\leq3;\, 0\leq\mu\leq\nu\leq\lambda\leq3}\ ,
$$
and as
$$
\Lie\left(\derpar{}{g_{\alpha\beta,\mu\nu}}\right)\hat L^{\lambda\sigma}\not=0
\ , \
\Lie\left(\derpar{}{g_{\alpha\beta,\mu\nu}}\right)(D_\tau\hat L^{\lambda\sigma})\not=0
\ , \
\Lie\left(\derpar{}{g_{\alpha\beta,\mu\nu}}\right)(D_\tau\hat L^{\lambda\sigma})\not=0 \ ,
$$
we have that  the constraints  \eqref{eqn:Res3} and \eqref{eqn:Res4}
(which define the final constraint submanifold $\W_f$
as a submanifold of $\W_{\Lag_{\mathfrak V}}={\rm graph}\,\Leg_{\mathfrak V}$
in the unified formalism)
are not $\rho^r_{\cal P}$-projectable (see diagram \eqref{maindiag}),
and this means that there are no Hamiltonian constraints and the Hamilton equations have solutions
everywhere in ${\cal P}$.
(What is equivalent, the Lagrangian constraints 
\eqref{Einseqagain} and \eqref{Einseqagain2}
are not $\Leg_{\mathfrak V}^o$-projectable). 
This is a consequence of the fact that, in the Lagrangian formalism, 
these constraints really  arise  as a consequence of demanding the holonomy condition 
and hence, as it was studied in \cite{LMMMR-2002}, they are not projectable by the Legendre map.
Then:

\begin{proposition}
An integrable (holonomic) multivector field solution to the equations \eqref{eq:hvf} is
$$
{\bf X}_h=\bigwedge_{\nu=0}^3\left(\frac {\partial}{\partial x^\nu}+\sum_{\alpha\leq\beta}\left(g_{\alpha\beta,\nu}\frac{\partial}{\partial g_{\alpha\beta}}+g_{\lambda\sigma}(\Gamma_{\nu \alpha }^\lambda\Gamma_{\mu \beta}^\sigma+\Gamma_{\nu \beta}^\lambda\Gamma_{\mu \alpha }^\sigma)\frac{\partial}{\partial g_{\alpha\beta;\mu}}\right)\right)\in\vf^4({\cal P})\ .
$$
\end{proposition}
\proof
The proof is given in the appendix \ref{solutions}.
\qed

For the integral sections of ${\bf X}_h$, which are solutions to \eqref{eq:hsec},
if $\psi(x^\alpha)=(x^\alpha,\psi_{\alpha\beta}(x^\alpha),\psi_{\alpha\beta,\mu}(x^\alpha))$,
then the equation \eqref{eq:hsec} reads
\beann
\left.\left(D_\mu L^{\alpha\beta,\mu}-\frac{\partial L}{\partial g_{\alpha\beta}}\right)\right|_{\psi}&=&0 \ ,\\
\left.\left(\frac{\partial L^{\alpha\beta,\mu\nu}}{\partial g_{\lambda\sigma}}- \frac{\partial L^{\lambda\sigma,\nu}}{\partial g_{\alpha\beta,\mu}}\right)g_{\lambda\sigma,\nu}\right|_{\psi}&=&\left(\frac{\partial L^{\alpha\beta,\mu\nu}}{\partial g_{\lambda\sigma}}- \frac{\partial L^{\lambda\sigma,\nu}}{\partial g_{\alpha\beta,\mu}}\right)\frac{\partial \psi_{\lambda\sigma}}{\partial x^\nu} \ .
\eeann
The last equation is equivalent to the holonomy condition,
$\displaystyle\frac{\partial \psi_{\lambda\sigma}}{\partial x^\nu}=\psi_{\lambda\sigma,\nu}$ 
(see the appendix \ref{solutions}). Writing the first one in terms of the Hamiltonian we obtain
$$
\left.\left(\frac{\partial L^{\alpha\beta,\nu}}{\partial g_{ab,\mu}}-\frac{\partial L^{ab,\mu\nu}}{\partial g_{\alpha\beta}}\right)\right|_{\psi}\frac{\partial \psi_{ab,\mu}}{\partial x^\nu}=\left.-\frac{\partial H_{\cal P}}{\partial g_{\alpha\beta}}\right|_{\psi}-\left.\psi_{ab,\mu}\left(\frac{\partial L^{\alpha\beta,\mu}}{\partial g_{ab}}-\frac{\partial L^{ab,\mu}}{\partial g_{\alpha\beta}}\right)\right|_{\psi}\ .
$$
And rearranging the terms, these equations are equivalent to the Einstein equations \eqref{EEq}.

{\bf Formulation using multimomentum coordinates.}

As we have said, the coordinates 
$(x^\mu,g_{\alpha\beta},g_{\alpha\beta,\mu})$ 
arise naturally from the unified formalism. 
Nevertheless, the standard way to describe locally the Hamiltonian formalism
of classical field theory consists in using the natural coordinates in the multimomentum phase spaces;
that is, multimomentum coordinates.
Then, the first relevant result is:

\begin{proposition}The coordinates $p^{\alpha\beta,\mu}$ and $g_{\alpha\beta,\mu}$ are in one-to-one correspondence.
\end{proposition}
\begin{proof}
The starting point is to consider the constraints
$p^{\alpha\beta,\mu}=L^{\alpha\beta,\mu}(x^\mu,g_{\alpha\beta},g_{\alpha\beta,\mu})$
which define partially the constraint submanifold $\W_{\Lag_{\mathfrak V}}$,
and from these relations we can isolate the coordinates $g_{\alpha\beta,\mu}$. Indeed, the functions
\beann
V_{\alpha\beta,\mu}(g_{\alpha\beta},p^{\alpha\beta,\mu})&=&
\frac{p^{\lambda\sigma,\nu}}{3\varrho n(\alpha\beta)}(-2g_{\alpha \lambda}g_{\beta\mu}g_{\sigma\nu}-2g_{\alpha\mu}g_{\beta \lambda}g_{\sigma\nu}+6g_{\alpha \lambda}g_{\beta \sigma}g_{\mu \nu}+
\\ & &
\qquad\qquad \quad g_{\alpha \nu}g_{\beta\mu}g_{\lambda\sigma}+g_{\alpha\mu}g_{\beta \nu}g_{\lambda\sigma})
\eeann
satisfy that
$$
g_{\alpha\beta,\mu}=V_{\alpha\beta,\mu}(g_{\alpha\beta},L^{\lambda\sigma,\nu}(g_{\alpha\beta},g_{\alpha\beta,\mu})) \ ,
$$
and these relations give the coordinates $g_{\alpha\beta,\mu}$
as functions of $p^{\lambda\sigma,\nu}$ and the other coordinates.
\end{proof}

Thus we can use $(x^\mu,g_{\alpha\beta},p^{\alpha\beta,\mu})$ 
as coordinates of $\mathcal{P}$ and then rewrite the Hamiltonian function
$$
H_{\mathfrak V}(x^\mu,g_{\alpha\beta},p^{\alpha\beta,\mu})=
H_{\mathfrak V}(x^\mu,g_{\alpha\beta},V_{\alpha\beta,\mu}(p^{\alpha\beta,\mu},g_{\alpha\beta})) \ .
$$
The field equations are derived again from \eqref{eq:hvf} expressed using the new coordinates.
Now, the Hamilton-Cartan form $\Omega_h$ has the local expression:
$$
\Omega_{h_{\mathfrak V}}=\d H_{\mathfrak V}\wedge \d^4x-\sum_{\alpha\leq\beta}\d p^{\alpha\beta,\mu}\wedge\d g_{\alpha\beta}\wedge\d^3x_\mu-\sum_{\alpha\leq\beta}\d L^{\alpha\beta,\mu\nu}\wedge\d V_{\alpha\beta,\mu}\wedge\d^3x_\nu \ ,
$$
and the local expression of a representative of a class $\{{\bf X}_h\}$ of 
semi-holonomic multivector fields in ${\cal P}$ is
$$
{\bf X}_h=\bigwedge_{i=\nu}^4\left(\frac{\partial}{\partial x^\nu}+F_{\alpha\beta,\nu}\frac{\partial}{\partial g_{\alpha\beta}}+G^{\alpha\beta,\mu}_\nu\frac{\partial}{\partial p^{\alpha\beta,\mu}}\right);.
$$
with $F_{\alpha\beta,\nu}(x^\mu,g_{\alpha\beta},p^{\alpha\beta,\mu}),G^{\alpha\beta,\mu}_\nu(x^\mu,g_{\alpha\beta},p^{\alpha\beta,\mu})\in C^\infty(\mathcal{P})$.
From \eqref{eq:hvf} we obtain
\bea\nonumber
\frac{\partial H_{\mathfrak V}}{\partial g_{\alpha\beta}}&=&-G ^{\alpha\beta,\mu}_\mu+G^{rs,k}_\nu\frac{\partial V_{ab,c}}{\partial p^{rs,k}}\frac{\partial L^{ab,c\nu}}{\partial g_{\alpha\beta}}+F_{rs,\nu}\left(\frac{\partial V_{ab,c}}{\partial g_{rs}}\frac{\partial L^{ab,c\nu}}{\partial g_{\alpha\beta}}-\frac{\partial V_{ab,c}}{\partial g_{\alpha\beta}}\frac{\partial L^{ab,c\nu}}{\partial g_{rs}}\right)
\\\nonumber
\frac{\partial H_{\mathfrak V}}{\partial p^{\alpha\beta,\mu}}&=&F_{\alpha\beta,\mu}-F_{rs,\nu}\frac{\partial V_{ab,c}}{\partial p^{\alpha\beta,\mu}}\frac{\partial L^{ab,c\nu}}{\partial g_{rs}} \ ,
\eea
which would be the classical Hamilton-De Donder-Weil equations
for a first order field theory except by the fact that they contain 
extra-terms because the Hilbert-Einstein Lagrangian is of second order and 
$\displaystyle L^{\alpha\beta,\mu\nu}=\frac{1}{n(\mu\nu)}\frac{\partial L}{\partial g_{\alpha\beta,\mu\nu}}$
does not vanish.
A solution to these equations is
$$
{\bf X}_h=\bigwedge_{i=\nu}^4\left(\frac{\partial}{\partial x^\nu}+V_{\alpha\beta,\mu}\frac{\partial}{\partial g_{\alpha\beta}}+g_{rs}(\Gamma^r_{\nu\lambda}\Gamma^s_{\mu\sigma}+\Gamma^r_{\nu\sigma}\Gamma^s_{\mu\lambda})\frac{\partial V_{\alpha\beta\mu}}{\partial g_{\lambda\sigma,\gamma}}\frac{\partial}{\partial p^{\alpha\beta,\mu}}\right)\ ,
$$
where the velocities in the connection are expressed using the momenta, which is a holonomic (i.e., integrable) multivector field in ${\cal P}$.

Finally, we consider the equations of the integral sections
of ${\bf X}_h$.
These equations can be obtained from equation \eqref{eq:hsec} which,
for a section $\psi(x^\alpha)=(x^\alpha,\psi_{\alpha\beta}(x^\alpha),\psi^{\alpha\beta,\mu}(x^\alpha))$, 
leads to
\beann
\left.\frac{\partial H_{\mathfrak V}}{\partial g_{\alpha\beta}}\right|_\psi&=&\frac{\partial \psi^{\alpha\beta,\mu}}{\partial x^\mu}+\frac{\partial \psi^{rs,k}}{\partial x^\nu}\left.\left(\frac{\partial V_{ab,c}}{\partial p^{rs,k}}\frac{\partial L^{ab,c\nu}}{\partial g_{\alpha\beta}}\right)\right|_\psi+\frac{\partial \psi_{rs}}{\partial x^\nu}\left.\left(\frac{\partial V_{ab,c}}{\partial g_{rs}}\frac{\partial L^{ab,c\nu}}{\partial g_{\alpha\beta}}-\frac{\partial V_{ab,c}}{\partial g_{\alpha\beta}}\frac{\partial L^{ab,c\nu}}{\partial g_{rs}}\right)\right|_\psi
\\
\left.\frac{\partial H_{\mathfrak V}}{\partial p^{\alpha\beta,\mu}}\right|_\psi&=&\frac{\partial\psi_{\alpha\beta}}{\partial x^\mu}-\frac{\partial\psi_{rs}}{\partial x^\nu}\left.\left(\frac{\partial V_{ab,c}}{\partial p^{\alpha\beta,\mu}}\frac{\partial L^{ab,c\nu}}{\partial g_{rs}}\right)\right|_\psi \ .
\eeann

\section{An equivalent first-order Lagrangian to Hilbert-Einstein}
\label{compare}

As we pointed out at the end of Section \ref{lhuni},
there exists a first-order Lagrangian equivalent to the 
Hilbert-Einstein Lagrangian \cite{first,rosado}. 
Now we study the Lagrangian and the Hamiltonian formalism of this model, comparing them with the Hamiltonian formulations for the Hilbert-Einstein Lagrangian presented in the above section.
As it is a first order Lagrangian, we need to use the multisymplectic formalisms
developed for these kind of theories; in particular,
those reviewed in \cite{art:Roman09}.

The configuration manifold $\pi: E\rightarrow M$, 
is the same described in Section \ref{presta}, and the
Lagrangian formalisms takes place in the first jet bundle 
$J^1\pi$, with coordinates $(x^\mu,g_{\alpha\beta},g_{\alpha\beta,\mu})$. 
The first-order Lagrangian density proposed in \cite{rosado} is
$\overline\Lag=\overline{L}\,\d^4x$, where the Lagrangian function is
\beq
\overline{L}=L_0-\sum_{\substack{\alpha\leq\beta\\\lambda\leq\sigma}}g_{\alpha\beta,\mu}g_{\lambda \sigma,\nu}\derpar{L^{\alpha\beta,\mu\nu}}{g_{\lambda\sigma}}
\in\Cinfty(J^1\pi) \ .
\label{barL}
\eeq
The Poincar\'e-Cartan form for this Lagrangian is
\beq\label{FirstOrderForm}
\Omega_{\overline\Lag}=\d \overline{L}\wedge\d^4x-\sum_{\alpha\leq\beta}\d \frac{\partial \overline{L}}{\partial g_{\alpha\beta,\mu}}\wedge\d g_{\alpha\beta}\wedge\d^3x_\mu \ .
\eeq
The Lagrangian $\overline{L}$ is regular and hence
$\Omega_{\overline\Lag}$ is a multisymplectic form.
For the Lagrangian system $(J^1\pi,\Omega_{\overline\Lag})$ 
we look for solutions to the equations \eqref{eqn:LagDynEqSect} or \eqref{eqn:LagDynEqMultiVF} and,
as the system is regular, solutions exist everywhere in $J^1\pi$
(there are no Lagrangian constraints).
Although it is a first order system, in \cite{rosado} it is shown how these equations coincide with the Einstein equations.

As $\overline{L}$ is regular, we can state the standard Hamiltonian formalism for first-order regular field theories.
Being $J^1\pi^*$ the (``first-order'') reduced multimomentum bundle, whose natural coordinates are
$(x^\mu,g_{\alpha\beta},\overline{p}^{\alpha\beta,\mu})$,
the corresponding Legendre map
$\mathcal{F\overline{L}}\colon J^1\pi\to J^1\pi^*$
is given by
$$
\mathcal{F\overline{L}}^*x^\mu=x^\mu\quad ,\quad
\mathcal{F\overline{L}}^*g_{\alpha\beta}=g_{\alpha\beta}\quad ,\quad
\mathcal{F\overline{L}}^*\overline{p}^{\alpha\beta,\mu}=
\frac{\partial\overline{L}}{\partial g_{\alpha\beta,\mu}}=
L^{\alpha\beta,\mu}-\sum_{\substack{\lambda\leq\sigma}}g_{\lambda\sigma,\nu}\frac{\partial L^{\lambda\sigma,\nu\mu}}{\partial g_{\alpha\beta}} \ .
$$

Then we have the Hamilton-Cartan form 
$\Omega_{\overline h}:=((\mathcal{F\overline L})^{-1})^*\Omega_{\overline\Lag}\in\df^4(J^1\pi^*)$. 
This multisymplectic form can also be obtained introducing
the Hamiltonian section 
$\overline{h}\colon J^1\pi^*\rightarrow \Lambda_2^4(E)$ whose local expression is
$$
\overline{h}(x^\mu,g_{\alpha\beta},\overline{p}^{\alpha\beta,\mu})=(x^\mu,g_{\alpha\beta},-\overline{H},\overline{p}^{\alpha\beta,\mu})\ .
$$
where $\overline{H}$ is
the Hamiltonian function associated with $\overline{L}$, whose local expression is
$$
\overline{H}=
\sum_{\alpha\leq\beta}\overline{p}^{\alpha\beta,\mu}(g_{\alpha\beta,\mu}\circ{\mathcal{F\overline{L}}^{-1}})-\overline{L}\circ({\mathcal{F\overline{L}})^{-1}}=
\overline{L}\circ({\mathcal{F\overline{L}})^{-1}} \ .
$$
In this way, we have constructed the Hamiltonian system $(J^1\pi^*,\Omega_{\overline h})$
and the corresponding Hamilton field equations have solutions everywhere in $J^1\pi^*$
(there are no Hamiltonian constraints).
Furthermore,
as $\mathcal{F\overline{L}}$ is a diffeomorphism,
every solution to the Lagrangian problem stated for the Lagrangian system 
$(J^1\pi,\Omega_{\overline\Lag})$ induces a solution to the Hamiltonian problem 
stated for the Hamiltonian system $(J^1\pi^*,\Omega_{\overline h})$
via this Legendre map, and conversely.

The following result relates this approach to the one we have presented in the above section.

\begin{proposition}
$\Phi^*H_{\mathfrak V}=\overline{L}$ and, as a consequence,
$\Phi^*\Omega_{h_{\mathfrak V}}=\Omega_{\overline\Lag}$.
\end{proposition}
\begin{proof}
In order to prove these equalities, it suffices to prove that,
$H_{\mathfrak V}$ and $\Omega_{h_{\mathfrak V}}$ have the same local coordinate expressions than $\overline{L}$ and $\Omega_{\overline\Lag}$, respectively.

First, from \eqref{hamiltonianexp}, using \eqref{barL} and
taking into account the coordinate expressions stated 
in \eqref{L0}, \eqref{L1}, \eqref{L2}, and \eqref{L3}, we obtain that
\beann
H_{\mathfrak V}&=&
\sum_{\substack{\alpha\leq \beta\\\mu\leq \nu}}L^{\alpha\beta,\mu\nu}g_{\alpha\beta,\mu\nu}+
\sum_{\substack{\alpha\leq \beta}}L^{\alpha\beta,\mu}g_{\alpha\beta,\mu}-L=\sum_{\substack{\alpha\leq \beta}}\left(\frac{\partial L_0}{\partial g_{\alpha\beta,\mu}}-D_v L^{\alpha\beta,\mu\nu}\right)g_{\alpha\beta,\mu}-L_0
\\
&=&2L_0-\sum_{\substack{\alpha\leq \beta\\\lambda\leq\sigma}}g_{\alpha\beta,\mu}g_{\lambda\sigma,\nu}\frac{\partial L^{\alpha\beta,\mu\nu}}{\partial g_{\lambda\sigma}}-L_0=\overline{L} \ .
\eeann
We have used that $\displaystyle\frac{\partial L_0}{\partial g_{\alpha\beta,\mu}}g_{\alpha\beta,\mu}=2L_0$, which is a consequence of $L_0$ being homogeneous of degree 2 on the velocities. Now we compute 
$$
\frac{\partial\overline{L}}{\partial g_{\alpha\beta,\mu}}=\frac{\partial L_0}{\partial g_{\alpha\beta,\mu}}-\sum_{\lambda\leq\sigma}g_{\lambda\sigma,\nu}\left(\frac{\partial L^{\alpha\beta,\mu\nu}}{\partial g_{\lambda\sigma}}+\frac{\partial L^{\lambda\sigma,\nu\mu}}{\partial g_{\alpha\beta}}\right)=L^{\alpha\beta,\mu}-g_{\lambda\sigma,\nu}\frac{\partial L^{\lambda\sigma,\nu\mu}}{\partial g_{\alpha\beta}} \ ;
$$
then, using these last results and bearing in mind
\eqref{FirstOrderForm}
and \eqref{HamiltonianForm}, we have that
\beann
\Omega_{\overline\Lag}&=&\d \overline{L}\wedge\d^4x-\sum_{\alpha\leq\beta}\d \frac{\partial \overline{L}}{\partial g_{\alpha\beta,\mu}}\wedge\d g_{\alpha\beta}\wedge\d^3x_\mu
\\
&=&\d H_{\mathfrak V}\wedge\d^4x-\sum_{\alpha\leq\beta}\d L^{\alpha\beta,\mu}\wedge\d g_{\alpha\beta}\wedge\d^3x_\mu+\sum_{\substack{\alpha\leq\beta\\\lambda\leq\sigma}}\d \!\left(g_{\lambda\sigma,\nu}\frac{\partial L^{\lambda\sigma,\nu\mu}}{\partial g_{\alpha\beta}}\right)\wedge\d g_{\alpha\beta}\wedge\d^3x_\mu
\\
&=&\d H_{\mathfrak V}\wedge\d^4x-\sum_{\alpha\leq\beta}\d L^{\alpha\beta,\mu}\wedge\d g_{\alpha\beta}\wedge\d^3x_\mu+\sum_{\substack{\alpha\leq\beta\\\lambda\leq\sigma}}\frac{\partial L^{\lambda\sigma,\nu\mu}}{\partial g_{\alpha\beta}}\d g_{\lambda\sigma,\nu}\wedge\d g_{\alpha\beta}\wedge\d^3x_\mu
\\
&+&\sum_{\alpha\leq\beta}g_{\lambda\sigma,\nu}\frac{\partial^2 L^{\lambda\sigma,\nu\mu}}{\partial g_{\gamma\eta}\partial g_{\alpha\beta}}\d g_{\gamma\eta}\wedge\d g_{\alpha\beta}\wedge\d^3x_\mu \ .
\eeann
The last term vanishes because the coefficient is symmetric under the change of the indices
$\gamma,\lambda$ by $\alpha,\beta$, but the exterior product is skewsymmetric. 
Finally, notice that $ L^{\lambda\sigma,\nu\mu}$ do not contain derivatives of the metric, thus we can write
$$
\!\sum_{\substack{\alpha\leq\beta\\\lambda\leq\sigma}}\!\frac{\partial L^{\lambda\sigma,\nu\mu}}{\partial g_{\alpha\beta}}\d g_{\lambda\sigma,\nu}\wedge\d g_{\alpha\beta}\wedge\d^3x_\mu=-\!\sum_{\sigma\leq\lambda}\!\d L^{\lambda\sigma,\nu\mu}\wedge\d g_{\lambda\sigma,\nu}\wedge\d^3x_\mu \ ,
$$
and, therefore, we can conclude that
$\Omega_{\overline\Lag}$ and $\Omega_{h_{\mathfrak V}}$ have the same
local expression.
\end{proof}

As a consequence of this result,
the solutions to the Hamiltonian problem stated for the Hamiltonian system 
$({\cal P},\Omega_{\cal P})$ and
to the Lagrangian problem stated for the Lagrangian system
 $(J^1\pi,\Omega_{\overline\Lag})$
are in one-to-one correspondence by the map $\Phi$.

Observe that we have also the diffeomorphism
$\Psi=\Phi^{-1}\circ\mathcal{F\overline{L}}\colon{\cal P}\to J^1\pi^*$.
Therefore, the solutions to the Hamiltonian problems stated for the Hamiltonian systems 
$({\cal P},\Omega_{h_{\mathfrak V}})$ and $(J^1\pi^*,\Omega_{\overline h})$
are also one-to-one related by this map.

Summarizing, we have proved that the following formulations are equivalent:
$$
\xymatrix{
 \ & \ (J^1\pi^*,\Omega_{\overline h}) \ar@/_0pc/[rr]^-{\mathcal{F\overline\Lag}} \ & \ & \ 
(J^1\pi,\Omega_{\overline\Lag})  \ar@/_0pc/[ll] \ar@/_0pc/[rr]^-{\Phi}
\ & \ & \
(\mathcal{P},\Omega_{h_{\mathfrak V}})  \ar@/_0pc/[ll]
}
$$
(where, in the last case, we can use the local description using multimomentum coordinates or not).
Locally, this  equivalence means that all the formulations
lead to the same equations (Einstein's equations), 
up to a change of variables and, hence,
every solution in each formalism induces a solution in the others via the appropriate diffeomorphism.
The following diagram summarizes all the picture:
$$
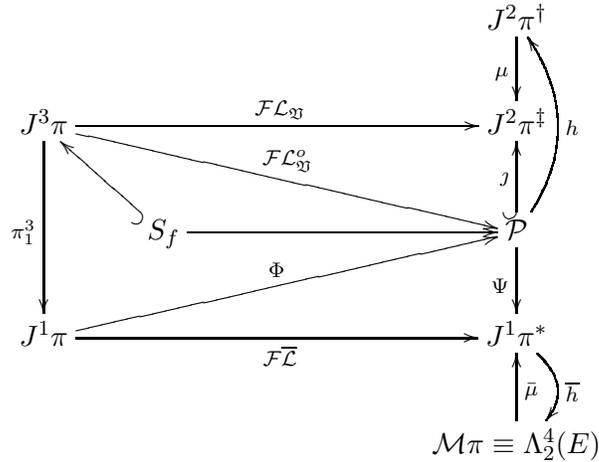
\xymatrix{
 & \ & \ & \ & J^2\pi^\dagger  \ar[d]_{\mu} \\
J^3\pi  \ar[dd]_{\pi^3_1}
\ar[rrrr]^{\mathcal{FL}_{\mathfrak V}} \ar[drrrr]^{\mathcal{FL}^o_{\mathfrak V}} & \ & \ & \ & J^2\pi^\ddagger \\
\ & S_f \ar@{^{(}->}[ul] 
\ar[rrr] & \ & \ & \mathcal{P} \ar@{^{(}->}[u]^{\jmath}
\ar@/_1.3pc/[uu]_{h}  \ar[d]_{\Psi}\\
J^1\pi  \ar[urrrr]^{\Phi} \ar[rrrr]_{\mathcal{F\overline{L}}} & \ & \ & \ & J^1\pi^* \ar@/^1.3pc/[d]^{\overline{h}} \\
 & \ & \ & \ & {\cal M}\pi\equiv\Lambda^4_2(E) \ar[u]_{\bar\mu}
}
$$

\section{The Einstein-Hilbert model with energy-matter sources}
\label{Sec4}

\subsection{Previous statements}

The Einstein-Hilbert model with energy-matter sources is described by a Lagrangian density
$\mathcal{L}_{\mathfrak S}=\mathcal{L}_{\mathfrak V}+\mathcal{L}_{\mathfrak m}$, where
$\mathcal{L}_{\mathfrak m}=L_{\mathfrak m}\,(\overline{\pi}^2)^*\eta\in\Omega^4(J^2\pi)$,
and $L_{\mathfrak m}\in C^\infty(J^2\pi)$ represents the energy-matter source and depends only 
on the metric and the first and second derivatives of its components.
It is related with the {\sl stress-energy-momentum tensor} ${\rm T}_{\mu\nu}$ by
$$
{\rm T}_{\mu\nu}=
\frac{c^4}{\varrho n(\mu\nu)8\pi G}\,g_{\alpha\mu}\,g_{\beta\nu}\,L_{\mathfrak m}^{\alpha\beta}\ .
$$
(For a geometric study on the stress-energy-momentum tensors see, for instance,
\cite{FGR,FoRo,GoMa,krupka2,voicu}).
Then, we can write 
$\mathcal{L}_{\mathfrak S}=L_{\mathfrak S}\,(\overline{\pi}^2)^*\eta\in\Omega^4(J^2\pi)$, 
with $L_{\mathfrak S}=L_{\mathfrak V}+L_{\mathfrak m}\in C^\infty(J^2\pi)$.

The behaviour of the theory depends on the source. 
Nevertheless, some qualitative properties can be studied in general, 
as long as we know the degeneracy of the source.

\begin{definition}
For a function $f\in C^\infty (J^2\pi)$, consider
$$
f^{\alpha\beta,\mu\nu}:=\frac{1}{n(\mu\nu)}\frac{\partial f}{\partial g_{\alpha\beta,\mu\nu}},\quad {f}^{\alpha\beta,\mu}:=\frac{\partial f}{\partial g_{\alpha\beta,\mu}}-D_\nu {f}^{\alpha\beta,\mu\nu},\quad \displaystyle {f}^{\alpha\beta}=\frac{\partial f}{\partial g_{\alpha\beta}}-D_\mu {f}^{\alpha\beta,\mu}.
$$
(Notice that $f^{\alpha\beta,\mu}\in C^\infty (J^3\pi)$ and
$f^{\alpha\beta}\in C^\infty (J^4\pi)$).
Then, the \emph{degree} of $f$ is the smallest natural number $\deg(f)=s$ such that:
$$
\Lie(X){f}^{\alpha\beta,\mu}=\Lie(X){f}^{\alpha\beta,\mu\nu}=0\quad ;\quad
\mbox{\rm for every $X\in \mathfrak{X}^V(\pi^3_{s-1})$\ ; \
($0\leq\alpha\leq\beta\leq3, 0\leq\mu\leq\nu\leq3$)} \ .
$$
If $f^{\alpha\beta,\mu}={f}^{\alpha\beta,\mu\nu}=0$, we define $\deg(f)=0$.
\end{definition}

Now, applying the results of \cite{art:GR-2016,rosado} we obtain that:

\begin{proposition}
If $\deg (f)=s$, then
$\Lie(X)f^{\alpha\beta}=0$;\, for every $X\in \mathfrak{X}^V(\pi^3_{s})$
($\alpha\leq\beta$), and hence $f^{\alpha\beta}$ are 
$\pi^3_{s}$-projectable functions.
\end{proposition}

The degree of $L_{\mathfrak S}$ characterizes partially the behaviour of the theory,
as we are going to see in the next paragraphs. 
For instance, if a Lagrangian is regular it has degree $4$, but there are also singular Lagrangians with degree $4$. The Hilbert-Einstein Lagrangian in vacuum, $L_{\mathfrak V}$, has degree $2$.
For a source such that $\deg(L_{\mathfrak m})>2$, we have that 
$\deg(L_{\mathfrak S})=\deg(L_{\mathfrak m})$. 
The so-called $f(R)$ theories of gravity have $\deg(L_{\mathfrak S})>2$. 
For these kinds of systems it is possible to obtain some constraints in the unified and the Lagrangian formalisms but the Hamiltonian formalism depends strongly on the particular  energy-matter source.
For a source such that $\deg(L_{\mathfrak m})\leq2$, we have that 
$\deg(L_{\mathfrak S})\leq2)$, and
these theories have a well defined Hamiltonian formalism;
in particular, for the case that $\deg(L_{\mathfrak m})\leq1$ we obtain the general semiholonomic solution. 
These cases include the energy-matter sources coupled only to the metric;
that is, $\deg(L_{\mathfrak m})=0$, like the electromagnetic source 
or the perfect fluid. We will present the former as an example.

\subsection{Lagrangian-Hamiltonian unified formalism}

As $L_{\mathfrak S}\in C^\infty(J^2\pi)$, we can work with the same manifolds introduced in
Section \ref{lhuni}; that is, the symmetric higher-order jet multimomentum bundles 
$\mathcal{W}=J^3\pi\times_{J^1\pi}J^2\pi^\ddagger$  and
$\W_r = J^{3}\pi \times_{J^{1}\pi} \, J^{2}\pi^\ddagger$. 
The pull-back of the Lagrangian to these manifolds is denoted in the same way as above: 
$\hat{L}_{\mathfrak S}= (\pi^3_2 \circ \rho_1^r)^*L_{\mathfrak S}\in C^\infty(\mathcal{W}_r)$ 
(or in $C^\infty(\mathcal{W})$). 
Then,
$$
\hat{H}_{\mathfrak S}=\sum_{\alpha\leq\beta}p^{\alpha\beta,\mu}g_{\alpha\beta,\mu}
+\sum_{\substack{\alpha\leq\beta\\\mu\leq\nu}}p^{\alpha\beta,\mu\nu}g_{\alpha\beta,\mu\nu}-\hat{L}_{\mathfrak S} \ ;
$$
The Liouville forms in $\W_r$,
$\Theta_{{\mathfrak S}r} $ and $\Omega_{{\mathfrak S}r}$, are defined likewise and have the local expressions
$$
\begin{array}{l}
\displaystyle
\Theta_{{\mathfrak S}r}=
-\hat{H}_{\mathfrak S}\d^4x+\sum_{\alpha\leq\beta}p^{\alpha\beta,\mu}\d g_{\alpha\beta}\wedge \d^{3}x_\mu
+\sum_{\alpha\leq\beta}\frac{1}{n(\mu\nu)}p^{\alpha\beta,\mu\nu}\d g_{\alpha\beta,\mu}\wedge \d^{3}x_{\nu}
 \ , \\[10pt]
\displaystyle
\Omega_{{\mathfrak S}r} =\d \hat{H}_{\mathfrak S} \wedge \d^4x-\sum_{\alpha\leq\beta}\d p^{\alpha\beta,\mu}\wedge \d g_{\alpha\beta}\wedge \d^{3}x_\mu-\sum_{\alpha\leq\beta}\frac{1}{n(\mu\nu)}\d p^{\alpha\beta,\mu\nu}\wedge \d g_{\alpha\beta,\mu}\wedge \d^{3}x_{\nu} \ .
\end{array}
$$

As in Section  \ref{lhuni}, the Lagrangian-Hamiltonian problem associated with the system 
$(\W_r,\Omega_{{\mathfrak S}r})$
consists in finding holonomic sections $\psi \in \Gamma(\rho_{\mathfrak m}^r)$ satisfying
any of the following equivalent conditions:
\begin{enumerate}
\item $\psi$ is a solution to the equation
\begin{equation}\label{eqn:UnifFieldEqSectMass}
\psi^*\inn(X)\Omega_{{\mathfrak S}r} = 0 \, , \quad \mbox{for every } X \in \mathfrak{X}(\mathcal{W}_r) \ .
\end{equation}
\item $\psi$ is an integral section of a multivector field contained in a class of holonomic multivector fields $\{ {\bf X} \} \subset \mathfrak{X}^4(\mathcal{W}_r)$
satisfying the equation
\begin{equation}\label{eqn:UnifDynEqMultiVFMass}
\inn({\bf X})\Omega_{{\mathfrak S}r} = 0 \ .
\end{equation}
\end{enumerate}

Proposition \ref{prop:GraphLegMapSect}, 
which defines the Legendre transformation, 
also holds for $\hat L_{\mathfrak S}$:

\begin{proposition}\label{prop:LegMass}
A section $\psi \in \Gamma(\rho_M^r)$ solution to the equation \eqref{eqn:UnifFieldEqSectMass} 
takes values in a $140$-codimensional submanifold 
$\jmath_{\Lag_{\mathfrak S}}\colon\mathcal{W}_{\mathcal{L}_{\mathfrak S}}\hookrightarrow \mathcal{W}_r$ which is identified with the graph of
a bundle map $\mathcal{FL}_{\mathfrak S} \colon J^3\pi \to J^{2}\pi^\ddagger$, over $J^1\pi$, defined locally by
$$
\mathcal{FL}_{\mathfrak S}^{\ \ *}p^{\alpha\beta,\mu}=
\frac{\partial \hat{ L}_{\mathfrak S}}{\partial g_{\alpha\beta,\mu}} - 
\sum_{\nu=0}^{3}\frac{1}{n(\mu\nu)}
D_\nu\left( \frac{\partial \hat{ L}_{\mathfrak S}}{\partial g_{\alpha\beta,\mu\nu}}\right)=
\hat{ L}_{\mathfrak S}^{\alpha\beta,\mu} \ , \
\mathcal{FL}_{\mathfrak S}^{\ \ *}p^{\alpha\beta,\mu\nu}=
\frac{\partial \hat {L}_{\mathfrak S}}{\partial g_{\alpha\beta,\mu\nu}}\, .
$$
What is equivalent, the submanifold 
$\mathcal{W}_{\mathcal{L}_{\mathfrak S}}$ is the graph of a bundle morphism
$\widetilde{\mathcal{FL}_{\mathfrak S}} \colon J^3\pi \to J^2\pi^\dagger$ over $J^1\pi$ defined locally by
\beann
\displaystyle\widetilde{\mathcal{FL}}_{\mathfrak S}^{\ *}p^{\alpha\beta,\mu}&=&
\frac{\partial \hat L_{\mathfrak S}}{\partial g_{\alpha\beta,\mu}} - 
\sum_{\nu=0}^{3}\frac{1}{n(\mu\nu)}
D_\nu\left( \frac{\partial \hat L_{\mathfrak S}}{\partial g_{\alpha\beta,\mu\nu}}\right)=
\hat L_{\mathfrak S}^{\alpha\beta,\mu} \ ,
\\
\widetilde{\mathcal{FL}}_{\mathfrak S}^{\ *}p^{\alpha\beta,\mu\nu} &=&
\frac{\partial \hat L_{\mathfrak S}}{\partial g_{\alpha\beta,\mu\nu}}\, , \\ 
\displaystyle
\widetilde{\mathcal{FL}}_{\mathfrak S}^{\ *}p&=& \hat{L}_{\mathfrak S} - 
g_{\alpha\beta,\mu}\left(\frac{\partial \hat L_{\mathfrak S}}{\partial g_{\alpha\beta,\mu}} - 
\sum_{\nu=0}^{3}\frac{1}{n(\mu\nu)}
D_\nu\left( \frac{\partial \hat L_{\mathfrak S}}{\partial g_{\alpha\beta,\mu\nu}}\right)\right)-
g_{\alpha\beta,\mu\nu}\frac{\partial \hat L_{\mathfrak S}}{\partial g_{\alpha\beta,\mu\nu}}\\
&=& \displaystyle
\hat{L}_{\mathfrak S}-\sum_{\alpha\leq\beta}p^{\alpha\beta,\mu}g_{\alpha\beta,\mu}-
\sum_{\alpha\leq\beta,\mu\leq\nu}p^{\alpha\beta,\mu\nu}g_{\alpha\beta,\mu\nu} \, .
\eeann
\end{proposition}

\begin{theorem}
\label{theo:submanifoldMatter} 
A solution to the equation \eqref{eqn:UnifDynEqMultiVFMass} exists only in a submanifold
${\cal W}_{\mathfrak S}\hookrightarrow\W_r$ wich, depending on the degree of $L_{\mathfrak m}$, 
is locally defined by the following constraints (for $0\leq\alpha\leq\beta\leq3$, $0\leq\mu\leq\nu\leq3$,
$0\leq\tau\leq3$):
\bi{} 
\item{If $\deg(L_{\mathfrak m})=4$:}$\quad
p^{\alpha\beta,\mu\nu}- \frac{\partial \hat L_{\mathfrak S}}{\partial g_{\alpha\beta,\mu\nu}}=0 , \quad 
p^{\alpha\beta,\mu}-\hat{L}^{\alpha\beta,\mu}_{\mathfrak S}=0. 
$
\item{If $\deg(L_{\mathfrak m})=3$:}$\quad
p^{\alpha\beta,\mu\nu}- \frac{\partial \hat L_{\mathfrak S}}{\partial g_{\alpha\beta,\mu\nu}}=0 , \quad 
p^{\alpha\beta,\mu}-\hat{L}^{\alpha\beta,\mu}_{\mathfrak S}=0 , \quad 
 \hat{L}^{\alpha\beta}_{\mathfrak S}=0 .
$
\item{If $\deg(L_{\mathfrak m})\leq2$:}$\quad
p^{\alpha\beta,\mu\nu}- \frac{\partial \hat L_{\mathfrak S}}{\partial g_{\alpha\beta,\mu\nu}}=0 , \quad 
p^{\alpha\beta,\mu}-\hat{L}^{\alpha\beta,\mu}_{\mathfrak S}=0 , \quad 
 \hat{L}^{\alpha\beta}_{\mathfrak S}=0  , \quad 
D_\tau\hat{L}^{\alpha\beta}_{\mathfrak S}=0.
$
\ei{} 
\end{theorem}
\begin{proof}
For the case $\deg(L_{\mathfrak m})=4$, the first two restrictions,
which involve the momenta, 
hold for every second order field theory (Proposition \ref{prop:LegMass} and \cite{pere}). 

If $\deg(L_{\mathfrak m})\leq2$, then $\deg(L_{\mathfrak S})=c\leq 2$. 
Therefore $\Theta_{L_{\mathfrak S}}$ is $\pi^3_{c-1}$-projectable (in particular $\pi^3_{1}$-projectable), 
which implies the other two restrictions \cite{art:GR-2016}. 
They can also be obtained by a similar procedure as in Section \ref{lhuni}. 

Likewise, if $\deg(L_{\mathfrak m})=3$, then $\deg(L_{\mathfrak S})=3$, and $\Theta_{L_{\mathfrak S}}$ is $\pi^3_{2}$-projectable, 
which implies $\hat{L}^{\alpha\beta}_{\mathfrak S}=~0$.
\end{proof}

Depending on the energy-matter term, maybe there are not any holonomic solution on 
${\cal W}_{\mathfrak S}$ . In this situations, 
a smaller submanifold has to be considered in order to find a holonomic solution.

\subsection{Lagrangian and Hamiltonian formalisms}

In section \ref{LagForm} we have stated how to recover 
the Lagrangian formalism from the unified formalism 
for the Hilbert-Einstein Lagrangian with no energy-matter souces. 
As in that case, now the Lagrangian formalism takes place in $J^3\pi$, and the
Poincar\' e-Cartan forms \eqref{thetalag} associated with
the Hilbert-Einstein Lagrangian with energy-matter sources are
$$
\Theta_{\mathcal{L}_{\mathfrak S}} \equiv \widetilde{\mathcal{FL}}_{\mathfrak S}^{\ *}\Theta_1^s \in \Omega^{4}(J^3\pi)
\quad , \quad
\Omega_{\mathcal{L}_{\mathfrak S}} \equiv \widetilde{\mathcal{FL}}_{\mathfrak S}^{\ *}\Omega_1^s = -\d\Theta_{\mathcal{L}_{\mathfrak S}} \in \Omega^{5}(J^3\pi) \ ,
$$
which have the local expressions
\beann
\Theta_{\mathcal{L}_{\mathfrak S}}&=&H_{\mathfrak S}\d^4x +\sum_{\alpha\leq\beta}L_{\mathfrak S}^{\alpha\beta,\mu}\d g_{\alpha\beta}\wedge \d^3x_\mu+\sum_{\alpha\leq\beta}L_{\mathfrak S}^{\alpha\beta,\mu\nu}\d g_{\alpha\beta,\mu}\wedge \d^3x_{\nu}\ ,
\\
 \Omega_{{\mathcal{L}}_{\mathfrak S}}&=&
\d H_{\mathfrak S}\wedge\d^4x-\sum_{\alpha\leq\beta}\d L_{\mathfrak S}^{\alpha\beta,\mu}\wedge\d g_{\alpha\beta}\wedge \d^{m-1}x_\mu-\sum_{\alpha\leq\beta}\d L_{\mathfrak S}^{\alpha\beta,\mu\nu}\wedge\d g_{\alpha\beta,\mu}\wedge \d^{m-1}x_{\nu} \ ;
\eeann
where
\beann
H_{\mathfrak S}&\equiv& 
(\jmath_\Lag\circ(\rho_1^\mathcal{L})^{-1})^*{\hat H}_{\mathfrak S}=
\sum_{\substack{\alpha\leq \beta}}L_{\mathfrak S}^{\alpha\beta,\mu\nu}g_{\alpha\beta,\mu\nu}+
\sum_{\substack{\alpha\leq \beta}}L_{\mathfrak S}^{\alpha\beta,\mu}g_{\alpha\beta,\mu}-L_{\mathfrak S} \ ,
\eeann
and $L_{\mathfrak S}^{\alpha\beta,\mu\nu},L_{\mathfrak S}^{\alpha\beta,\mu}$ 
have the same coordinate expressions than
$\hat L^{\alpha\beta,\mu\nu},\hat L^{\alpha\beta,\mu}$,
and $\hat L_0$.

The Lagrangian problem associated with the Lagrangian system 
$(J^3\pi,\Omega_{\mathcal{L}_{\mathfrak S}})$ is stated like
in equations \eqref{eq:ls} and \eqref{eq:lf}, but for
$\Omega_{\mathcal{L}_{\mathfrak S}}$ instead of $\Omega_{\mathcal{L}_{\mathfrak V}}$. 
The solutions are related to the solutions of the unified formalism by 
Proposition \ref{prop:UnifToLagSect} and Theorem \ref{thm:UnifToLagMultiVF}.

The Lagrangian counterpart of theorem \ref{theo:submanifoldMatter} is:

\begin{corollary}\label{coro:submanifoldMatter} 
A solution to the equation \eqref{eq:lf} exists only in a submanifold
$S_{\mathfrak S}\hookrightarrow J^3\pi$ wich, depending on the degree of $L_{\mathfrak m}$, 
is locally defined by the following constraints (for $0\leq\alpha\leq\beta\leq3$):
\bi{} 
\item{If $\deg(L_{\mathfrak m})=3$:}$\quad
{L}^{\alpha\beta}_{\mathfrak S}=0 .
$
\item{If $\deg(L_{\mathfrak m})\leq2$:}$\quad
 {L}^{\alpha\beta}_{\mathfrak S}=0  , \quad 
D_\tau{L}^{\alpha\beta}_{\mathfrak S}=0.
$
\ei{}
\end{corollary}

The existence of holonomic solutions depends on the energy-mass term. 
In some cases we must continue the constraint algorithm, 
together with an integrability algorithm.

Finally, the equations of the integral sections \eqref{eq:ls} 
can be analyzed in a similar fashion as in Section \ref{FieEqSec},
and using Proposition \ref{prop:UnifToLagSect}. 
This leads to the Euler-Lagrange equations
\beq{}\label{eq:EinEqMass}
L_{\mathfrak S}^{\alpha\beta}|_{j^3\phi}=
L_{\mathfrak V}^{\alpha\beta}|_{j^3\phi}+L_{\mathfrak m}^{\alpha\beta}|_{j^3\phi}
=\left.-\varrho\,n(\alpha\beta) \left(R^{\alpha\beta}-
\frac{1}{2}g^{\alpha\beta}R-\frac{1}{\varrho n(\alpha\beta)}L_{\mathfrak m}^{\alpha\beta}\right)\right|_{j^3\phi}=0 \ ,
\eeq
Introducing the {\sl stress-energy-momentum tensor} as
$$
T_{\mu\nu}=\frac{c^4}{8\pi G\varrho n(\alpha\beta)}g_{\alpha\mu}g_{\beta\nu}L_{\mathfrak m}^{\alpha\beta} \ .
$$
where $G$ as the Newton's gravitational constant and $c$ the speed of light, then
$$
R_{\mu\nu}-\frac{1}{2}g_{\mu\nu}R=\frac{8\pi G}{c^4}T_{\mu\nu} \ ,
$$
and equations \eqref{eq:EinEqMass} are equivalent to 
the Einstein equations with stress-energy-momentum tensor.

All these results can be also obtained applying the constraint algorithm straightforwardly
to the equation  \eqref{eq:ls}, in the same way as we have done for the unified formalism;
then doing a purely Lagrangian analysis.

For establishing the multimomentum Hamiltonian formalism we use
the Legendre maps $\mathcal{FL}_{\mathfrak S}$ and $\widetilde{\mathcal{FL}}_{\mathfrak S}$ 
defined in Proposition \ref{prop:LegMass}.
Now, we denote 
$\widetilde{\mathcal{P}_{\mathfrak S}}=\widetilde{\mathcal{FL}}_{\mathfrak S}(J^3\pi) \stackrel{\tilde{\jmath}}{\hookrightarrow} J^2\pi^\dagger$
and $\mathcal{P}_{\mathfrak S}=\mathcal{FL}_{\mathfrak S}(J^3\pi) \stackrel{\jmath}{\hookrightarrow} J^2\pi^\ddagger$, 
with the natural projection
$\bar{\pi}_{\mathcal{P}_{\mathfrak S}} \colon \mathcal{P}_{\mathfrak S} \to M$.
In order to assure the existence of the Hamiltonian formalism we demand that
the Lagrangian density $\mathcal{L}_{\mathfrak S} \in \Omega^4(J^2\pi)$ is, at least, 
{\sl almost-regular}. Then we can define the Hamiltonian forms
$\Theta_{h_{\mathfrak S}}$ and $\Omega_{h_{\mathfrak S}}$, 
and then we have the Hamiltonian system $({\cal P}_{\mathfrak S},\Omega_{h_{\mathfrak S}})$. 
The \textsl{Hamiltonian problem} associated with this system
is stated in \eqref{eq:hsec} and \eqref{eq:hvf}, but with
$\Omega_{h_{\mathfrak S}}$ instead of $\Omega_{h_{\mathfrak V}}$. 
This Hamiltonian formalism is recovered from the unified formalism 
following Proposition \ref{Hamprop} and Theorem \ref{Hamteor}.

In the actual case, the formalism depends strongly on the singularity of the theory. 
Nevertheless, if $\deg(L_{\mathfrak m})\leq 2$ (or equivalently $\deg(L_{\mathfrak S})\leq 2$), 
we have a similar situation as in the vacuum case.
In particular:

\begin{proposition}
If $\deg(L_{\mathfrak S})\leq2$, then $L_{\mathfrak S}$ is an almost-regular Lagrangian
and
$\mathcal{P}_{\mathfrak S}$ is diffeomorphic to $J^1\pi$.
\label{1stdifeo2}
\end{proposition}
\begin{proof}
If $\deg(L_{\mathfrak S})\leq 2$, we have that
$$
\Tan_{{j^3_x}\phi}\mathcal{FL}_{\mathfrak S}=\left( \begin{array}{ccccc}
Id_{4}& 0 & 0 & 0 & 0 \\
0 & Id_{10} & 0 & 0 & 0 \\
0 & 0 &Id_{40} & 0 & 0  \\
\displaystyle\frac{\partial \hat L_{\mathfrak S}^{\alpha\beta,\mu}}{\partial x^\tau} & 
\displaystyle\frac{\partial \hat L_{\mathfrak S}^{\alpha\beta,\mu}}{\partial g_{\gamma\delta}} & \displaystyle\frac{\partial \hat L_{\mathfrak S}^{\alpha\beta,\mu}}{\partial g_{\gamma\delta,\tau}}&
 0 & 0  \\
\displaystyle\frac{\partial \hat L_{\mathfrak S}}{\partial x^\tau\partial g_{\alpha\beta,\mu\nu}}  &
\displaystyle\frac{\partial \hat L_{\mathfrak S}}{\partial g_{\gamma\delta}\partial g_{\alpha\beta,\mu\nu}} & \displaystyle\frac{\partial \hat L_{\mathfrak S}}{\partial g_{\gamma\delta,\tau}\partial g_{\alpha\beta,\mu\nu}} & 0 & 0 
\end{array} \right)
$$
Then we have that ${\rm rank}(\Tan_{{j^3_x}\phi}\mathcal{FL}_{\mathfrak S})=54$ 
at every point  $j^3_x\phi\in J^3\pi$. Therefore 
$\Tan\mathcal{FL}_{\mathfrak S}$ is surjective and 
$\mathcal{FL}_{\mathfrak S}$ is a submersion. 
From here the proof is the same as
in Proposition \ref{prop:AlmReg}.
\end{proof}

In general the functions ${\hat {L}_{\mathfrak S}}^{\alpha\beta,\mu}$ are not invertible,
thus we use the non momenta coordinates
$(x^\mu,g_{\alpha\beta},g_{\alpha\beta,\mu})$ as local coordinates 
adapted to $\mathcal{P}_{\mathfrak S}$. The function $H_{\mathfrak S}$ is defined by
$$
H_{\mathfrak S}=
\sum_{\alpha\leq \beta}L_{\mathfrak S}^{\alpha\beta,\mu\nu}g_{\alpha\beta,\mu\nu}+
\sum_{\substack{\alpha\leq \beta}}L_{\mathfrak S}^{\alpha\beta,\mu}g_{\alpha\beta,\mu}-L_{\mathfrak S}  \ ,
$$
and the Hamilton-Cartan form have the coordinate expressions
$$
 \Omega_{h_{\mathfrak S}}=-\d\Theta_{h_{\mathfrak S}}=
\d H_{\mathfrak S}\wedge\d^4x-\sum_{\alpha\leq\beta}\d L_{\mathfrak S}^{\alpha\beta,\mu}\wedge\d g_{\alpha\beta}\wedge \d^3x_\mu-
\sum_{\alpha\leq\beta}\d L_{\mathfrak S}^{\alpha\beta,\mu\nu}\wedge\d g_{\alpha\beta,\mu}\wedge \d^3x_{\nu}\ .
$$
The resulting Hamiltonian equations for sections \eqref{eq:hsec} are
$$
\left.\left(\frac{\partial L_{\mathfrak S}^{\alpha\beta,\nu}}{\partial g_{ab,\mu}}-\frac{\partial L_{\mathfrak S}^{ab,\mu\nu}}{\partial g_{\alpha\beta}}\right)\right|_{\psi}\frac{\partial \psi_{ab,\mu}}{\partial x^\nu}=\left.-\frac{\partial H_{\mathcal{P}_{\mathfrak S}}}{\partial g_{\alpha\beta}}\right|_{\psi}-\left.\psi_{ab,\mu}\left(\frac{\partial L_{\mathfrak S}^{\alpha\beta,\mu}}{\partial g_{ab}}-\frac{\partial L_{\mathfrak S}^{ab,\mu}}{\partial g_{\alpha\beta}}\right)\right|_{\psi}\ ,
$$
and rearranging the terms, these equations are locally equivalent 
to the Einstein equations \eqref{eq:EinEqMass}.

If $\deg(L_{\mathfrak S})>2$, then $\mathcal{FL}_{\mathfrak S}$ may not be a submersion and,
hence, $L_{\mathfrak S}$ is not almost-regular.
In these cases the construction of the Hamiltonian formalism is more complicated.

\subsection{Example: Electromagnetic source }

Consider the case of a free electromagnetic source with electromagnetic tensor $F^{\mu\nu}$.
The corresponding Lagrangian function is
$$
L_{\mathfrak m}=\sqrt{|det(g_{\alpha\beta})|}\,F_{\mu\nu}F^{\mu\nu} \ ,
$$
where the components of the tensor $F_{\mu\nu}$ are functions on the base manifold $M$. 
In this case, $\deg(L_{ \mathfrak{m}})=1$, 
and the stress-energy-momentum tensor is
$$
{\rm T}_{\mu\nu}=
\frac{c^4}{\varrho n(\mu\nu)8\pi G}g_{\alpha\mu}g_{\beta\nu}L_{\mathfrak m}^{\alpha\beta}=
\frac{c^4}{\varrho n(\mu\nu)8\pi G}g_{\alpha\mu}g_{\beta\nu}\derpar{L_{\mathfrak m}}{g_{\alpha\beta}}=
\frac{c^4}{4\pi G}\left(\frac{1}{4}g_{\mu\nu}F^{\alpha\beta}F_{\alpha\beta}-g^{\alpha\beta}F_{\mu\alpha}F_{\nu\beta}\right) \ .
$$

The corresponding form $\Theta_{{\mathfrak S}r} $ is $\pi^3_1$-projectable, 
which implies that $\mathfrak{X}^V(\pi^3_1)$ are gauge vector fields. 
By Theorem \ref{theo:submanifoldMatter}, solutions to the field equations
exist on the points of the submanifold defined by
$$\quad
p^{\alpha\beta,\mu\nu}-\frac{\hat{L}_{\mathfrak S}}{\partial g_{\alpha\beta,\mu\nu}}=0 \quad , \quad 
p^{\alpha\beta,\mu}-\hat{L}^{\alpha\beta,\mu}_{\mathfrak S}=0 
\quad , \quad 
 \hat{L}^{\alpha\beta}_{\mathfrak S}=0  \quad , \quad 
D_\tau\hat{L}^{\alpha\beta}_{\mathfrak S}=0 \ .
$$
The first two restrictions define the Legendre transformation, 
and the last two fix the gauge freedom of the higher derivatives. 
The local expression of any semiholonomic multivetor field solution to 
(\ref{eqn:UnifDynEqMultiVFMass}) can be obtained by combining these restrictions,
the holonomy conditions, and the solution obtained in the Appendix \ref{sect:SolMatter},
\bea
{\bf{X}}_{LH}&&=
\bigwedge_{\tau=0}^3\sum_{\substack{\alpha\leq\beta\\\mu\leq\nu\leq\lambda}}\left(\frac{\partial}{\partial x^\tau}+ 
g_{\alpha\beta,\tau}\frac{\partial}{\partial g_{\alpha\beta}}+ 
g_{\alpha\beta,\mu\tau}\frac{\partial}{\partial g_{\alpha\beta,\mu}}+ 
g_{\alpha\beta,\mu\nu\tau}\frac{\partial}{\partial g_{\alpha\beta,\mu\nu}}+\right.
\nonumber \\ \nonumber
&&\left.D_\tau D_\lambda\hat{F}_{\alpha\beta;\mu,\nu}\frac{\partial}{\partial g_{\alpha\beta,\mu\nu\lambda}}+
D_\tau\hat{L}^{\alpha\beta,\mu}_{\mathfrak S}\frac{\partial}{\partial p^{\alpha\beta,\mu}}+
D_\tau\frac{\hat{L}_{\mathfrak S}}{\partial g_{\alpha\beta,\mu\nu}}\frac{\partial}{\partial p^{\alpha\beta,\mu\nu}}\right) ,
\label{mvfuni2}
\eea
where $\hat{F}_{\alpha\beta;\mu,\nu}= (\pi^3_1 \circ \rho_1^r)^*F_{\alpha\beta;\mu,\nu}\in C^\infty(\mathcal{W}_r)$, and
$$
F_{\alpha\beta;\mu,\nu}=
g_{\lambda\sigma}(\Gamma_{\nu \alpha }^\lambda\Gamma_{\mu \beta}^\sigma+
\Gamma_{\nu \beta}^\lambda\Gamma_{\mu \alpha }^\sigma)+\frac{c^4}{4\pi G}
g_{\alpha\beta} \left(g^{\lambda\sigma}F_{\mu\lambda}F_ {\nu\sigma}-\frac54g_{\mu\nu}F_{\lambda\sigma}F^{\lambda\sigma}\right)+
{F^h}_{\lambda\sigma;\mu,\nu}\in C^\infty(J^1\pi) \ .
$$

The Lagrangian formalism takes place in $J^3\pi$, 
but the Corollary \ref{coro:submanifoldMatter} states that a solution exists in the submanifold defined by 
$$
 {L}^{\alpha\beta}_{\mathfrak S}=0 \quad , \quad 
D_\tau{L}^{\alpha\beta}_{\mathfrak S}=0 \ .
$$
The Euler-Lagrange equations (\ref{eq:EinEqMass}) are equivalent to the Einstein equations
$$
\left.\left(R_{\mu\nu}-\frac{1}{2}g_{\mu\nu}R\right)\right|_{j^3\phi}=
\left.\frac{c^4}{4\pi G}\left(\frac{1}{4}g_{\mu\nu}F^{\alpha\beta}F_{\alpha\beta}-g^{\alpha\beta}F_{\mu\alpha}F_{\nu\beta}\right) \right| _{j^3\phi}\ ,
$$
A section $\psi\colon M\rightarrow E$ is a solution to the Einstein equations if,
on the points of its image, it is a section of a multivector field with local expression
$$
{\bf{X}}_{L}=
\bigwedge_{\tau=0}^3\sum_{\substack{\alpha\leq\beta\\\mu\leq\nu\leq\lambda}}\left(\frac{\partial}{\partial x^\tau}+ 
g_{\alpha\beta,\tau}\frac{\partial}{\partial g_{\alpha\beta}}+ 
g_{\alpha\beta,\mu\tau}\frac{\partial}{\partial g_{\alpha\beta,\mu}}+ 
g_{\alpha\beta,\mu\nu\tau}\frac{\partial}{\partial g_{\alpha\beta,\mu\nu}}+
D_\tau D_\lambda{F'}_{\alpha\beta;\mu,\tau}\frac{\partial}{\partial g_{\alpha\beta,\mu\nu\lambda}}\right) \ ,
$$
where ${F'}_{\alpha\beta;\mu,\nu}= {\pi^3_1}^*F_{\alpha\beta;\mu,\nu}\in C^\infty(J^3\pi)$.

For the Hamiltonian formalism, we have the Hamiltonian system $({\cal P}_{\mathfrak S},\Omega_{h_{\mathfrak S}})$, 
where ${\cal P}_{\mathfrak S}$ is diffeomorphic to $J^1\pi$,
as a consequence of Proposition  \ref{1stdifeo2}, and the Hamiltonian function 
giving the Hamiltonian section $h$ is
$$
H_{{\cal P}_{\mathfrak S}}=H_{\cal P}-L_{\mathfrak{m}}\ ,
$$
where $H_{\cal P}$ is the Hamiltonian for the vacuum case (\ref{hamiltonianexp}). 
A semiholonomic multivector field solution to (\ref{eq:hvf}) has the local expression
$$
{\bf{X}}_{H}=
\bigwedge_{\tau=0}^3\sum_{\substack{\alpha\leq\beta}}\left(\frac{\partial}{\partial x^\tau}+ 
g_{\alpha\beta,\tau}\frac{\partial}{\partial g_{\alpha\beta}}+ 
F_{\alpha\beta;\mu,\tau}\frac{\partial}{\partial g_{\alpha\beta,\mu}}\right) \ .
$$

\section{Symmetries}
\label{symmetries}

\subsection{Symmetries in the Lagrangian formalism}

Symmetries for physical systems are transformations that 
map solutions to the physical equations into solutions.
Geometrically, they are assumed to be locally generated by vector fields 
and then they are called {\sl infinitesimal symmetries},
and these are the kinds of symmetries we are mainly interested.
First, we state geometrically these concepts, reviewing the following basic
definitions and results for the Lagrangian formalism 
(see, for instance, \cite{art:deLeon_etal2004,EMR-96,GPR-2017,Herman-2017} for more details).

Consider the Lagrangian system $(J^3\pi,\Omega_\Lag)$
(in our case, $\Lag=\Lag_{\mathfrak V}$ or
$\Lag=\Lag_{\mathfrak S}$),
with final constraint submanifold $S_f\hookrightarrow J^3\pi$, 
and let $\ker^4\Omega_\Lag:=\{{\bf X}\in\df^4(J^3\pi)\,\vert\, \inn({\bf X})\Omega_\Lag=0\}$.
We are interested in the following kind of symmetries:

\begin{definition}
A {\rm Cartan (Noether) symmetry} of $(J^3\pi,\Omega_\Lag)$
is a diffeomorphism $\Phi\colon J^3\pi\to J^3\pi$
such that $\Phi(S_f)=S_f$ and $\Phi^*\Omega_\Lag=\Omega_\Lag$.
If, in addition, $\Phi^*\Theta_\Lag=\Theta_\Lag$, then 
$\Phi$ is said to be an {\rm exact Cartan symmetry}.
Furthermore, if $\Phi=j^3\varphi$ 
for a diffeormorphism $\varphi\colon E\to E$, 
the Cartan symmetry is {\rm natural}.

An {\rm infinitesimal Cartan} or {\rm Noether symmetries} of $(J^3\pi,\Omega_\Lag)$ is a
vector field $X\in\vf (J^3\pi)$, tangent to $S_f$,
whose local flows are local Cartan symmetries; that is, satisfying that $\Lie(X)\Omega_\Lag=0$.
If, in addition, $\Lie(X)\Theta_\Lag=0$, then $X$ is said to be an
{\rm infinitesimal exact Cartan symmetry}.
Furthermore, if $X=j^3Y$, for some $Y\in\vf(E)$, 
then the infinitesimal symmetry is called {\rm natural}.
\end{definition}

Symmetries transform solutions to the field equations into solutions:

\begin{proposition}
If $\Phi=j^3\phi\colon J^3\pi\to J^3\pi$,
for a diffeormorphism $\varphi\colon E\to E$,
is a natural Cartan symmetry, 
and ${\bf X}\in\ker^4\Omega_\Lag$ is holonomic,
then $\Phi_t$ transforms the holonomic sections of ${\bf X}$
into holonomic sections, and hence
$\Phi_*{\bf X}\in\ker^4\Omega_\Lag$ is also holonomic.

As a consequence, if $X=j^3Y\in\vf(J^3\pi)$ is a natural infinitesimal Cartan symmetry,
and $\Phi_t$ is a local flow of $X$, then
$\Phi_t$ transforms the holonomic sections of ${\bf X}$
into holonomic sections.
\end{proposition}
\proof
Let $j^3\varphi\colon M\to J^3\pi$
be an holonomic section of ${\bf X}$,
for a section $\varphi\colon M\to E$;
then it is a solution to the field equations \eqref{eq:ls} and 
then $(j^3\varphi)^*\inn(X')\Omega_\Lag=0$,
for every $X'\in\vf(J^3\pi)$.
Therefore.
\beann
(j^3(\phi\circ\varphi))^*\inn(X')\Omega_\Lag &=&
((j^3\varphi)^*(j^3\phi)^*\inn(X')\Omega_\Lag)=
(j^3\varphi)^*\inn((j^3\phi)_*^{-1}X')(j^3\phi)^*\Omega_\Lag
\\ &=&
(j^3\varphi)^*\inn((j^3\phi)_*^{-1}X')\Omega_\Lag=0 \ ,
\eeann
since 
$(j^3\varphi)$ is a solution to the field equations.
Therefore $j^3(\phi_t\circ\varphi)$
is also a solution to the field equation.
The last statement is immediate since, as $X$ is an infinitesimal natural Cartan symmetry
its local flows $\Phi_t\colon J^3\pi\to J^3\pi$ are canonical liftings of the local flows
$\phi_t\colon E\to E$ of $Y$.
\qed

Symmetries are associated to the existence of
conserved quantities (or conservation laws).

\begin{definition}
A {\rm conserved quantity}
of  the Lagrangian system $(J^3\pi,\Omega_\Lag)$
is a form $\xi\in\df^3(J^3\pi)$ such that
$\Lie({\bf X})\xi=0$, for every $\bar\pi^3$-transverse multivector field
 ${\bf X}\in\ker^4\Omega_\Lag$.
 \end{definition}
 
Then, $\xi\in\df^3(J^3\pi)$ is a conserved quantity
if, and only if, $\Lie({\bf Z})\xi=0$, for every ${\bf Z}\in\ker^4\Omega_\Lag$,
and the following property holds: if $\xi\in\df^3(J^3\pi)$ is a
 conserved quantity and ${\bf X}\in\ker^4\Omega_\Lag$ is integrable,
 then $\xi$ is closed on the integral submanifolds of ${\bf X}$; that is,
if $j_S\colon S\hookrightarrow J^3\pi$ is an integral submanifold, 
then $\d j_S^*\xi=0$.
Therefore, if $\xi\in\df^3(J^3\pi)$, for every integral section $\psi\colon M\to J^3\pi$
of ${\bf X}$, 
there is a unique $X_{\psi^*\xi}\in\vf(M)$ such that $\inn(X_{\psi^*\xi})\eta=\psi^*\xi\in\df^3(M)$ and,
if ${\rm div}X_{\psi^*\xi}$ denotes the divergence of $X_{\psi^*\xi}$,
from the definition $\Lie(X_{\psi^*\xi})\eta:={\rm div}X_{\psi^*\xi}\eta$,
we obtain that $({\rm div}X_{\psi^*\xi})\,\eta=\d{\psi^*\xi}$.
Then, $\xi$ is a conserved quantity if,  and only if, 
${\rm div}X_{\psi^*\xi}=0$,
and hence, in every bounded domain $U\subset M$, Stokes theorem says that
$$
\int_{\partial U}{\psi^*\xi}=\int_U ({\rm div}X_{\psi^*\xi})\,\eta=\int_U\d{\psi^*\xi}=0 \ .
$$
The form $\psi^*\xi$ is called the {\sl current} associated with the conserved quantity $\xi$,

For infinitesimal Cartan symmetries,
the condition $\Lie(X)\Omega_\Lag=0$ is equivalent 
to demanding that $\inn(X)\Omega_\Lag$
is a closed $3$-form in $J^3\pi$ and hence
an infinitesimal Cartan symmetry is a locally Hamiltonian vector field
for the (pre)multisymplectic form $\Omega_\Lag$,
and $\xi_X\in\df^3(U)$ is the corresponding local Hamiltonian form
in an open neighbourhood $U\subset J^3\pi$: that is,
$\inn(X)\Omega_\Lag\vert_U=\d\xi_X$.
Therefore:

 \begin{theorem}
 {\rm (Noether):}
Let $X\in\vf (J^3\pi)$ be an infinitesimal Cartan symmetry, with $\inn(X)\Omega_\Lag=\d\xi_X$ (on $U\subset J^3\pi$). 
Then, for every $\bar\pi^3$-transverse multivector field ${\bf X}\in\ker^4\Omega_\Lag$
(and hence for every holonomic multivector field ${\bf X}\in\ker^4\Omega_\Lag$), we have that
 $$
 \Lie({\bf X})\xi_X=0 \ ;\ \mbox{\rm (on $U\subset J^3\pi$)} \ .
 $$
 that is, any Hamiltonian $3$-form
 $\xi_X$ associated with $X$ is a conserved quantity.
 As a particular case, if $X$ is an exact infinitesimal Cartan symmetry then
 $\xi_X=\inn(X)\Theta_\Lag$.
 
\noindent (For every integral submanifold $\psi$ of ${\bf X}$,
the form $\psi^*\xi_X$ is called a {\rm Noether current}, in this context).
 \label{Nth}
 \end{theorem}

It is well known that canonical liftings of diffeomorphisms and vector
fields preserve the canonical structures of $J^3\pi$. Nevertheless,
the form $\Omega_\Lag$ is not
canonical, since it depends on the choice of the Lagrangian density
$\Lag$, and then it is not invariant by these canonical liftings,
unless the Lagrangian density $\Lag$ is also
invariant by them. In this way, $\Omega_\Lag$ and hence the
Euler-Lagrange equations are invariant.
Thus, bearing in mind that $\Lag\in\df^4(J^2\pi)$,
another particular kind of symmetries are defined as follows:

\begin{definition}
A {\rm Lagrangian symmetry} of $(J^3\pi,\Omega_\Lag)$ is a diffeomorphism
$j^3\phi\colon J^3\pi\to J^3\pi$, for some $\phi\in{\rm Diff}(E)$,
 such that $(j^3\phi)(S_f)=S_f$ and
$(j^3\phi)((\bar\pi^3_2)^*\Lag)=(j^2\phi)^*\Lag=\Lag$.

An {\rm infinitesimal Lagrangian symmetry} of
$(J^3\pi,\Omega_\Lag)$ is a
vector field $j^3Y\in\vf(J^3\pi)$, for some $Y\in\vf(E)$,
such that $j^3Y$ is tangent to $S_f$ and
$\Lie (j^3Y)((\bar\pi^3_2)^*\Lag)=\Lie (j^2Y)\Lag=0$.
\end{definition}

Therefore, as a direct consequence of the above discussion, we have:

\begin{proposition}
If $j^3\phi\colon J^3\pi\to J^3\pi$ is a Lagrangian symmetry, then
$(j^3\phi)^*\Theta_{\Lag} =\Theta_{\Lag}$, and hence
it is an exact Cartan symmetry.

As a consequence, if \ $j^3Y\in\vf(J^3\pi)$ is an infinitesimal Lagrangian symmetry, then
$\Lie (j^3Y)\Theta_{\Lag} =0$, and hence it is
an infinitesimal exact Cartan symmetry.
\label{invasim}
\end{proposition}

Finally, a particular case of infinitesimal Cartan symmetries are
those vector fields, tangent to $S_f$, such that $\inn(X)\Omega_\Lag=0$.
Thus, we define:

\begin{definition} 
\label{gaugeinsym}
A {\rm (geometric) gauge vector field} (or a {\rm gauge variation}) is a vector field
\,$X\in\ker\,\Omega_\Lag$, tangent to $S_f$.
The $\bar\pi^3$-vertical elements of \,$\ker\,\Omega_\Lag$, which are tangent to $S_f$,
are the {\rm vertical gauge vector fields}
(or {\rm vertical gauge variations}).
Finally, if \,$X\in\ker\,\Omega_\Lag$
is $\bar\pi^3$-vertical, is tangent to $S_f$ and
is a natural vector field, it is said to be a {\rm natural gauge vector field}
or a {\rm natural gauge symmetry}.
\end{definition}

\noindent {\bf Remark}:
The origin of these definitions relies in the following facts:
First, the existence of gauge symmetries and of
gauge freedom is related to the non-regularity of the Lagrangian $\Lag$ and hence
gauge vector fields must be elements of $\ker\,\Omega_\Lag$.
Second,  gauge vector fields must be $\bar\pi^3$-vertical since,
in this way, we assure that the base manifold $M$ 
does not contain gauge equivalent points and then all the 
gauge degrees of freedom are in the fibres of $J^3\pi$
(and therefore, after doing a reduction procedure
or a gauge fixing in order to remove the gauge multiplicity,
the base manifold $M$ remains unchanged).
Finally, we demand that gauge symmetries are natural because
this condition assures that they transform holonomic solutions 
to the field equations into holonomic solutions.

\subsection{Symmetries in the unified formalism}

Consider the premultisymplectic system $(\W_r,\Omega_r)$
associated with the Lagrangian system $(J^3\pi,\Omega_\Lag)$, 
the submanifold $\W_\Lag\hookrightarrow\W_r$ 
where the field equations in the unified formalism are compatible, 
and the final constraint submanifold 
$\W_f\hookrightarrow\W_\Lag\hookrightarrow\W_r$.
Remember that, 
as it was stated in Section \ref{LagForm}, the map 
$\rho_1^\mathcal{L}\colon \mathcal{W}_\mathcal{L}\to J^3\pi$
is a diffeomorphism, with $\rho_1^\mathcal{L}(\W_f)=S_f$ and $(\rho_1^\mathcal{L})^*\Omega_\mathcal{L} =
\jmath_\mathcal{L}^*\Omega_r$.

As we are interested in infinitesimal symmetries,
the analysis is done only for this case
(the study for diffeomorphisms is similar).
Thus, denote 
$\underline{\vf(\W_\Lag)}:=\{ X^r\in\vf (\W_r)\ 
\vert\ \mbox{\rm $X^r$ is tangent to $\W_\Lag$} \}$.
First, if 
$X^r\in\underline{\vf(\W_\Lag)}$ satisfies that $\Lie(X^r)\Omega_r=0$, 
then
$$
0=\jmath_\mathcal{L}^*[\Lie(X^r)\Omega_r]=
\Lie(X^\Lag)(\jmath_\mathcal{L}^*\Omega_r)\ ;
$$
where $X
^\Lag\in\vf(\W_\Lag)$ is such that 
$\jmath_{\mathcal{L}*}X^\Lag=X^r\vert_{\W_\Lag}$.
Therefore,
$$
0=(\rho_1^\mathcal{L})^*[\Lie(X^\Lag)(\jmath_\mathcal{L}^*\Omega_r)]=
\Lie((\rho^\mathcal{L}_1)_*X^\Lag)\Omega_\Lag \ .
$$
Furthermore, if $X^r$ is tangent to $\W_f$,
then $(\rho^\mathcal{L}_1)_*X^\Lag\equiv Y$
is tangent to $S_f$ and, in this case, $X\in\vf(J^3\pi)$ 
is an infinitesimal Cartan symmetry of the Lagrangian system $(J^3\pi,\Omega_\Lag)$.
This leads to define:

\begin{definition}
An {\rm infinitesimal Cartan} or {\rm Noether symmetries} of the
premultisymplectic system $(\W_r,\Omega_r)$ is a
vector field $X^r\in\underline{\vf(\W_\Lag)}$, tangent to $\W_f$, 
satisfying that $\Lie(X^r)\Omega_r=0$.
If, in addition, $\Lie(X^r)\Theta_r=0$, then $X^r$ is said to be an
{\rm infinitesimal exact Cartan symmetry}.
Furthermore, if the associated vector field $X\in\vf(J^3\pi)$
is natural (i.e., $X=j^3Y$, for some $Y\in\vf (E)$), 
then $X^r$ is said to be an {\rm infinitesimal natural Cartan symmetry}.
\end{definition}

Observe also that, if $X^r$ is an infinitesimal natural Cartan symmetry of $(\W_r,\Omega_r)$,
then it is an holonomic vector field in $\W_r$, 
and it is not difficult to prove that the local flows of $Y^r$ transforms
holonomic solutions to the field equations of the system
$(\W_r,\Omega_r)$ into holonomic solutions too.

In the same way, if 
$X^r\in\underline{\vf(\W_\Lag)}$ satisfies that 
$\Lie(X^r)((\pi^3_2 \circ \rho_1^r)^*\Lag_{\mathfrak V})=0$,
this implies that
$$
0=\Lie(X^\Lag)((\bar\pi^3_2\circ\rho_1^r\circ\jmath_\Lag)^*\Lag_{\mathfrak V})=
\Lie(X^\Lag)((\bar\pi^3_2\circ\rho_1^\Lag)^*\Lag_{\mathfrak V})=
(\rho_1^\Lag)^*[\Lie(X)((\bar\pi^3_2)^*\Lag_{\mathfrak V})] \ .
$$
and, as $\rho_1^\Lag$ is a diffeomorphism, from here we obtain that
$\Lie(X)((\bar\pi^3_2)^*\Lag_{\mathfrak V})=0$.
In addition, if $X$ is tangent to $S_f$ and is a natural vector field, then
it is a Lagrangian symmetry of the Lagrangian system
$(J^3\pi,\Omega_\Lag)$. This justifies the following:

\begin{definition}
An {\rm infinitesimal Lagrangian symmetry} of
the premultisymplectic system $(\W_r,\Omega_r)$ is a
vector field $X^r\in\underline{\vf(\W_\Lag)}$, tangent to $\W_f$, 
satisfying that $\Lie(X^r)((\pi^3_2 \circ \rho_1^r)^*\Lag_{\mathfrak V})=0$,
and that the associated $X\in\vf(J^3\pi)$
is natural.
\end{definition}

Conserved quantities are defined in the unified formalism 
for the premultisymplectic system $(\W_r,\Omega_r)$ 
like in the Lagrangian formalism for $(J^3\pi,\Omega_\Lag)$,
and the their properties (including Noether's theorem)
are stated in the same way (see also \cite{GPR-2017}).
As a consequence, if $\xi\in\df^3(J^3\pi)$ is a conserved quantity 
of the Lagrangian system $(J^3\pi,\Omega_\Lag)$, then
$(\rho_1^r)^*\xi\in\df^3(\W_r)$
is a conserved quantity of the premultisymplectic system $(\W_r,\Omega_r)$.

\subsection{Symmetries for the Einstein-Hilbert model}

Now, consider the Hilbert-Einstein Lagrangian (withouth energy-matter sources).

\begin{definition}
Let $F\colon M\to M$ be a diffeomorphism.
The {\rm canonical lift of $\phi$ to the bundle of metrics $E$}
is the diffeomorphism ${\cal F}\colon E\to E$ 
defined as follows: for every $(x,g_x)\in E$, then
${\cal F}(x,g_x):=(F(x),(F^{-1})^*(g_x))$.
(Thus $\pi\circ{\cal F}=F\circ\pi$).

Let $Z\in\vf (M)$.
The {\rm canonical lift of $Z$ to the bundle of metrics $E$}
is the vector field $Y_Z\in\vf(E)$ whose associated
local one-parameter groups of diffeomorphisms  ${\cal F}_t$
are the canonical lifts to the bundle of metrics $E$
of the local one-parameter groups of diffeomorphisms $F_t$ of $Z$.
\end{definition}

In coordinates, if
$\displaystyle Z=f^\mu(x)\frac{\partial}{\partial x^\mu}\in\mathfrak{X}(M)$,
the canonical lift of $Z$ to the bundle of metrics is
$$
Y_Z=f^\mu\frac{\partial}{\partial x^\mu}-\sum_{\alpha\leq \beta}\left(\frac{\partial f^\mu}{\partial x^\alpha}g_{\mu\beta}+\frac{\partial f^\mu}{\partial x^\beta}g_{\mu\alpha}\right)\frac{\partial}{\partial g_{\alpha\beta}} \ ,
$$
and\ then
\beann
j^1Y_Z&=&
f^\mu\frac{\partial}{\partial x^\mu}-\sum_{\alpha\leq \beta}\left(\frac{\partial f^\mu}{\partial x^\alpha}g_{\mu\beta}+\frac{\partial f^\mu}{\partial x^\beta}g_{\mu\alpha}\right)\frac{\partial}{\partial g_{\alpha\beta}}
\\ & &
-\sum_{\alpha\leq\beta}\left(\frac{\partial^2f^\nu}{\partial x^\alpha\partial x^\mu}g_{\nu\beta}+\frac{\partial^2f^\nu}{\partial x^\beta\partial x^\mu}g_{\alpha\nu}+\frac{\partial f^\nu}{\partial x^\alpha}g_{\nu\beta,\mu}+\frac{\partial f^\nu}{\partial x^\beta}g_{\alpha\nu,\mu}+\frac{\partial f^\nu}{\partial x^\mu}g_{\alpha\beta,\nu}\right)\frac{\partial}{\partial g_{\alpha\beta,\mu}}
\\ &\equiv&
f^\mu\frac{\partial}{\partial x^\mu}+\sum_{\alpha\leq \beta}Y_{\alpha\beta}\frac{\partial}{\partial g_{\alpha\beta}}+\sum_{\alpha\leq\beta}Y_{\alpha\beta\mu}\frac{\partial}{\partial g_{\alpha\beta,\mu}}\ .
\eeann
For every $Z\in\mathfrak{X}(M)$, as $\mathcal{L}_{\mathfrak V}$ 
is invariant under diffeomorphisms, we have that
$$
\Lie (j^2Y_Z)\mathcal{L}_{\mathfrak V}=
\Lie (j^3Y_Z)((\pi^3_2)^*\mathcal{L}_{\mathfrak V})=0 \ .
$$ 
Furthermore $j^3Y_Z$ is tangent to $S_f$. In fact, as it is a natural vector field
that leaves the Hilbert-Einstein Lagrangian invariant, 
then the corresponding Euler-Lagrange equations
(the Einstein equations) are also invariant,
and hence for the constraints \eqref{Einseqagain} we have that
$$
\Lie(j^3Y_Z)L^{\alpha\beta}=0 \ ;
$$
while for the constraints \eqref{Einseqagain2} we obtain
\beann
\Lie(j^3Y_Z)(D_\tau L^{\alpha\beta})&=&\Lie(j^3Y_Z)\Lie(D_\tau)L^{\alpha\beta}=
\Lie(D_\tau)\Lie(j^3Y_Z)L^{\alpha\beta}+\Lie([j^3Y_Z,D_\tau])L^{\alpha\beta}
\\ &=& 
\Lie([j^3Y_Z,D_\tau])L^{\alpha\beta}=
-\frac{\partial f^\mu}{\partial x^\tau} D_\mu L^{\alpha\beta}=0
\quad \mbox{\rm (on $S_f$)} \ .
\eeann
Therefore $j^3Y_Z$ is an infinitesimal Lagrangian symmetry and then
it is an exact infinitesimal Cartan symmetry.
Its associated conserved quantity is $\xi_{Y_Z}=\inn(j^3Y_Z)\Theta_{\mathcal{L}_{\mathfrak V}}$ and,
as $\Theta_{\mathcal{L}_{\mathfrak V}}$ is $\pi^3_1$-basic,
there exists $\Theta^1_{\mathcal{L}_{\mathfrak V}}\in\df^4(J^1\pi)$
(which has the same coordinate expression) such that 
$\Theta_{\mathcal{L}_{\mathfrak V}}=(\pi^3_1)^*\Theta^1_{\mathcal{L}_{\mathfrak V}}$;
then
\beann
\xi_{Y_Z}&=&\inn(j^3Y)\Theta_{\mathcal{L}_{\mathfrak V}}=\inn(j^1Y)\Theta^1_{\mathcal{L}_{\mathfrak V}}
=\left(\sum_{\alpha\leq\beta}Y_{\alpha\beta}L^{\alpha\beta,\mu}+
\sum_{\alpha\leq\beta}Y_{\alpha\beta\nu}L^{\alpha\beta,\nu\mu}-f^\mu H \right)\d^3x_\mu
\\ & &
+\sum_{\alpha\leq\beta}\left(f^\nu L^{\alpha\beta,\mu}-f^\mu L^{\alpha\beta,\nu}\right)\d g_{\alpha\beta}\wedge\d^2x_{\nu\mu}
+\sum_{\alpha\leq\beta}\left(f^\nu L^{\alpha\beta,\lambda\mu}-f^\mu L^{\alpha\beta,\lambda\nu}\right)\d g_{\alpha\beta,\lambda}\wedge\d^2x_{\nu\mu} \ ,
\eeann
where
$\displaystyle\d^2 x_{\mu\nu}=
\inn\left(\frac{\partial}{\partial x^\nu}\right)\inn\left(\frac{\partial}{\partial x^\mu}\right)\d^4x$.
These vector fields are the only 
natural infinitesimal Lagrangian symmetries for this model \cite{Krupka,rosado}.

\section{Conclusions and outlook}

We have presented a multisymplectic covariant description
of the Einstein-Hilbert model of General Relativity
using a unified formulation joining both 
the Lagrangian and Hamiltonian formalisms.

Our procedure consists in using the constraint algorithm 
to determine a submanifold of the higher-order jet-multimomentum bundle $\W_r$
where the field equations written for multivector fields
\eqref{eqn:UnifDynEqMultiVF} are compatible; that is, 
where there exist classes of  holonomic multivector fields 
$\{{\bf X}\}$ which are solution to these equations. 
These classes of multivector fields are associated with holonomic distributions, 
whose integral sections are solutions to \eqref{eqn:UnifFieldEqSect}.
Thus, the constraints arising from the algorithm
determine where the image of the sections may lay.
This algorithm is also the main tool in order to state many of the
fundamental characteristics of the theory.

The constraints \eqref{eqn:UniVec3} and \eqref{Wlag2}, which define $\W_\Lag$, 
are a natural consequence of the unified formalism and define 
the Legendre map which allows to state the Hamiltonian formulation
and the Hamilton-de Donder-Weyl version of the Einstein equations.
Nevertheless, as the Hilbert-Einstein Lagrangian $\mathcal{L}$ is singular,
the algorithm produces more constraints.

In the case of no energy-matter sources,
among the new constraints, the physical relevant equations are the 
primary constraints \eqref{eqn:Res3} which,
evaluated on the points of the holonomic sections, 
are just the Einstein equations.
As a consequence of the singularity of $\mathcal{L}$,
they are $2nd$-order PDE's, instead of $4th$-order 
as correspond to a $2nd$-order Lagrangian.
Einstein's equations appear as constraints of the theory because
they are $2nd$-order PDE's which are defined as a submanifold of a higher-order bundle
(containing $J^3\pi$ as a subbundle).

The constraints \eqref{eqn:Res3} and \eqref{eqn:Res4}
are also related with the fact that $\Theta_r$ is 
$(\pi^3_1\circ\rho^r_1)$-projectable
and, as a consequence of this, in the Lagrangian formalism,
the Poincar\'e-Cartan form $\Theta_\mathcal{L}$ projects onto a form in $J^1\pi$, 
which is not the Poincar\'e-Cartan form of any first-order Lagrangian. 
Nevertheless, there is are first-order regular Lagrangians 
which are equivalent to the Hilbert-Einstein Lagrangian
\cite{first,IS-2016,Krupka,KrupkaStepanova,rosado,rosado2}.
The Lagrangian and Hamiltonian formalism of one of these 
Lagrangians have been analyzed in detail.

Thus, te secondary constraints  \eqref{eqn:Res4} contain no physical information:
they are of geometrical nature and arise because we are using a manifold 
prepared for a second-order theory of a Lagrangian 
which is physically equivalent to a first-order Lagrangian.
Hence, the constraints \eqref{eqn:UniVec3} and \eqref{Wlag2}, which define the Legendre map, 
and \eqref{eqn:Res3}, which is equivalent to the Einstein equations,
are the only relevant equations.

When we recover the Lagrangian formalism from the unified one,
as a consequence of the singularity of the Hilbert-Einstein Lagrangian,
solutions to the Euler-Lagrange field equations
only exist in a constraint submanifold $S_f\hookrightarrow J^3\pi$.
Furthermore, if we interpret the Einstein-Hilbert model as a gauge theory 
having the second and third order velocities as gauge vector fields (see \eqref{gaugevf}), 
the constraints \eqref{eqn:Res3} and \eqref{eqn:Res4} fix this gauge partially
(both in the unified and the Lagrangian formalisms).
To fix the remaining gauge degrees of freedom would lead,  in the Lagrangian formalism,
to a submanifold of $S_f$ diffeomorphic to $J^1\pi$.
In a forthcoming paper we will present an interpretation of gauge symmetries
for multisymplectic classical field theories.

Furthermore,  in the Lagrangian formalism, 
the Lagrangian constraints arise as a consequence of demanding the
holonomy condition for the solutions to the field equations and the fact that
the Hessian matrix of the Hilbert-Einstein Lagrangian with respects 
to the highest-order coordinates in $J^3\pi$ vanishes identically.
Hence these kinds of constraints are not projectable by the Legendre map
(see \cite{GPR-91} for an analysis of this subject for higher-order dynamical theories).

The multimomentum Hamiltonian formalism for the
Einstein Hilbert model has not gauge freedom, since the Hamilton-Cartan form is regular and
${\cal P}$ is diffeomorphic to $J^1\pi$ and $J^1\pi^*$
(see also the results in \cite{CVB-2006}).
In fact, this formalism is the same than the multimomentum Hamiltonian formalism
for the regular $1st$-order equivalent Lagrangian $\overline{\cal L}$
analysed in Section \ref{compare}.

When the energy-matter sources are present, 
some of the geometrical and physical characteristics of the theory 
depend on the properties of the Lagrangian $L_{\mathfrak m}$ representing the source.
In particular, the number of constraints arising from the constraint algorithm,
the obtention of holonomic multivector fields solution 
to the Lagrangian field equations, 
and the construction of the covariant multimomentum formalism.
This study has been done in detail 
for some cases of energy-matter sources (those which we are called ``of degree $\leq 2$''),
which include as a particular case the energy-matter sources coupled  to the metric
(for instance, the electromagnetic source or the perfect fluid).

In all the cases, we have obtained explicitly semiholonomic multivector fields representing
integrable distributions whose integral sections are solutions to the field equations. 

Finally, we have done also a discussion about symmetries and conserved quantities.
After stating the basic conceptss and properties in the Lagrangian formalism 
(including Noether's theorem), we have  characterized the symmetries and conservation laws in the unified formalism.
The analysis is completed giving the expression of the natural Lagrangian symmetries
and their conserved quantities
for the Hilbert-Einstein Lagrangian in vacuum.
A deeper study on all these topics is currently in progress.

Another model for the Einstein gravity theory is given by the
so-called {\sl affine-metric} or {\sl Einstein-Palatini Lagrangian},
which is a highly degenerated first-order Lagrangian
$\tilde{\cal L}$ depending linearly on the components of the metric $g$ 
and the components of an arbitrary connection $\Gamma$.
The gauge freedom of this model is higher than in the Einstein-Hilbert  model.
It is proved that the conditions of the connection to be metric and torsionless
(which allows us to recover the Einstein-Hilbert model from the Einstein-Palatini model)
are really a partial fixing of this gauge freedom \cite{pons}.
A multisymplectic study of the affine-metric model using 
the unified formalism is given in \cite{art:Capriotti2},
and the geometric analysis of the corresponding constraint algorithm and the
gauge symmetries of the model will be developed in a forthcoming work.

\appendix

\section{Solutions to the Hamiltonian equations for the Einstein-Hilbert model}
\label{solutions}

We have seen that the Einstein equations can be stated from different geometrical points of view. 
In order to solve them, we can use whichever we find more appropriate. 
Indeed, as it is explained in \cite{pere}, 
the solutions can be transported  canonically from one formalism to another. 
In this section we solve the equations for multivector fields 
in the Hamiltonian formalism.

A solution to the Einstein equations is a metric over the manifold;
that is, a section $\psi\colon M\to E$. 
The multivector fields we find provide system of partial differential equations 
whose solutions are the sections \eqref{pdesect2}. In this sense, 
finding the multivector fields is only the first step on solving Einstein equations. 
Nevertheless, this approach leads to new equations, which may be more appealing. 
For instance, they have a unique solution provided an initial condition: 
there is no need of boundary conditions.

The relation between sections and multivector fields is explained in Section \ref{mvf}. 
Only holonomic multivector fields have associated holonomic integral sections. 
Nevertheless, we look first for semiholonomic multivector fields,
except in the case of the vacuum case, where we find
a particular solution which is a proper holonomic multivector field. 
It is used in Theorem \ref{theo:submanifold} to determine the final submanifold.

Since the equations for multivector fields are lineal, 
we proceed to find a particular solution and then the homogeneous solutions for the vacuum case. 
Later, we will consider energy-matter sources.

\subsection{Particular solution (without energy-matter sources)}

The Hamiltonian problem for the
premultisymplectic system ($\mathcal{P},\Omega_h$) consists in finding 
classes of  holonomic $\bar\pi_{\cal P}$-transverse multivector fields 
$\{{\bf X}_h\}\subset\mathfrak{X}^4(\mathcal{P})$ such that
\beq{}\label{eq:hvf2}
\inn{({\bf X}_h)}\Omega_{h_{\mathfrak V}}=0 \quad ,\quad\forall {\bf X}_h\in\{{\bf X}_h\}\subset\mathfrak{X}^4(\mathcal{P}) \ .
\eeq
The local expression of a representative of a class $\{{\bf X}_h\}$ of 
these kinds of multivector fields in $\mathcal{P}$ is
$$
{\bf X}_h=
\bigwedge_{\nu=0}^3\left(\frac {\partial}{\partial x^\nu}+
\sum_{\alpha\leq\beta}\left(F_{\alpha\beta,\nu}\frac{\partial}{\partial g_{\alpha\beta}}
+F_{\alpha\beta;\mu,\nu}\frac{\partial}{\partial g_{\alpha\beta;\mu}}\right)\right)\ .
$$
Equation \eqref{eq:hvf2} takes the local expression:
\bea\label{eq:sol1}
\frac{\partial H_{\mathfrak V}}{\partial g_{\alpha\beta,\mu}}
+\sum_{\lambda\leq\sigma}F_{\lambda\sigma,\nu} \left(\frac{\partial L^{\alpha\beta,\mu\nu}}{\partial g_{\lambda\sigma}}-\frac{\partial L^{\lambda\sigma,\nu}}{\partial g_{\alpha\beta,\mu}}\right)&=&0 \ ,\\
\label{eq:sol2}
\frac{\partial H_{\mathfrak V}}{\partial g_{\alpha\beta}}+
\sum_{\lambda\leq\sigma}F_{\lambda\sigma,\mu}\left(\frac{\partial L^{\alpha\beta,\mu}}{\partial g_{\lambda\sigma}}-\frac{\partial L^{\lambda\sigma,\mu}}{\partial g_{\alpha\beta}}\right)+
\sum_{\lambda\leq\sigma}F_{\lambda\sigma;\nu,\mu}\left(\frac{\partial L^{\alpha\beta,\mu}}{\partial g_{\lambda\sigma,\nu}}-
\frac{\partial L^{\lambda\sigma,\nu\mu}}{\partial g_{\alpha\beta}}\right)&=&0\ .
\eea
We denote $\displaystyle U^{\alpha\beta,\mu\nu,\lambda\sigma}=\frac{\partial L^{\alpha\beta,\mu\nu}}{\partial g_{\lambda\sigma}}- \frac{\partial L^{\lambda\sigma,\nu}}{\partial g_{\alpha\beta,\mu}}$,  
whose explicit expressions are
\bea{}\label{eq:U}
U^{\alpha\beta\mu \nu \lambda\sigma}&=&\frac{\varrho\, n(\alpha\beta)n(\lambda\sigma)}{4}\left(-2g^{\alpha\beta}g^{\lambda\sigma}g^{\mu\nu}+g^{\alpha\lambda}g^{\beta\sigma}g^{\mu\nu}+g^{\beta\lambda}g^{\alpha\sigma}g^{\mu\nu}\right.
\nonumber\\ \nonumber
 &+&\left.g^{\alpha\beta}g^{\lambda\mu}g^{\sigma\nu}+g^{\alpha\beta}g^{\sigma\mu}g^{\lambda\nu}+g^{\lambda\sigma}g^{\alpha\nu}g^{\beta\mu}+g^{\lambda\sigma}g^{\beta\nu}g^{\alpha\mu} \right.\\
 &-&\left.g^{\alpha\nu}g^{\lambda\mu}g^{\beta\sigma}-g^{\beta\nu}g^{\lambda\mu}g^{\alpha\sigma}-g^{\alpha\nu}g^{\sigma\mu}g^{\beta\lambda}-g^{\beta\nu}g^{\sigma\mu}g^{\alpha\lambda} \right) \ ,
\eea{}
and they fulfill the following relations:
$$
U^{\alpha\beta,\mu\nu,\lambda\sigma}=U^{\lambda\sigma,\mu\nu,\alpha\beta}=-U^{\alpha\mu,\beta\nu,\lambda\sigma} \ .
$$
The equations are algebraic, in the sense that no derivatives of 
$F_{\alpha\beta,\mu}$,  nor $F_{\alpha\beta,\mu\nu}$ appear. 
(The indices are symmetrized as usual).

We start by solving equation \eqref{eq:sol1}. First, we rewrite it as
$$
\sum_{\lambda\leq\sigma}(F_{\lambda\sigma,\nu}-g_{\lambda\sigma,\nu}) U^{\alpha\beta,\mu\nu,\lambda\sigma}=0 \ .
$$
Indeed, since 
$\displaystyle H_{\mathfrak V}=\sum_{\lambda\leq\sigma}L^{\lambda\sigma,\nu}g_{\lambda\sigma,\nu}-L_0$, we obtain
\beann
\frac{\partial H_{\mathfrak V}}{\partial g_{\alpha\beta,\mu}}&=&
\sum_{\lambda\leq\sigma}\frac{\partial L^{\lambda\sigma,\nu}}{\partial g_{\alpha\beta,\mu}}g_{\lambda\sigma,\nu}+L^{\alpha\beta,\mu}-\frac{\partial L_0}{\partial g_{\alpha\beta,\mu}}\\
&=&\sum_{\lambda\leq\sigma}\frac{\partial L^{\lambda\sigma,\nu}}{\partial g_{\alpha\beta,\mu}}g_{\lambda\sigma,\nu}+\frac{\partial L_0}{\partial g_{\alpha\beta,\mu}}-\sum_{\lambda\leq\sigma}\frac{\partial L^{\alpha\beta,\mu\nu}}{\partial g_{\lambda\sigma}}g_{\lambda\sigma,\nu}-\frac{\partial L_0}{\partial g_{\alpha\beta,\mu}}=-\sum_{\lambda\leq\sigma}U^{\alpha\beta,\mu\nu,\lambda\sigma}g_{\lambda\sigma,\nu} \ .
\eeann
Now we multiply it by
\bea{}\nonumber
V_{\alpha\beta\mu, abc}=\frac{1}{\varrho\, n(\alpha\beta)}(g_{\alpha \mu}g_{\beta b}g_{ac}+2g_{\alpha \mu}g_{\beta c}g_{ab}+g_{\alpha\beta}g_{b \mu}g_{ac}-g_{\alpha\beta}g_{\mu c}g_{ab}\\\nonumber
-3g_{\alpha a}g_{\beta c}g_{b \mu}-3g_{\alpha b}g_{\beta c}g_{a \mu}+g_{\alpha \mu}g_{\beta a}g_{bc}+g_{\alpha\beta}g_{c \mu}g_{ab}) \ ,
\eea{}
which works as a sort of inverse; then we obtain
\beann
\sum_{\lambda\leq\sigma}(F_{\lambda\sigma,\nu}-g_{\lambda\sigma,\nu}) U^{\alpha\beta,\mu\nu,\lambda\sigma}V_{\alpha\beta\mu, abc}&=&\frac32(F_{\lambda\sigma,\nu}-g_{\lambda\sigma,\nu})(\delta^\lambda_a\delta^\sigma_b\delta^\nu_c+\delta^\lambda_b\delta^\sigma_a\delta^\nu_c)
\\
&=&3(F_{ab,c}-g_{ab,c})=0\  .
\eeann
Therefore, $F_{\lambda\sigma,\nu}=g_{\lambda\sigma,\nu}$
and the holonomy condition is recovered. 
Using this condition, equation \eqref{eq:hvf2} becomes:
\beq
\label{eq:sol3}
\frac{\partial H_{\mathfrak V}}{\partial g_{\alpha\beta}}+\sum_{\lambda\leq\sigma}g_{\lambda\sigma,\mu}\left(\frac{\partial L^{\alpha\beta,\mu}}{\partial g_{\lambda\sigma}}-\frac{\partial L^{\lambda\sigma,\mu}}{\partial g_{\alpha\beta}}\right)-\sum_{\lambda\leq\sigma}F_{\lambda\sigma;\mu,\nu}U^{\lambda\sigma,\mu\nu,\alpha\beta}=0\ .
\eeq
These equations have as particular solution $F_{\lambda\sigma;\mu,\nu}^P=\frac12g_{\alpha\beta}(\Gamma_{\nu \lambda }^\alpha\Gamma_{\mu \sigma}^\beta+\Gamma_{\nu \sigma}^\alpha\Gamma_{\mu \lambda }^\beta)$, which can be checked after some computation. The multivector field
$$
{\bf X}_h^P=\bigwedge_{\nu=0}^3 X_\nu^P=\bigwedge_{\nu=0}^3\left(\frac {\partial}{\partial x^\nu}+\sum_{\alpha\leq\beta}\left(g_{\alpha\beta,\nu}\frac{\partial}{\partial g_{\alpha\beta}}+\frac12g_{\lambda\sigma}(\Gamma_{\nu \alpha }^\lambda\Gamma_{\mu \beta}^\sigma+\Gamma_{\nu \beta}^\lambda\Gamma_{\mu \alpha }^\sigma)\frac{\partial}{\partial g_{\alpha\beta;\mu}}\right)\right)\ ,
$$
is holonomic and $\bar\pi_{\cal P}$-transverse, and verifies that
$\inn({\bf X}_h^P)\Omega_h=0$. 
The last thing to check is that it is integrable. 
The Lie bracket for two arbitrary components 
$X_\gamma^P$ and $X_\rho^P$ is
\beann
[{X_\gamma^P},{X_\rho^P}]&=&
\sum_{\alpha\leq\beta}\left(F_{\alpha\beta;\rho,\gamma}^P-F_{\alpha\beta;\gamma,\rho}^P\right)\frac{\partial}{\partial 
g_{\alpha\beta}}+\\ \nonumber
&&\sum_{\substack{\alpha\leq\beta\\\lambda\leq\sigma}}\left( g_{\lambda\sigma,\gamma}\frac{\partial F_{\alpha\beta;\mu,\rho}^P}{\partial g_{\lambda\sigma}}+
F_{\lambda\sigma;\nu,\gamma}^P\frac{\partial F_{\alpha\beta;\mu,\rho}^P}{\partial g_{\lambda\sigma,\nu}}-
g_{\lambda\sigma,\rho}\frac{\partial F_{\alpha\beta;\mu,\gamma}^P}{\partial g_{\lambda\sigma}}-
F_{\lambda\sigma;\nu,\rho}^P\frac{\partial F_{\alpha\beta;\mu,\gamma}^P}{\partial g_{\lambda\sigma,\nu}}
\right)\frac{\partial}{\partial g_{\alpha\beta,\mu}} \ .
\eeann
The vector field $[X_\gamma^P,X_\rho^P]$ is $\overline{\pi_1}$-vertical. Therefore, the integrability condition can only be achieved if 
$[X_\gamma^P,X_\rho^P]=0$. Imposing the condition on the coefficient of 
$\displaystyle\frac{\partial}{\partial g_{\alpha\beta}}$, we obtain that 
$F_{\alpha\beta;\rho,\gamma}^P-F_{\alpha\beta;\gamma,\rho}^P=0$.
These conditions are expected since, for a section, 
they represent the equality between second order crossed partial derivatives. 
Clearly the solution proposed fulfils this condition. 
After a rather long but straightforward computation, 
we can check that the coefficients of 
$\displaystyle\frac{\partial}{\partial g_{\alpha\beta,\mu}}$ also vanish.

\subsection{General solution (without energy-matter sources)}

The existence of a particular solution ${\bf X}_h^P$ to \eqref{eq:hvf2}
is relevant, because it implies that no extra restrictions are needed, 
as showed in Theorem \ref{theo:submanifold}. 
Now, we explore the general behaviour of the solutions to \eqref{eq:hvf2}.

As we have shown before, \eqref{eq:hvf2} boils down to \eqref{eq:sol3},
which are linear equations. Therefore, we can split any solution into 
a particular and a homogeneous part: 
$$
F_{\lambda\sigma;\mu,\nu}=\frac12g_{\alpha\beta}(\Gamma_{\nu \lambda }^\alpha\Gamma_{\mu \sigma}^\beta+\Gamma_{\nu \sigma}^\alpha\Gamma_{\mu \lambda }^\beta)+F^{\mathfrak h}_{\lambda\sigma;\mu,\nu} \ .
$$
The homogeneous part $F^{\mathfrak h}_{\lambda\sigma;\mu,\nu}$ is a set of functions 
which cancel out when contracted with \eqref{eq:U}, namely
\beq
\label{eq:homogenia}
\sum_{\lambda\leq\sigma}F^{\mathfrak h}_{\lambda\sigma;\mu,\nu}U^{\lambda\sigma,\mu\nu,\alpha\beta}=0 \ .
\eeq
The corresponding multivector field:
$$
{\bf X}_h=\bigwedge_{\nu=0}^3\left(\frac {\partial}{\partial x^\nu}+\sum_{\alpha\leq\beta}\left(F_{\alpha\beta,\nu}\frac{\partial}{\partial g_{\alpha\beta}}+\left(\frac12g_{\lambda\sigma}(\Gamma_{\nu \alpha }^\lambda\Gamma_{\mu \beta}^\sigma+\Gamma_{\nu \beta}^\lambda\Gamma_{\mu \alpha }^\sigma)+F^{\mathfrak h}_{\alpha\beta;\mu,\nu}\right)\frac{\partial}{\partial g_{\alpha\beta;\mu}}\right)\right)
$$
is a semiholonomic solution to \eqref{eq:hvf2}. 
Nevertheless, it may not be integrable. Thus, the integrability of 
${\bf X}_h$ leads to new constraints on the valid set of functions 
$T_{\alpha\beta;\mu,\nu}$. Condition \eqref{eq:homogenia} can be reformulated as follows:

\begin{lemma}
A set of functions $F^{\mathfrak h}_{\alpha\beta;\mu,\nu}$, symmetric under
the changes $\alpha\leftrightarrow\beta$ and $\mu\leftrightarrow\nu$,  
satisfies the condition
\beq\label{eq:lem1}
\sum_{\lambda\leq\sigma}F^{\mathfrak h}_{\lambda\sigma;\mu,\nu}U^{\lambda\sigma,\mu\nu,\alpha\beta}=0
\eeq
if, and only if,
\beq\label{eq:lem2}
g^{\lambda\sigma}\left(F^{\mathfrak h}_{\eta\tau;\lambda,\sigma}+
F^{\mathfrak h}_{\lambda\sigma;\eta,\tau}- F^{\mathfrak h}_{\lambda\eta;\tau,\sigma}-
F^{\mathfrak h}_{\lambda\tau;\eta,\sigma}\right)=0 \ .
\eeq
\end{lemma}
\begin{proof}
\eqref{eq:lem1} can be rewritten as
\bea\nonumber
\sum_{\lambda\leq\sigma}F^{\mathfrak h}_{\lambda\sigma;\mu,\nu}U^{\lambda\sigma,\mu\nu,\alpha\beta}\!&=&\!\varrho\, n(\alpha\beta)g^{\alpha\beta}g^{\lambda\sigma}g^{\mu\nu}\left(-\frac12F^{\mathfrak h}_{\lambda\sigma;\mu,\nu}+\frac12F^{\mathfrak h}_{\lambda\mu;\nu,\sigma}\right)\\\label{eq:lem3}
\!&+& \!\varrho\, n(\alpha\beta)g^{\lambda\sigma}g^{\alpha\mu}g^{\nu\beta}\left(\frac12F^{\mathfrak h}_{\mu\nu;\lambda,\sigma}+\frac12F^{\mathfrak h}_{\lambda\sigma;\mu,\nu}-\frac12F^{\mathfrak h}_{\lambda\mu;\nu,\sigma}-\frac12F^{\mathfrak h}_{\lambda\nu;\mu,\sigma}\!\right)
\eea
Contracting \eqref{eq:lem1} with $g_{\alpha\beta}$, we obtain
\beq
\label{eq:lem4}
2\varrho\, n(\alpha\beta)g^{\lambda\sigma}g^{\mu\nu}\left(-\frac{1}{2}F^{\mathfrak h}_{\lambda\sigma;\mu,\nu}+\frac{1}{2}F^{\mathfrak h}_{\lambda\mu;\nu,\sigma}\right)=0 \ .
\eeq
Therefore the first term in \eqref{eq:lem3} vanishes. 
Contracting the remaining term with $g_{\alpha\eta}g_{\beta\tau}$ 
we obtain \eqref{eq:lem2}. 

To prove the converse, contract \eqref{eq:lem2} with $g^{\eta\tau}$. 
The resulting expression is equivalent to \eqref{eq:lem4} because it is 
symmetric under the change $(\alpha\beta)\leftrightarrow(\eta\tau)$. 
Then, \eqref{eq:lem1} follows straighforwardly.
\end{proof}

The following theorem summarizes the above results:

\begin{theorem} \label{theo:SolVacFie}
For a class of multivectorfield $\{{\bf X}\}\subset{\mathfrak X}^4({\cal P})$,
the following conditions are equivalent:
\begin{itemize}
\item{} $\{{\bf X}\}$ is a solution to the Hamiltonian problem for the
system $({\cal P},\Omega_{h_{\mathfrak V}})$, namely, they are holonomic
multivector fields and satisfy the field equation
$$
\inn({\bf X})\Omega_{h_{\mathfrak V}}=0 \quad,\quad
\mbox{\rm for every ${\bf X}\in \{{\bf X}\}$} \ .
$$
\item{}
Using the coordinates $(x^\mu,g_{\alpha\beta},g_{\alpha\beta,\mu})$, 
the local expression of a represesentative of $\{{\bf X}\}$ is 
$$
{\bf X}=\bigwedge_{\nu=0}^3\left(\frac {\partial}{\partial x^\nu}+\sum_{\alpha\leq\beta}\left(g_{\alpha\beta,\nu}\frac{\partial}{\partial g_{\alpha\beta}}+\left(F^P_{\alpha\beta;\mu,\nu}+F^{\mathfrak h}_{\alpha\beta;\mu,\nu}\right)\frac{\partial}{\partial g_{\alpha\beta;\mu}}\right)\right)\ ,
$$
where 
$F^P_{\alpha\beta;\mu,\nu}=
\frac12g_{\lambda\sigma}(\Gamma_{\nu \alpha }^\lambda\Gamma_{\mu \beta}^\sigma+\Gamma_{\nu \beta}^\lambda\Gamma_{\mu \alpha }^\sigma)$
 and $F^{\mathfrak h}_{\alpha\beta;\mu,\nu}$ satisfy that:
\be
\item{$F^{\mathfrak h}_{\alpha\beta;\mu,\nu}=F^{\mathfrak h}_{\beta\alpha;\mu,\nu}=F^{\mathfrak h}_{\alpha\beta;\nu,\mu}$.}
\item{$
g^{\alpha\beta}\left(F^{\mathfrak h}_{\eta\tau;\alpha,\beta}+F^{\mathfrak h}_{\alpha\beta;\eta,\tau}-F^{\mathfrak h}_{\alpha\eta;\tau,\beta}-F^{\mathfrak h}_{\alpha\tau;\eta,\beta}\right)=0$.}
\item{ It is a solution to the following  differential equations
(integrability condition)}:
\ee
\end{itemize}
\vspace{-10mm}
\bea\nonumber
0&=&\sum_{\alpha\leq\beta}\left(F^{\mathfrak h}_{\alpha\beta;\mu,i}\frac{\partial F^{\mathfrak h}_{\lambda\sigma;\nu,j}}{\partial g_{\alpha\beta,\mu}}+\left(F_{\alpha\beta;\mu,i}^P\frac{\partial}{\partial g_{\alpha\beta,\mu}}+g_{\alpha\beta,i}\frac{\partial}{\partial g_{\alpha\beta}}+\frac{\partial}{\partial x^i}\right)F^{\mathfrak h}_{\lambda\sigma;\nu,j}+ F^{\mathfrak h}_{\alpha\beta;\mu,i}\frac{\partial F_{\lambda\sigma;\nu,j}^P}{\partial g_{\alpha\beta,\mu}}\right)\\\nonumber
&-&\sum_{\alpha\leq\beta}\left(F^{\mathfrak h}_{\alpha\beta;\mu,j}\frac{\partial F^{\mathfrak h}_{\lambda\sigma;\nu,i}}{\partial g_{\alpha\beta,\mu}}+\left(F_{\alpha\beta;\mu,j}^P\frac{\partial}{\partial g_{\alpha\beta,\mu}}+g_{\alpha\beta,j}\frac{\partial}{\partial g_{\alpha\beta}}+\frac{\partial}{\partial x^j}\right)F^{\mathfrak h}_{\lambda\sigma;\nu,i}+ F^{\mathfrak h}_{\alpha\beta;\mu,j}\frac{\partial F_{\lambda\sigma;\nu,i}^P}{\partial g_{\alpha\beta,\mu}}\right) \ .
\eea
\end{theorem}

The equivalent theorem for sections is:

\begin{theorem}\label{theo:SolVacSec} 
For a holonomic section $\psi:M\rightarrow \mathcal{P}$, the following conditions are equivalent:
\begin{enumerate}
\item{} $\psi$  is a solution to the Hamiltonian problem for the
system  $( {\cal P},\Omega_{h_{\mathfrak V}})$; namely it satisfies the field equation:
$$
\psi^*\inn({X})\Omega_{h_{\mathfrak V}}=0\quad ,\quad
\mbox{\rm for every $X\in \mathfrak{X}(\mathcal{P})$},
$$
\item{} $\psi$ is a solution to the vacuum Einstein's equations
$$\left.\left(R^{\alpha\beta}-
\frac{1}{2}g^{\alpha\beta}R\right)\right|_\psi=0, \quad\alpha,\beta=0,\dots,3.$$
\item{} $\psi$ is a solution to the differential equations
$$
\frac{\partial^2\psi_{\alpha\beta}}{\partial x^\mu\partial x^\nu}=\left.\left(F^{\mathfrak h}_{\alpha\beta;\mu,\nu}+\frac12g_{\lambda\sigma}(\Gamma_{\nu \alpha }^\lambda\Gamma_{\mu \beta}^\sigma+\Gamma_{\nu \beta}^\lambda\Gamma_{\mu \alpha }^\sigma)\right)\right|_\psi \ ,
$$
for some set of functions $F^{\mathfrak h}_{\alpha\beta;\mu,\nu}$ such that
$$
g^{\alpha\beta}\left(F^{\mathfrak h}_{\eta\tau;\alpha,\beta}+F^{\mathfrak h}_{\alpha\beta;\eta,\tau}- F^{\mathfrak h}_{\alpha\eta;\tau,\beta}-F^{\mathfrak h}_{\alpha\tau;\eta,\beta}\right)=0\ ,
$$
with initial conditions 
$\psi_{\alpha\beta}({x'}^\mu)=g'_{\alpha\beta}$, 
$\displaystyle\frac{\partial \psi_{\alpha\beta}}{\partial x^\mu}({x'}^\mu)=g'_{\alpha\beta,\mu}$.
\end{enumerate}
\end{theorem}
\begin{proof}
The equivalence $1\Longleftrightarrow 2$ is clear. 
The implication $1\Rightarrow 3$ comes from Theorem \ref{theo:SolVacFie}. 
To show $3\Rightarrow 2$, we first compute $R_{\alpha\beta}|_\psi$:
\beann{}
R_{\mu\eta}|_\psi&=&g^{\nu\lambda}R_{\lambda\mu,\nu\eta}|_\psi\\
&=& \left.-\frac{1}{2}g^{\nu\lambda}\left[\frac{\partial^2\psi_{\lambda\nu}}{\partial x^\mu\partial x^\eta}-\frac{\partial^2\psi_{\mu\nu}}{\partial x^\eta\partial x^\lambda}-\frac{\partial^2\psi_{\lambda\eta}}{\partial x^\mu\partial x^\nu}+\frac{\partial^2\psi_{\mu\eta}}{\partial x^\nu\partial x^\lambda}\right]\right|_\psi+\left.g^{\nu\lambda}g_{\tau\sigma}(\Gamma^\tau_{\eta\lambda}\Gamma^\sigma_{\mu\nu}-\Gamma^\tau_{\nu\lambda}\Gamma^\sigma_{\mu\eta})\right|_\psi
\\
&=&\left.-\frac12g^{\nu\lambda}\left(F^{\mathfrak h}_{\eta\mu;\nu,\lambda}+F^{\mathfrak h}_{\nu\lambda;\eta,\mu}- F^{\mathfrak h}_{\nu\eta;\mu,\lambda}-F^{\mathfrak h}_{\nu\mu;\eta,\lambda}\right)\right|_\psi=0 \ .
\eeann{}
Then
$$\left.\left(R^{\alpha\beta}-
\frac{1}{2}g^{\alpha\beta}R\right)\right|_\psi=\left.\left(g^{\alpha\mu}g^{\beta\eta}-\frac12g^{\alpha\beta}g^{\mu\eta}\right)R_{\mu\eta}\right|_\psi=0 \ .
$$
\end{proof}

These theorems characterize the solutions to Einstein's equations without sources. 
The multivector fields solution to \eqref{theo:SolVacFie} 
are described by the set of functions $F^{\mathfrak h}_{\alpha\beta;\mu,\nu}$
which have some combinatoric properties. 
The integral sections of an integrable multivector field are given by \eqref{pdesect2}. 
Every multivector field has one section at every point, therefore, 
only an initial condition is required to solve these equations. 
The condition 3 in Theorem \ref{theo:SolVacFie} is the integrability condition. 
If a multivector field is not integrable, we can still consider \eqref{pdesect2}, 
but we will find out that such equations have no solution everywhere. 
Thus, the integrability condition is also the condition of existence of solutions 
to \eqref{pdesect2}. 
Given an initial condition, there is several section 
solution to the equations: one for every multivector field. 
Nevertheless, two different multivector fields may lead to the same sections 
at a given point. These multiple solution are not gauge related, because 
the multisymplectic form is regular.

\subsection{General solution (with energy-matter sources)}
\label{sect:SolMatter}

\begin{theorem} 
Consider an energy-matter term $L_{\mathfrak m}$ with degree  $\leq 1$,
and the system $({\cal P}_{\mathfrak S},\Omega_{h_{\mathfrak S}})$. 
For a class of multivector field 
$\{{\bf X}\}\subset{\mathfrak X}^4({\cal P}_{\mathfrak S})$, 
the following conditions are equivalent:
\begin{itemize}
\item $\{{\bf X}\}$ is a class of semiholonomic multivector fields solution to the  equation
$$
\inn({\bf X})\Omega_{h_{\mathfrak S}}=0\quad ,\quad\mbox{\rm for every ${\bf X}\in \{{\bf X}\}$} \ .
$$
\item The local expression of a representative ${\bf X}\in \{{\bf X}\}$ is
$$
{\bf{X}}=
\bigwedge_{\nu=0}^3 \sum_{\alpha\leq\beta}\left(\frac{\partial}{\partial x^\nu}+ 
g_{\alpha\beta,\nu}\frac{\partial}{\partial g_{\alpha\beta}}+
{F}_{\alpha\beta;\mu,\nu}\frac{\partial}{\partial g_{\alpha\beta,\mu}}\right)
$$
with
$$
F_{\lambda\sigma;\mu,\nu}=\frac12g_{\lambda\sigma}(\Gamma_{\nu \alpha }^\lambda\Gamma_{\mu \beta}^\sigma+\Gamma_{\nu \beta}^\lambda\Gamma_{\mu \alpha }^\sigma)+g_{\lambda\sigma}\left(g_{\alpha\beta}g_{\mu\nu}-\frac13g_{\alpha\mu}g_{\beta\nu}\right)\frac{{L_{\mathfrak m}}^{\alpha\beta}}{\varrho\,n(\alpha\beta)}+{F^{\mathfrak h}}_{\lambda\sigma;\mu,\nu}.
$$
 and where ${F^h}_{\alpha\beta;\mu,\nu}$ satisfies:
\be
\item{${F^{\mathfrak h}}_{\alpha\beta;\mu,\nu}={F^{\mathfrak h}}_{\beta\alpha;\mu,\nu}={F^{\mathfrak h}}_{\alpha\beta;\nu,\mu}$.}
\item{$
g^{\alpha\beta}\left({F^{\mathfrak h}}_{\eta\tau;\alpha,\beta}+{F^{\mathfrak h}}_{\alpha\beta;\eta,\tau}- 
{F^{\mathfrak h}}_{\alpha\eta;\tau,\beta}-{F^{\mathfrak h}}_{\alpha\tau;\eta,\beta}\right)=0$.}
\ee
\end{itemize}
\end{theorem}
\begin{proof}
The local expression of the equations is
$$
\frac{\partial H_{\mathfrak V}}{\partial g_{\alpha\beta}}+
\sum_{\lambda\leq\sigma}g_{\lambda\sigma,\mu}\left(\frac{\partial L_{\mathfrak V}^{\alpha\beta,\mu}}{\partial g_{\lambda\sigma}}-
\frac{\partial L_{\mathfrak V}^{\lambda\sigma,\mu}}{\partial g_{\alpha\beta}}\right)-
\sum_{\lambda\leq\sigma}F_{\lambda\sigma;\mu,\nu}U^{\lambda\sigma,\mu\nu,\alpha\beta}=
-L_{\mathfrak m}^{\alpha\beta}\ .
$$
Then we split the unknown functions in three parts:
$$
F_{\lambda\sigma;\mu,\nu}={F^R}_{\lambda\sigma;\mu,\nu}+{F^\mathfrak{m}}_{\lambda\sigma;\mu,\nu}+{F^{\mathfrak h}}_{\lambda\sigma;\mu,\nu} \ .
$$
This first term is a solution to the equations at vacuum:
$$
\frac{\partial H_{\mathfrak V}}{\partial g_{\alpha\beta}}+
\sum_{\lambda\leq\sigma}g_{\lambda\sigma,\mu}\left(\frac{\partial L_{\mathfrak V}^{\alpha\beta,\mu}}{\partial g_{\lambda\sigma}}-\frac{\partial L_{\mathfrak V}^{\lambda\sigma,\mu}}{\partial g_{\alpha\beta}}\right)=
\sum_{\lambda\leq\sigma}{F^R}_{\lambda\sigma;\mu,\nu}U^{\lambda\sigma,\mu\nu,\alpha\beta}\ .
$$
As we have seen before, we can choose 
${F^R}_{\lambda\sigma;\mu,\nu}=
\frac12g_{\alpha\beta}(\Gamma_{\nu \lambda }^\alpha\Gamma_{\mu \sigma}^\beta+\Gamma_{\nu \sigma}^\alpha\Gamma_{\mu \lambda }^\beta)$. 
The second term is a solution to
$$
\sum_{\lambda\leq\sigma}{F^\mathfrak{m}}_{\lambda\sigma;\mu,\nu}U^{\lambda\sigma,\mu\nu,\alpha\beta}={L_{\mathfrak m}}^{\alpha\beta} \ .
$$
We can choose 
${F^\mathfrak{m}}_{\lambda\sigma;\mu,\nu}=\frac{1}{\varrho\,n(\tau\gamma)}
g_{\lambda\sigma}\left(g_{\tau\mu}g_{\gamma\nu}-
\frac13g_{\tau\gamma}g_{\mu\nu}\right){L_{\mathfrak m}}^{\tau\gamma}$, 
which belongs to $C^\infty(J^1\pi)$ because $\deg(L_{\mathfrak m})\leq1$ . Indeed,
\beann
\sum_{\lambda\leq\sigma}\frac{1}{\varrho\,n(\tau\gamma)}g_{\lambda\sigma}\left(g_{\tau\mu}g_{\gamma\nu}-\frac13g_{\tau\gamma}g_{\mu\nu}\right){L_{\mathfrak m}}^{\tau\gamma}U^{\lambda\sigma,\mu\nu,\alpha\beta}=\\
\frac{n(\alpha\beta)}{n(\tau\gamma)}\left(g_{\tau\mu}g_{\gamma\nu}-\frac13g_{\tau\gamma}g_{\mu\nu}\right){L_{\mathfrak m}}^{\tau\gamma}\left(\frac12g^{\alpha\mu}g^{\beta\nu}+\frac12g^{\alpha\nu}g^{\beta\mu}-g^{\alpha\beta}g^{\mu\nu}\right)=
\\
\frac{n(\alpha\beta)}{2\,n(\tau\gamma)}(\delta^\alpha_\tau\delta^\beta_\gamma+\delta^\beta_\tau\delta^\alpha_\gamma){L_{\mathfrak m}}^{\tau\gamma}=\frac12({L_{\mathfrak m}}^{\alpha\beta}+{L_{\mathfrak m}}^{\beta\alpha})={L_{\mathfrak m}}^{\alpha\beta} \ .
\eeann
Finally, the third term is solution to the homogeneous equation
$$
\sum_{\lambda\leq\sigma}{F^{\mathfrak h}}_{\lambda\sigma;\mu,\nu}U^{\lambda\sigma,\mu\nu,\alpha\beta}=0
$$
For \eqref{eq:lem2}, this equation is equivalent to the statement. 
Notice that any other $F^R$ or $F^\mathfrak{m}$ can be obtained from these ones 
by adding a suitable function of the type $F^{\mathfrak h}$.
\end{proof}

It is important to remark that the solution given by this theorem 
may not be integrable. But any integrable solution follows this structure. The corresponding result for sections is:

\begin{theorem}\label{theo:SolMatSec} 
For a holonomic section  $\psi:M\rightarrow \mathcal{P}_{\mathfrak S}$,
the following conditions are equivalent:
\begin{enumerate}
\item{} $\psi$  is a solution to the Hamiltonian problem
for the system $( {\cal P},\Omega_{h_{\mathfrak S}})$, namely it  satisfy 
to the field equation
$$
\psi^*\inn({X})\Omega_{h_{\mathfrak S}}=0\quad ,\quad
\mbox{\rm for every $X\in \vf({\cal P}_{\mathfrak S})$} \ .
$$
\item{} $\psi$ is solution to the Einstein equations.
$$
\left.\left(R^{\alpha\beta}-
\frac{1}{2}g^{\alpha\beta}R\right)\right|_\psi=-\frac{1}{\varrho n(\alpha\beta)}L_{\mathfrak m}^{\alpha\beta}|_\psi \ .
$$

\item{} $\psi$ is solution to the differential equations
$$
\frac{\partial^2\psi_{\alpha\beta}}{\partial x^\mu\partial x^\nu}=\left.\left(F^{\mathfrak h}_{\alpha\beta;\mu,\nu}+\frac12g_{\lambda\sigma}(\Gamma_{\nu \alpha }^\lambda\Gamma_{\mu \beta}^\sigma+\Gamma_{\nu \beta}^\lambda\Gamma_{\mu \alpha }^\sigma)+g_{\alpha\beta}\left(g_{\tau\mu}g_{\gamma\nu}-\frac13g_{\tau\gamma}g_{\mu\nu}\right)\frac{{L_{\mathfrak m}}^{\tau\gamma}}{\varrho\,n(\tau\gamma)}\right)\right|_\psi \ ,
$$
for some set of functions $F^{\mathfrak h}_{\alpha\beta;\mu,\nu}$ such that
$$
g^{\alpha\beta}\left(F^{\mathfrak h}_{\eta\tau;\alpha,\beta}+F^{\mathfrak h}_{\alpha\beta;\eta,\tau}- F^{\mathfrak h}_{\alpha\eta;\tau,\beta}-F^{\mathfrak h}_{\alpha\tau;\eta,\beta}\right)=0\ ,
$$
and with initial conditions $\psi_{\alpha\beta}({x'}^\mu)=g'_{\alpha\beta}$, 
$\displaystyle\frac{\partial \psi_{\alpha\beta}}{\partial x^\mu}({x'}^\mu)=
g'_{\alpha\beta,\mu}$.
\end{enumerate}
\end{theorem}

\section{Multivector fields}
\label{mvf}

(See \cite{art:Echeverria_Munoz_Roman98} for details).

\begin{definition}
Let $\kappa\colon{\cal M}\to M$ be a fiber bundle.

An {\rm $m$-multivector field} in ${\cal M}$ is a skew-symmetric contravariant 
tensor of order $m$ in ${\cal M}$. The set of $m$-multivector fields 
in ${\cal M}$ is denoted $\vf^m({\cal M})$.

A multivector field $\mathbf{X}\in\vf^m({\cal M})$ is said to be {\rm locally decomposable} if,
for every $p\in {\cal M}$, there is an open neighbourhood  $U_p\subset{\cal M}$
and $X_1,\ldots ,X_m\in\vf (U_p)$ such that $\mathbf{X}\vert_{U_p}=X_1\wedge\ldots\wedge X_m$.

Locally decomposable $m$-multivector fields $\mathbf{X}\in\vf^m({\cal M})$ are locally associated with $m$-dimensional
distributions $D\subset\Tan{\cal M}$, and multivector fields associated with
the same distribution make an {\rm equivalence class} $\{ {\bf X}\}$ in the set $\vf^m({\cal M})$.
Then,
$\mathbf{X}$ is {\rm integrable} if its associated distribution is integrable. 

 A multivector field $\mathbf{X}\in\mathfrak{X}^m({\cal M})$ is 
 \textsl{$\kappa$-transverse} if, for every $\beta\in\Omega^m(M)$ with
$\beta (\kappa(p))\not= 0$, at every point
$p\in{\cal M}$, we have that  $(\inn(\mathbf{X})(\kappa^*\beta))_p\not= 0$. 
If $\mathbf{X}\in\mathfrak{X}^m({\cal M})$ is
integrable, then it is       
$\kappa$-transverse if, and only if, its integral manifolds are local sections of $\kappa$.
In this case, if $\psi\colon U\subset M\to{\cal M}$ is a local
section and $\psi (U)$ is the integral manifold 
of $\mathbf{X}$ at $p$, then  $T_p({\rm Im}\,\psi) = \mathcal{D}_p(\mathbf{X})$
and $\psi$ is an {\rm integral section} of ${\bf X}$.
\end{definition}

For every $\mathbf{X}\in\mathfrak{X}^m({\cal M})$, 
there exist $X_1,\ldots ,X_r\in\mathfrak{X} (U)$ such that
$$
\mathbf{X}\vert_{U}=\sum_{1\leq i_1<\ldots <i_m\leq r} f^{i_1\ldots i_m}X_{i_1}\wedge\ldots\wedge X_{i_m} \, ,
$$
with $f^{i_1\ldots i_m} \in C^\infty (U)$, $m \leqslant r\leqslant{\rm dim}\,{\cal M}$.
Then, the condition of ${\bf X}$ to be integrable is locally equivalent to
$[X_i,X_j]=0$, for every $i,j=1,\ldots,m$.
If two multivector fields ${\bf X},{\bf X}'$ belong to the same
equivalence class $\{ {\bf X}\}$ then, for every $U\subset {\cal M}$,
there exists a non-vanishing function $f\in\Cinfty(U)$ such that ${\bf X}'=f{\bf X}$ on $U$.

\begin{definition}
If $\Omega\in\df^k({\cal M})$ and $\mathbf{X}\in\mathfrak{X}^m({\cal M})$,
the {\rm contraction} between ${\bf X}$ and $\Omega$ is defined as
the natural contraction between tensor fields; in particular,
$$
 \inn({\bf X})\Omega\mid_{U}:= \sum_{1\leq i_1<\ldots <i_m\leq
 r}f^{i_1\ldots i_m} \inn(X_1\wedge\ldots\wedge X_m)\Omega 
=
 \sum_{1\leq i_1<\ldots <i_m\leq r}f^{i_1\ldots i_m} \inn
 (X_1)\ldots\inn (X_m)\Omega \ ,
$$
if $k\geq m$, and equal to zero if $k<m$.
The {\rm Lie derivative} of $\Omega$ with respect to ${\bf X}$ is defined as the graded bracket
( it is an operation of degree $m-1$)
 $$
\Lie({\bf X})\Omega:=[\d , \inn ({\bf X})]\Omega=(\d\inn ({\bf X})-(-1)^m\inn ({\bf X})\d)\Omega \ .
 $$
\end{definition}

\begin{definition}
In the case that ${\cal M}=J^k\pi$, a multivector field $\mathbf{X}\in\mathfrak{X}^m(J^k\pi)$
is said to be {\rm holonomic} if it is integrable and 
its integral sections are holonomic sections of $\bar\pi^k$
(and hence it is locally decomposable and $\bar\pi^k$-transverse).
\end{definition}

For a fiber manifold $\tau:E\rightarrow M$ with coordinates $(x^i,y^\alpha)$, a $\tau$-transverse and locally decomposable multivector field ${\bf X} \in \mathfrak{X}^m(E)$ is
$$
\mathbf{X} =\bigwedge_{i=1}^m\left(\derpar{}{x^i}+X_i^\alpha\derpar{}{y^\alpha}\right) \ .
$$
A section of $\tau$, $\psi(x^i) = (x^i,\,\psi^{\alpha}(x^i))$,
is an integral section of ${\bf X}$ if its component functions
satisfy the following system of partial differential equations
\beq\label{pdesect2}
\derpar{\psi^{\alpha}}{x^i}=X_i^\alpha\circ\psi \ .
\eeq


\section*{Acknowledgments}

We acknowledge the financial support of the 
{\sl Ministerio de Ciencia e Innovaci\'on} (Spain), projects
MTM2014--54855--P and
MTM2015-69124--REDT, and of
{\sl Generalitat de Catalunya}, project 2014-SGR-634.
We also want to thank the referee for his valuable suggestions and comments.



\begin{thebibliography}{9}

{\small

\bibitem{art:Aldaya_Azcarraga78_2}
V. Aldaya, J.A. de Azcarraga,
``Variational Principles on $r\,th$ order jets of fibre bundles in Field Theory'',
\textsl{J. Math. Phys.} \textbf{19}(9) (1978), 1869--1875.
(doi: 10.1063/1.523904).

\bibitem{art:Capriotti}
S. Capriotti,
``Differential geometry, Palatini gravity and reduction'',
{\sl J. Math. Phys.} {\bf 55}(1) (2014) 012902.
(doi: 10.1063/1.4862855).

\bibitem{art:Capriotti2}
S. Capriotti, 
``Unified formalism for Palatini gravity'',
{\sl Int. J. Geom. Meth. Mod. Phys.} {\bf 15}(3) (2018) 1850044.
(doi:10.1142/S0219887818500445).

\bibitem{first}
  M. Castrill\'on, J. Mu\~noz-Masqu\'e, M.E. Rosado,
 ``First-order equivalent to Einstein-Hilbert Lagrangian'',
{\sl J. Math. Phys.} {\bf 55}(8) (2014) 082501.
(doi: 10.1063/1.4890555).

\bibitem{CVB-2006}
R. Cianci, S. Vignolo, D. Bruno,
``General Relativity as a constrained Gauge Theory''
{\sl Int. J. Geom. Meth. Mod. Phys.} {\bf 3}(8) (2006) 1493-1500.
(doi: 10.1142/S0219887806001818).

\bibitem{CreTe-2016}
C. Cremaschini, M. Tessarotto,
``Manifest Covariant Hamiltonian Theory of General Relativity'',
{\sl App. Phys. Research} {\bf 8}(2) (2016) 60-81.
(doi: 10.5539/apr.v8n2p60).

\bibitem{CreTe-2016b}
C. Cremaschini, M. Tessarotto,
``Hamiltonian approach to GR - Part 1: covariant theory of classical gravity'',
{\sl Eur. Phys. J. C} (2017) {\bf 77}:329. (doi: 10.1140/epjc/s10052-017-4854-1).

\bibitem{pons} 
N. Dadhich, J.M. Pons,
``On the equivalence of the Einstein-Hilbert and the Einstein-Palatini formulations 
of general relativity for an arbitrary connection???, 
{\sl Gen. Rel. Grav.} {\bf 44}(9) (2012) 2337-2352.
(doi: 10.1007/s10714-012-1393-9).

\bibitem{LMMMR-2002}
M. de Le\'on, J. Mar\'\i n-Solano, J.C. Marrero, M.C. Mu\~noz-Lecanda, N. Rom\'an-Roy, 
``Singular Lagrangian systems on jet bundles'',
\textsl{Fortsch. Phys.} \textbf{50}(2) (2002) 105-169.
(doi: 10.1002/1521-3978(200203)50:2).

\bibitem{art:deLeon_Marin_Marrero_Munoz_Roman05}
M. de Le\'on, J. Mar\'\i n-Solano, J.C. Marrero, M.C. Mu\~noz-Lecanda, N. Rom\'an-Roy, 
``Premultisymplectic constraint algorithm for field theories'',
\textsl{Int. J. Geom. Meth. Mod. Phys.} \textbf{2}(5) (2005) 839-871.
(doi: 10.1142/S0219887805000880).

\bibitem{art:deLeon_etal2004}
M. de Le\'on, D. Mart\'in de Diego, A. Santamar\'ia-Merino,
``Symmetries in classical field theory'',
{\sl Int. J. Geom. Meths. Mod. Phys.} {\bf 1}(5) (2004) 651-710.
(doi: 10.1142/S0219887804000290).

\bibitem{EMR-96}
A. Echeverr\'\i a-Enr\'\i quez, M.C. Mu\~noz-Lecanda, N. Rom\'an-Roy,
``Geometry of Lagrangian first-order classical field theories''.
{\sl Forts. Phys.} {\bf 44} (1996) 235-280.
(doi: 10.1002/prop.2190440304).

\bibitem{art:Echeverria_Munoz_Roman98}
A. Echeverr\'\i a-Enr\'\i quez, M.C. Mu\~noz-Lecanda, N. Rom\'an-Roy,
``Multivector fields and connections: Setting Lagrangian equations in field theories'',
\textsl{J. Math. Phys.} \textbf{39}(9) (1998) 4578-4603.
(doi: 10.1063/1.532525).

\bibitem{art:Echeverria_Lopez_Marin_Munoz_Roman04}
A. Echeverr\'\i a-Enr\'\i quez, C. L\'opez, J. Mar\'\i n-Solano, M.C. Mu\~noz-Lecanda, N. Rom\'an-Roy,
``Lagrangian-Hamiltonian unified formalism for field theory'',
 {\sl J. Math. Phys.} \textbf{45}(1) (2004) 360-380.
(doi: 10.1063/1.1628384).

\bibitem{ESG-1995}
G. Esposito, C. Stornaiolo, G. Gionti,
``Spacetime Covariant Form of Ashtekar's Constraints''
{\sl Nuovo Cim.B} {\bf 110}(10) (1995) 1137-1152.
(doi: 10.1007/BF02724605).

\bibitem{FGR}
A. Fern\'andez, P.L. Garc\'ia, C. Rodrigo, 
``Stress-energy-momentum tensors in higher order variational calculus'', 
{\sl J. Geom. Phys.} {\bf 34}(1) (2000) 41-72.
(doi: 10.1016/S0393-0440(98)00063-1).

\bibitem{FoRo}
M. Forger, H. Romer, 
``Currents and the Energy-Momentum Tensor in Classical Field Theory: A Fresh Look at an Old Problem'',
{\sl  Annals Phys.} {\bf 309}(2) (2004) 306-389.
(doi: 10.1016/j.aop.2003.08.011).

\bibitem{proc:Garcia_Munoz83}
P.L. Garc\'\i a, J. Mu\~noz-Masqu\'e,
``On the geometrical structure of higher order variational calculus'',
\textsl{Procs. IUTAM-ISIMM Symposium on Modern Developments in Analytical Mechanics},
\textbf{Vol. I} (Torino, 1982). 
\textsl{Atti. Accad. Sci. Torino Cl. Sci. Fis. Math. Natur.} \textbf{117} (1983) suppl. 1, 127-147.

\bibitem{GPR-2017}
J. Gaset, P.D. Prieto-Mart\'inez, N. Rom\'an-Roy,
``Variational principles and symmetries on fibered multisymplectic manifolds'',
{\sl Comm. in Maths.} {\bf 24}(2) 137-152.
(doi: 10.1515/cm-2016-0010).

\bibitem{art:GR-2016}
J. Gaset, N. Rom\'an-Roy,
``Order reduction, projectability and constraints of second-order field theories and higher-order mechanics'', 
{\sl Rep. Math. Phys.} {\bf 78}(3) (2016) 327-337.
(doi: 10.1063/1.4940047).

\bibitem{GMS-97}
G. Giachetta, L. Mangiarotti, G. Sardanashvily,
{\sl New Lagrangian and Hamiltonian methods in field theory},
 World Scientif\/ic Publishing Co., Inc., River Edge, NJ, 1997.
(ISBN: 981-02-1587-8.).

\bibitem{GIMMSY-mm}
M.J. Gotay, J. Isenberg, J.E. Marsden, R. Montgomery,
``Momentum maps and classical relativistic fields. I.~Covariant theory'',
arXiv:physics/9801019 [math-ph] (2004).

\bibitem{GoMa}
M. J. Gotay, J.E. Marsden, 
``Stress-Energy-Momentum Tensors and the Belinfante-Rosenfeld Formula'', 
{\sl Contemp. Math.} {\bf 132} (1992) 367-392.

\bibitem{GPR-91}
X. Gr\`acia, J.M. Pons, N. Rom\'an-Roy,
``Higher order conditions for singular Lagrangian dynamics'',
{\sl J. Phys. A: Math. Gen.} {\bf 25}(7) (1992) 1989-2004. (doi: 10.1088/0305-4470/25/7/037).

\bibitem{Herman-2017}
J. Herman,
``Noether's theorem in multisymplectic geometry'', 
{\sl Diff. Geom. Appls.} {\bf 56}(2) (2018) 260-294.
(doi:10.1016/j.difgeo.2017.09.003).

\bibitem{IS-2016}
A, Ibort, A. Spivak,
``On a covariant Hamiltonian description of Palatini's gravity on manifolds with boundary'',
arXiv:1605.03492 [math-ph] (2016).

\bibitem{Ka-q1}
I.V. Kanatchikov,
``Precanonical quantum gravity: quantization without the space-time decomposition'',
{\sl Int. J. Theor. Phys.} {\bf 40} (2001), 1121--1149.
(doi:10.1023/A:1017557603606).

\bibitem{Ka-q2}
I.V. Kanatchikov,
``On precanonical quantization of gravity'',
{\it Nonlin. Phenom. Complex Sys.} (NPCS) {\bf 17} (2014) 372-376.

\bibitem{Ka-2016}
I.V. Kanatchikov,
``On the `spin connection foam' picture of quantum gravity from precanonical quantization'',
{\sl Procs. of the 14th Marcel Grossmann Meeting on General Relativity},
Univ. Roma ``La Sapienza'', Italy, 2015, (2017) 3907-3915.
(doi: 10.1142/9789813226609-0519).

\bibitem{art:Kouranbaeva_Shkoller00}
S. Kouranbaeva, S. Shkoller,
``A variational approach to second-order multisymplectic field theory'',
\textsl{J. Geom. Phys.} \textbf{35} (2000) 333--366.
(doi: 10.1016/S0393-0440(00)00012-7).

\bibitem{krupka2}
D. Krupka, ``Variational principles for energy-momentum tensors'',
{\sl Rep. Math. Phys.} {\bf 49}(2/3) (2002), 259-268.
(doi: 10.1016/S0034-4877(02)80024-6).

\bibitem{Krupka}
D. Krupka,
{\sl Introduction to Global Variational Geometry}, Atlantis Studies
in Variational Geometry, Atlantis Press. 2015.
(doi: 10.2991/978-94-6239-0737).

\bibitem{KrupkaStepanova}
D. Krupka, O. Stepankova,
``On the Hamilton form in second order calculus of variations'', 
{\sl Procs. Int. Meeting on Geometry and
Physics}, 85-101. Florence 1982, Pitagora, Bologna, 1983.

\bibitem{art:Munoz85}
J. Mu\~{n}oz-Masqu\'{e},
``Poincar\'{e}-Cartan forms in higher order variational calculus on fibred manifolds'',
{\sl Rev. Mat. Iberoamericana} \textbf{1} (1985), 85--126.
(doi: 10.4171/RMI/20).

\bibitem{pere}
  P.D. Prieto Mart\'inez, N. Rom\'an-Roy,
  ``A new multisymplectic unified formalism for second-order classical field theories'',
{\sl J. Geom. Mech.} {\bf 7}(2) (2015) 203-253.
(doi:	10.3934/jgm.2015.7.203).

\bibitem{pere2}
  P.D. Prieto Mart\'inez, N. Rom\'an-Roy,
  ``Variational principles for multisymplectic second-order classical field theories'',
{\sl Int. J. Geom. Meth. Mod. Phys.} {\bf 12}(8) (2015) 1560019.
(doi: 10.1515/cm-2016-0010).

\bibitem{art:Roman09}
N. Rom\'{a}n-Roy,
``Multisymplectic {Lagrangian} and {Hamiltonian} formalisms of classical field theories'',
\textsl{Symm. Integ. Geom. Methods Appl. (SIGMA)} \textbf{5} (2009) 100, 25pp.
(doi: 10.3842/SIGMA.2009.100).

\bibitem{rosado}
   M.E. Rosado, J. Mu\~noz-Masqu\'e,
``Integrability of second-order Lagrangians admitting a first-order Hamiltonian formalism'',
 {\sl Diff. Geom. and Apps.}
{\bf 35} (Sup. September 2014) (2014) 164-177.
(doi: 10.1016/j.difgeo.2014.04.006).

\bibitem{rosado2}
M.E. Rosado, J. Mu\~noz-Masqu\'e, ``Second-order Lagrangians admitting a first-order Hamiltonian formalism'', 
{\sl J. Annali di Matematica} {\bf 197}(2) (2018) 357-397. 
(doi: 10.1007/s10231-017-0683-y).

\bibitem{rovelli}
C. Rovelli,
``A note on the foundation of relativistic mechanics. II: Covariant Hamiltonian General Relativity'',
 in {\sl Topics in Mathematical Physics, General Relativity and Cosmology}, 
H. Garcia-Compean, B. Mielnik, M. Montesinos, M. Przanowski eds, 
 397, (World Scientific, Singapore) (2006).

\bibitem{Sd-95}
G. Sardanashvily,
{\sl Generalized Hamiltonian formalism for f\/ield theory. Constraint systems},
World Scientif\/ic Publishing Co., Inc., River Edge, NJ, 1995.
(ISBN: 981-02-2045-6).

\bibitem{book:Saunders89}
D.J. {Saunders}, \textsl{The geometry of jet bundles}, London Mathematical
  Society, Lecture notes series, vol. 142, Cambridge University Press,
  Cambridge, New York 1989.
(ISBN 13: 978-0521369480).

\bibitem{art:Saunders_Crampin90}
D.J. {Saunders} and M.~{Crampin}, ``On the {Legendre} map in higher-order field
  theories'', \textsl{J. Phys. A: Math. Gen.} \textbf{23}(14) (1990) 3169--3182.
(doi:10.1088/0305-4470/23/14/016).

\bibitem{art:Skinner_Rusk83}
R. Skinner, R. Rusk, 
``Generalized Hamiltonian dynamics. I. Formulation on $T^*Q\oplus TQ$'',
{\sl J. Math. Phys.} \textbf{24}(11) (1983) 2589-2594.
(doi: 10.1063/1.525654).
  
\bibitem{vey1}
D. Vey,
``Multisymplectic formulation of vielbein gravity. De Donder-Weyl formulation, Hamiltonian $(n-1)$-forms'',
{\sl Class. Quantum Grav.} {\bf 32}(9) (2015) 095005.
(doi: 10.1088/0264-9381/32/ 9/095005).

\bibitem{vey2}
D. Vey,
``10-plectic formulation of gravity and Cartan connections'',
Preprint hal-01408289 (2016).

\bibitem{voicu}
N. Voicu,
``Energy-momentum tensors in classical field theories. A modern perspective'',
{\sl Int. J. Geom. Meth. Mod. Phys.} {\bf 13}(8) (2016)  1640001 (20 pp).
(doi: 10.1142/S0219887816400016).

}

\end{thebibliography}
\end{document}